\documentclass[12pt]{article}
\pdfoutput=1
\usepackage{amsmath,amsthm,amsfonts}
\usepackage{fullpage}
\usepackage[skip=10pt plus1pt, indent=40pt]{parskip}
\usepackage{t1enc}
\usepackage{appendix}
\usepackage{enumerate}
\usepackage{epsfig}
\usepackage{makeidx}
\usepackage{amssymb}
\usepackage{amsfonts,amsmath,amsthm,amssymb,nccmath}
\usepackage[authoryear]{natbib}
\usepackage{mathabx,relsize}
\usepackage{graphics,graphicx,psfrag,epsf}
\usepackage{pdflscape}
\usepackage{afterpage}
\usepackage{capt-of}
\usepackage{lscape}
\usepackage{float}
\usepackage{caption}
\usepackage{subcaption}
\usepackage{comment}
\usepackage{lipsum}
\usepackage{ctable} 

\allowdisplaybreaks[3]

\usepackage{xcolor}
\definecolor{dbl}{rgb}{0.4,0.0,0.97}

\usepackage[pdftex,hypertexnames=false,linktocpage=true,colorlinks]{hyperref}
\hypersetup{colorlinks=true,linkcolor=dbl,anchorcolor=dbl,citecolor=dbl,filecolor=dbl,urlcolor=dbl,bookmarksnumbered=true}\RequirePackage{hypernat}

\usepackage{listings}
\usepackage{multirow}
\usepackage[utf8]{inputenc}
\usepackage[utf8]{luainputenc}
\usepackage[english]{babel}
\usepackage[top=1.5in,bottom=1.2in,left=1in,right=1in]{geometry}

\makeindex

\renewcommand\thesection{\arabic{section}}

\newcommand{\ncom}{\newcommand}
\ncom{\ul}{\underline}
\ncom{\m}{\mathbb{R}}
\ncom{\np}{N_P(\mu,\sGm)}
\ncom{\nno}{\nonumber}
\ncom{\non}{\nonumber}
\ncom{\ds}{\displaystyle}
\ncom{\half}{\frac{1}{2}}
\ncom{\mbx}{\makebox{.25cm}}
\ncom{\hs}{\mbox{\hspace{.25cm}}}
\ncom{\rar}{\rightarrow}
\ncom{\Rar}{\Rightarrow}
\ncom{\noin}{\noindent}
\ncom{\sz}{\scriptsize}
\ncom{\rf}{\ref}
\ncom{\s}{\sqrt{2}}
\ncom{\sgm}{\sigma}
\ncom{\sGm}{\Sigma}
\ncom{\Sgm}{\sigma^2}
\ncom{\psgm}{\sigma^{\prime}}
\ncom{\dt}{\delta}
\ncom{\Dt}{\Delta}
\ncom{\lmd}{\lambda}
\ncom{\Lmd}{\Lambda}
\ncom{\Th}{\theta}
\ncom{\e}{\eta}
\ncom{\Ch}{\chi^2}
\ncom{\pcc}{\stackrel{P}{>}}
\ncom{\lp}{\stackrel{L_{p}}{>}}
\ncom{\dist}{{\rm\,dist}}
\ncom{\sspan}{{\rm\,span}}
\ncom{\re}{{\rm Re\,}}
\ncom{\im}{{\rm Im\,}}
\ncom{\sgn}{{\rm sgn\,}}
\ncom{\hone}{\mbox{\hspace{1em}}}
\ncom{\htwo}{\mbox{\hspace{2em}}}
\ncom{\hthree}{\mbox{\hspace{3em}}}
\ncom{\hfour}{\mbox{\hspace{4em}}}
\ncom{\vone}{\vskip 2ex}
\ncom{\vtwo}{\vskip 4ex}
\ncom{\vonee}{\vskip 1.5ex}
\ncom{\vthree}{\vskip 6ex}
\ncom{\vfour}{\vspace*{8ex}}
\ncom{\norm}{\|\;\;\|}
\ncom{\integ}[4]{\int_{#1}^{#2}\,{#3}\,d{#4}}
\ncom{\vspan}[1]{{{\rm\,span}\{ #1 \}}}
\ncom{\dm}[1]{ {\displaystyle{#1} } }
\ncom{\ri}[1]{{#1} \index{#1}}
\ncom{\vecs}{\mathbf{s}}
\ncom{\rvec}{\mathbf{r}}
\ncom{\rbvec}{\mathbf{\bar{r}}}
\ncom{\rmat}{\mathbf{R}}
\ncom{\rbmat}{\mathbf{\bar{R}}}
\ncom{\mmat}{\mathbf{M}}
\ncom{\mbmat}{\mathbf{\bar{M}}}
\ncom{\nmat}{\mathbf{N}}
\ncom{\zmat}{\mathbf{Z}}

\usepackage[T1]{fontenc}
\usepackage{lmodern}
\usepackage{bm}
\renewcommand{\vec}[1]{\bm{#1}}

\ncom{\vecone}{\vec{{1}}}
\ncom{\identity}{\vec{I}}
\ncom{\zero}{\vec{{0}}}
\ncom{\matone}{\vec{J}}
\ncom{\hatmat}{\vec{H}}
\ncom{\eps}{\epsilon}
\ncom{\cmat}{\vec{C}}
\ncom{\cbmat}{\vec{\bar{C}}}

\usepackage{titlesec}
 \usepackage{etoolbox, chngcntr}
\AtBeginEnvironment{appendices}{%
 \titleformat{\section}{\bfseries\Large}{\appendixname~\thesection:}{0.5em}{}%
 \titleformat{\subsection}{\bfseries\large}{\thesubsection}{0.5em}{}%
\counterwithin{equation}{section}
}

\ncom{\mybib}{\bibliography{references}\bibliographystyle{apalike}}

\newtheorem{remarks}{\bf Remark}[section]

\newtheorem{lemmas}{\bf Lemma}[section]

\newtheorem{thm}{\bf Theorem}[section]

\numberwithin{equation}{section}

\makeatletter
\def\namedlabel#1#2{\begingroup
    #2%
    \def\@currentlabel{#2}%
    \phantomsection\label{#1}\endgroup
}
\makeatother

\makeatletter
\def\@makefnmark{%
  \leavevmode
  \raise.9ex\hbox{\fontsize\sf@size\z@\normalfont\tiny\@thefnmark}}
\makeatother

\makeatother

\providecommand{\keywords}[1]
{
  \small	
  \textbf{\textit{Keywords---}} #1
}

\usepackage{natbib}

\setcitestyle{aysep={}}
\usepackage{placeins}

\usepackage{threeparttable}

\begin{document}
\renewcommand{\subsectionautorefname}{subsection}

\pagenumbering{arabic}

\title{\bf Efficient designs for multivariate crossover trials}
 \author{Shubham Niphadkar\thanks{ Department of Mathematics, Indian Institute of Technology Bombay, Mumbai 400 076, India}\hspace{.2cm} and Siuli Mukhopadhyay\footnotemark[1] \textsuperscript{,} \thanks{Corresponding author; Department of Mathematics, Indian Institute of Technology Bombay, Powai, Mumbai 400 076, India; Email: siuli@math.iitb.ac.in}\hspace{.2cm} }
\date{}
  \maketitle

\begin{abstract}
This article aims to study efficient/trace optimal designs for crossover trials with multiple responses recorded from each subject in the time periods. A multivariate fixed effects model is proposed with direct and carryover effects corresponding to the multiple responses. The corresponding error dispersion matrix is chosen to be either of the proportional or the generalized Markov covariance type, permitting the existence  of direct and cross-correlations within and between the multiple responses. The corresponding information matrices for direct effects under the two types of dispersions are used to determine efficient designs. The efficiency of orthogonal array designs of Type $I$ and strength $2$ is investigated for a wide choice of covariance functions, namely, Mat($0.5$), Mat($1.5$) and Mat($\infty$). To motivate these multivariate crossover designs, a gene expression dataset in a $3 \times 3$ framework is utilized.
\end{abstract}
\vspace{.5em}

\keywords{Completely symmetric; Correlated response; Markov-type covariance; Orthogonal arrays; Proportional covariance; Trace optimal.}

\section{Introduction}
In biomedical and clinical studies, we often come across crossover or repeated measure setups with $g$ responses, where $g>1$. For instance, \citet{Grender1993AnalysisResponse} studies the effect of caffeine on the mathematical ability of boys aged 14 and 16 years in a crossover framework. In this study, two responses diastolic and systolic blood pressures, of each boy are recorded from each of the two periods in a $2\times 2$ crossover design setup.  More recently, \citet{leaker2017nasal} considered a randomized, placebo controlled, double-blind $3 \times 3$ crossover trial, where in each time period gene expression values corresponding to multiple genes were recorded for each subject, resulting in a multiple response crossover trial. Due to more than one response being measured from the same subject at the same and different time periods, there is a necessity to account for both between and within response correlations in the covariance structure. Using the gene experiment study of \citet{leaker2017nasal}, we motivate a multivariate response model for a $p,t\geq 3$ crossover setup with correlation within and between responses. Our main objective is to find $p$-period and $t$-treatment trace optimal/efficient crossover designs for $g$ correlated responses. Though these correlated response crossover trials are quite popular in the statistical literature \citep[see][]{Grender1993AnalysisResponse, chinchilli1996design, tudor200020, pareek2023likelihood}, mostly the topics of estimation and testing of hypothesis for such trials have been studied without any attention being paid to the actual designing of such studies. To the best of our knowledge, this is the first work in optimal designs for correlated response crossover experiments.\par

\noindent The key difficulty in finding optimal designs for such crossover models is specifying the direct and cross-covariance for the multiple responses. In the direct covariance, correlations arising from within responses measured in different time periods from each subject are acknowledged. This closely resembles the covariance pattern we see in usual longitudinal/repeated measures studies. To consider the relations between different responses measured in the same and different time periods we use cross-correlations. Two popular covariance structures from the statistical literature, resembling the proportional covariance and the generalized Markovian-type covariance \citep[see][]{Journel1999955, chiles2012geostatistics, dasgupta2022optimal}, are utilized to model the correlations within (direct covariance) and between (cross-covariance) responses. Under the proportional structure, the between response covariance is assumed to be proportional to the within response covariance, $\boldsymbol{V}$, where $\boldsymbol{V}$ is positive definite. The Markov-type covariance has been proposed in the statistical literature solely for $g=2$ or the bivariate response case, with the within and between response covariances  formed by  two respective matrices, $\boldsymbol{V}_1$ and $\boldsymbol{V}_R$, where $\boldsymbol{V}_1$ is a positive definite matrix and $\boldsymbol{V}_R$ is a valid correlation matrix. Possible choices of covariance functions, for example the Exponential (Mat($0.5$)), Gaussian (Mat($\infty$)) and Matern are used in this article for $\boldsymbol{V}$ and ($\boldsymbol{V}_1$, $\boldsymbol{V}_R$) (see Section~\ref{working-cov} for details). Thus, it is important to note that the proportional and Markovian-type covariances are not at all restrictive in nature, in fact they allow us to explore a wide range of covariance functions for modelling the within and between response correlations.\par

\noindent For obtaining optimal designs, information matrices under the proportional and generalized Markovian-type covariance assumptions are determined. We show that these information matrices fail to satisfy the usual sufficient conditions \citep{Kiefer1975ConstructionIi} and the conditions by \citet{yen1986conditions} for universal optimal designs, leading us to instead search for trace optimal or efficient crossover designs \citep{bate2006construction}. For the proportional structure, an orthogonal array design represented by $OA_{I} \left( n=\lambda t \left(t-1 \right), p=t, t, 2 \right)$, where $\lambda$ is a positive integer and $t\geq 3$, is proved to be trace optimal/efficient. Under the more complicated Markovian structure, for $g=2$, i.e., a bivariate model, we show that the same orthogonal array design is highly efficient for various choices of within and between response covariance matrices. Further, we utilize the gene experiment dataset to illustrate the choice of a highly efficient design in a multivariate crossover setup. To reiterate, it is very important to understand that in this article we are able to determine efficient designs theoretically, for a wide range of covariance structures as opposed to obtaining optimal designs only for the exponential covariance case \citep[][]{dasgupta2022optimal}.\par

\noindent Various researchers, namely, \citet{matthews1987optimal, matthews1990optimal}, \citet{kunert1991cross}, \citet{Stufken1991SomeDesigns}, \citet{Kushner1997OptimalityDesigns, Kushner1998OptimalObservations}, \citet{bate2006construction} and \citet{Singh2021EfficientSettings} have studied efficient/highly efficient crossover designs. While, universally optimal crossover designs have been discussed by \citet{hedayat1978repeated}, \citet{cheng1980balanced}, \citet{kunert1984optimality}, \citet{Kunert2000OptimalityErrors}, \citet{hedayat2003universal, hedayat2004universal}, to name a few. For a detailed review of crossover designs, we would like to refer to the paper by \citet{Bose2013DevelopmentsDesigns} and books by \citet{Senn2002Cross-overResearch}, \citet{Bose2009OptimalDesigns} and \citet{Kenward2014CrossoverTrials}. However, all of these works concentrate only on single responses measured in each period. Very recently, \citet{niphadkar2023universally} have determined universal optimality results for multiple response crossover trials, where the responses are assumed to be uncorrelated with a homoscedastic covariance matrix.\par

\noindent The original contributions of this article include (i) a multivariate response crossover model with proportional and Markovian covariance structures to model correlations, (ii) information matrices for direct effects, (iii) theoretical results for efficient designs for a wide range of covariance functions, particularly -trace optimal design under the proportional covariance structure and highly efficient designs for the Markovian structure.

\subsection{Motivating example: A gene expression study}
\label{motexample}
\renewcommand{\figureautorefname}{Fig.}
\renewcommand{\thefigure}{\arabic{figure}}
We motivate our multivariate crossover model using a gene expression dataset from \citet{leaker2017nasal}. In this study, multiple gene expressions from subjects were measured in a randomized, double-blind $3\times 3$ crossover setup to understand the biomarkers of mucosal inflammation. Two doses of an oral drug were compared with a placebo, thus resulting in a $3$ treatment trial. Each of the subjects were assigned to any one of the $3$ treatment sequences; $ABC$, $CAB$ and $BCA$, where treatment $A$ represented the $10$ mg dose of the drug, while $B$ and $C$ were the placebo and $25$ mg dose of the drug, respectively. For the purpose of illustration, we considered responses/measurements corresponding to $3$ genes from each of the $3$ periods, thus in the gene example, $g=3$. The information about these $3$ genes is given in \autoref{table}. This dataset on gene profile is publicly available at the NCBI Gene Expression Omnibus \citep{Clough2016} with accession number $GSE67200$. \autoref{fig:period effect} and \autoref{fig:treatment effect} display the plot of average response values for each gene with respect to periods and treatments, respectively, where the averaging is done over subjects. \autoref{fig:subject effect} displays a similar average response plot with respect to subject, where averaging is done over periods. From \autoref{fig:period effect}, \autoref{fig:subject effect} and \autoref{fig:treatment effect} we note that effects due to periods, subjects and treatments differ across genes. \autoref{fig:overall effect} shows that the average response values vary across genes. These plots thus help us to develop a multivariate crossover model in Section~\ref{proposed-model} with varying intercepts for each gene and differential effects of periods, treatments and subjects on the $3$ gene responses.\par

\noindent We also studied the correlation profile of the $3$ genes in detail and noted significant correlations between multiple gene pairs. In order to check for correlation within and between the $3$ responses, several tests were performed. The test results are shown in \autoref{within_gene} and \autoref{between_gene}. We note from \autoref{within_gene} that at the $5\%$ level of significance, within gene correlation for genes $1$ and $2$ are high between periods, while in \autoref{between_gene} we see significant correlations present between gene pairs $(1,2)$ and $(2,3)$.

\begin{figure}
\centering
\begin{subfigure}{.5\textwidth}
  \centering
  \includegraphics[height=\linewidth, width=1.0\linewidth]{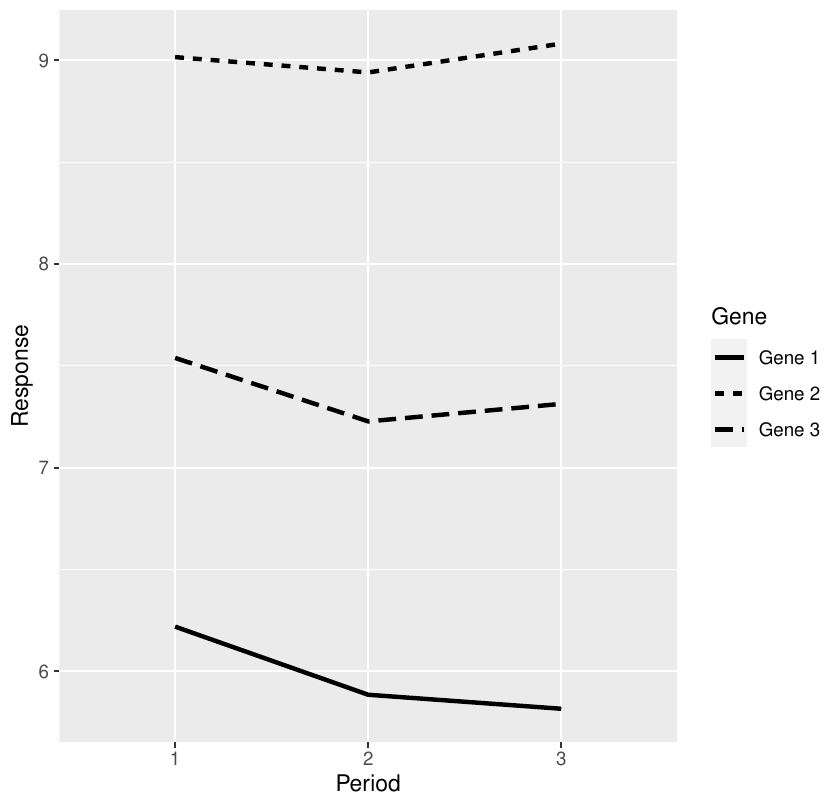}
  \caption{Period effects}
  \label{fig:period effect}
\end{subfigure}%
\begin{subfigure}{.5\textwidth}
  \centering
  \includegraphics[height=\linewidth, width=1.0\linewidth]{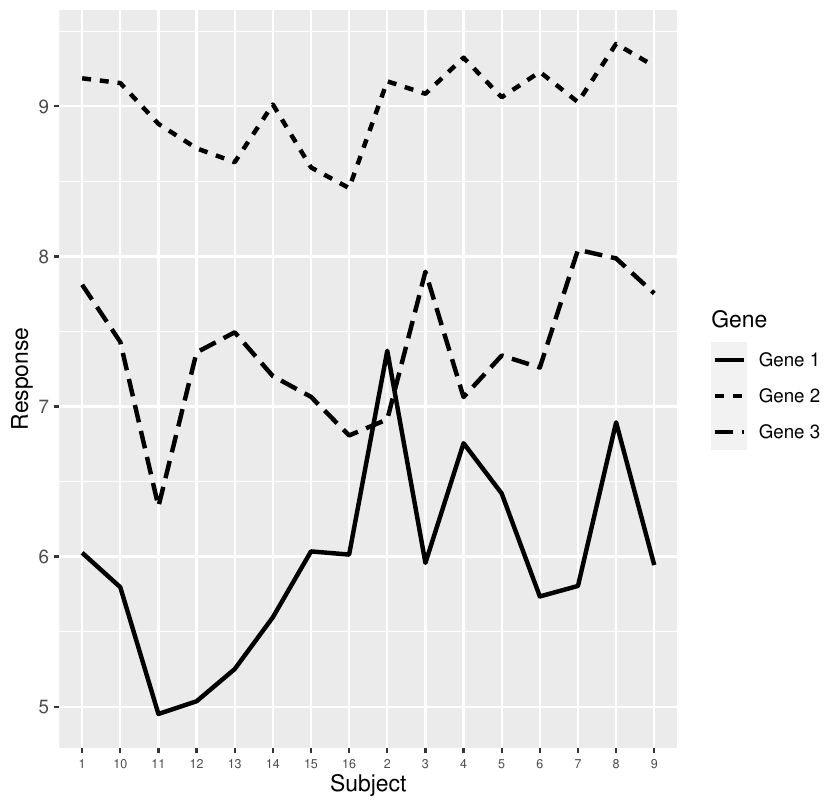}
  \caption{Subject effects}
  \label{fig:subject effect}
\end{subfigure}\\
\begin{subfigure}{.5\textwidth}
  \centering
  \includegraphics[height=\linewidth, width=1.0\linewidth]{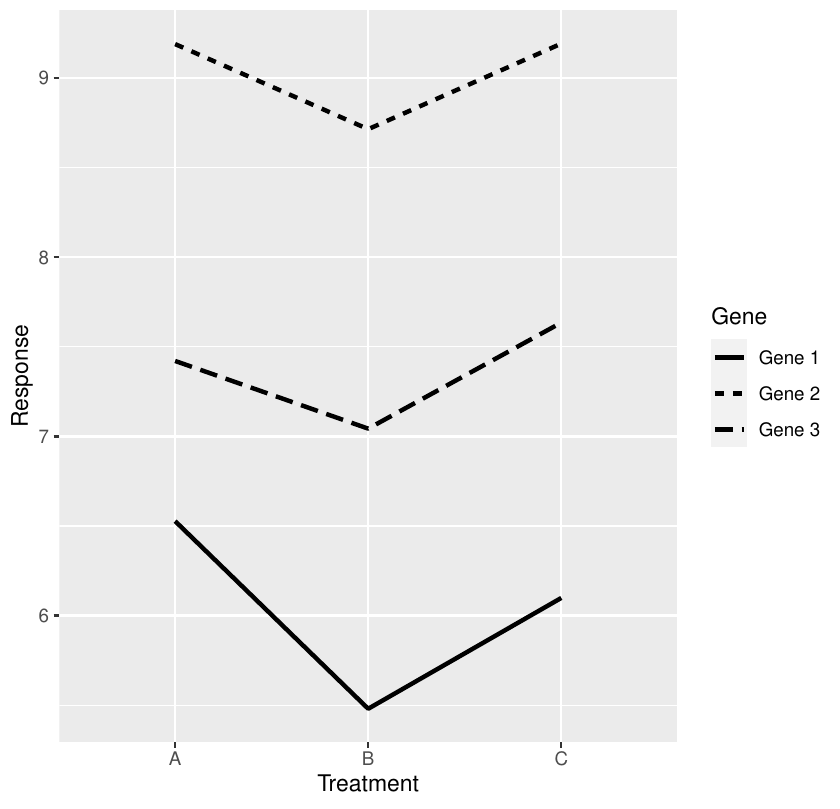}
  \caption{Treatment effects}
  \label{fig:treatment effect}
\end{subfigure}%
\begin{subfigure}{.5\textwidth}
  \centering
  \includegraphics[height=\linewidth, width=1.0\linewidth]{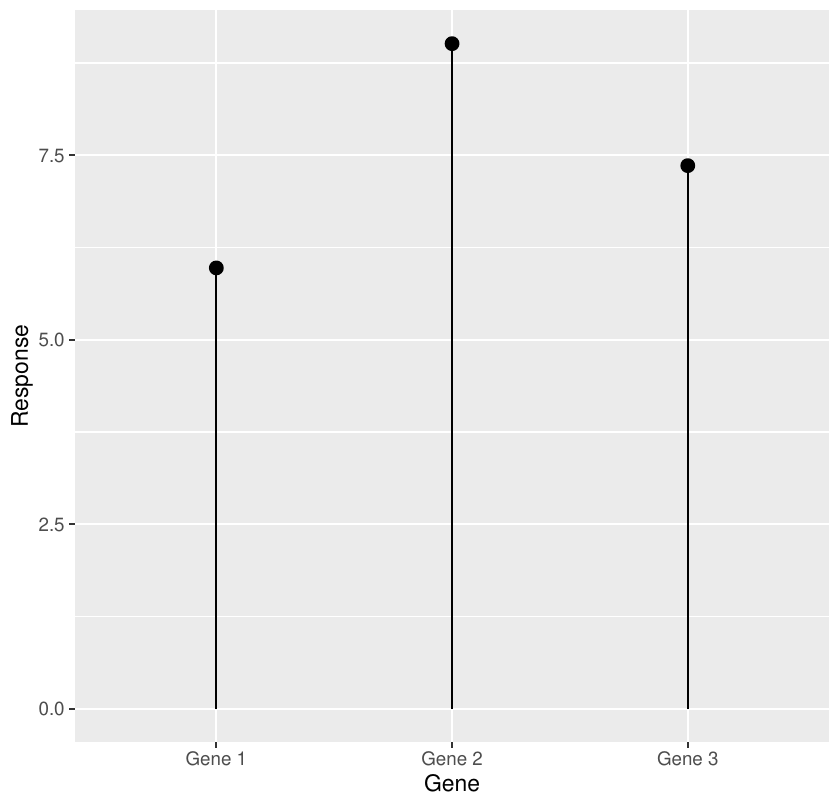}
  \caption{Overall mean}
  \label{fig:overall effect}
\end{subfigure}
\caption{Plots of the period, subject, treatment and gene versus average response corresponding to the $3$ genes.}
\label{effects-example}
\end{figure}

\begin{figure}
\centering
\begin{subfigure}{.3\textwidth}
  \centering
  \includegraphics[height=\linewidth, width=1.0\linewidth]{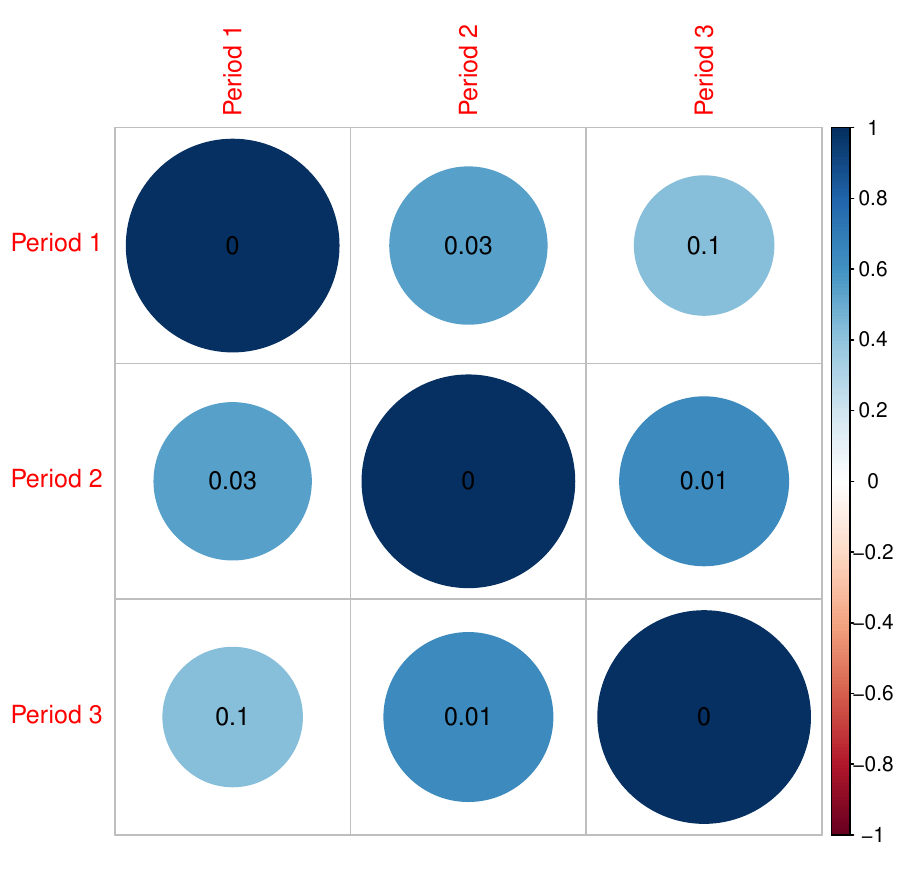}
  \caption{$Gene_1$}
  \label{fig:g1}
\end{subfigure}%
\begin{subfigure}{.3\textwidth}
  \centering
  \includegraphics[height=\linewidth, width=1.0\linewidth]{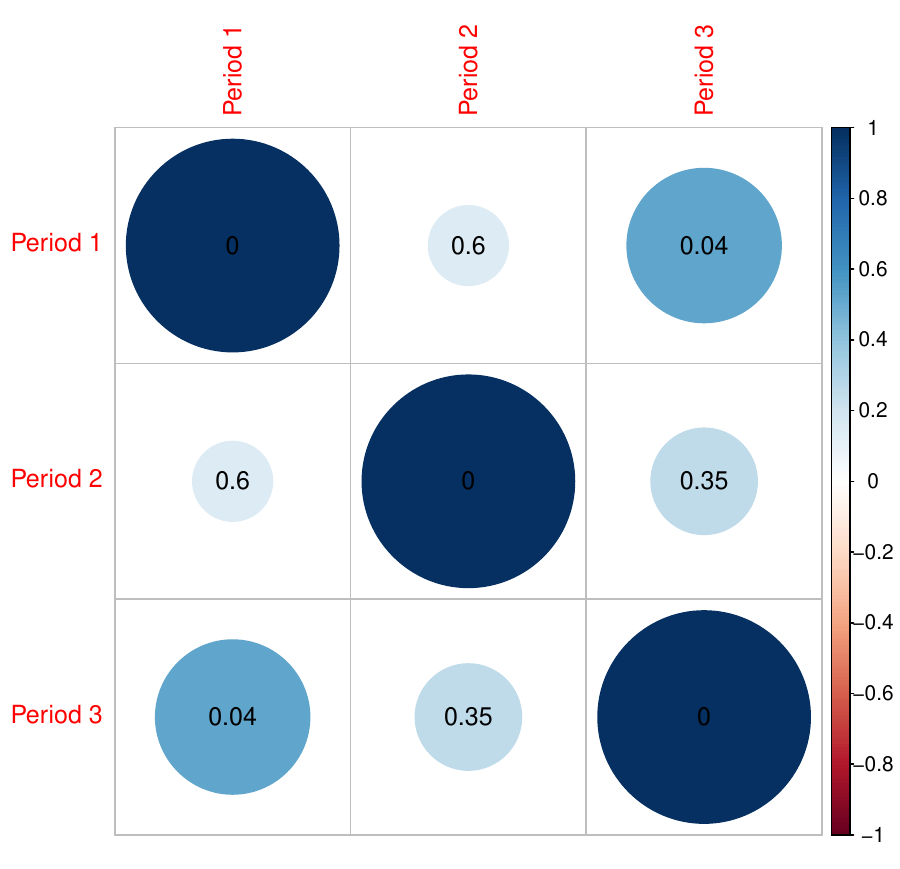}
  \caption{$Gene_2$}
  \label{fig:g2}
\end{subfigure}%
\begin{subfigure}{.3\textwidth}
  \centering
  \includegraphics[height=\linewidth, width=1.0\linewidth]{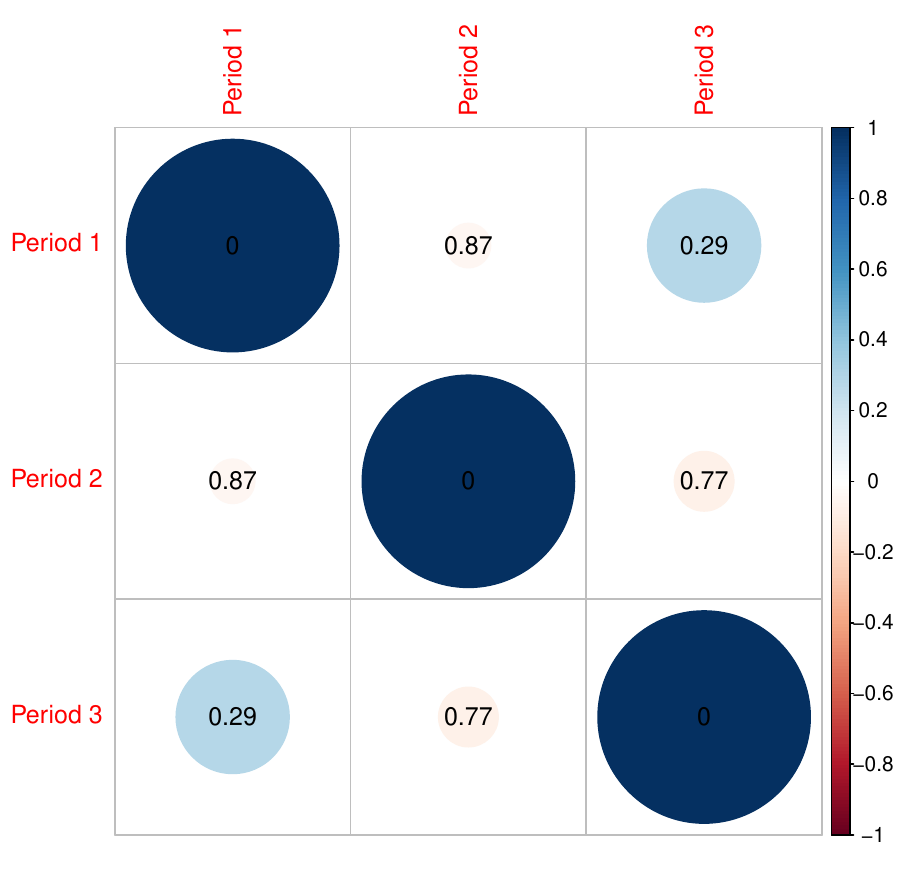}
  \caption{$Gene_3$}
  \label{fig:g3}
\end{subfigure}
\caption{Within response correlation plots corresponding to the $3$ genes. The values are the p-values of the corresponding correlation tests while the colours represent the quantity of correlation present.}
\label{within_gene}
\end{figure}

\begin{figure}
\centering
\begin{subfigure}{.3\textwidth}
  \centering
  \includegraphics[height=\linewidth, width=1.0\linewidth]{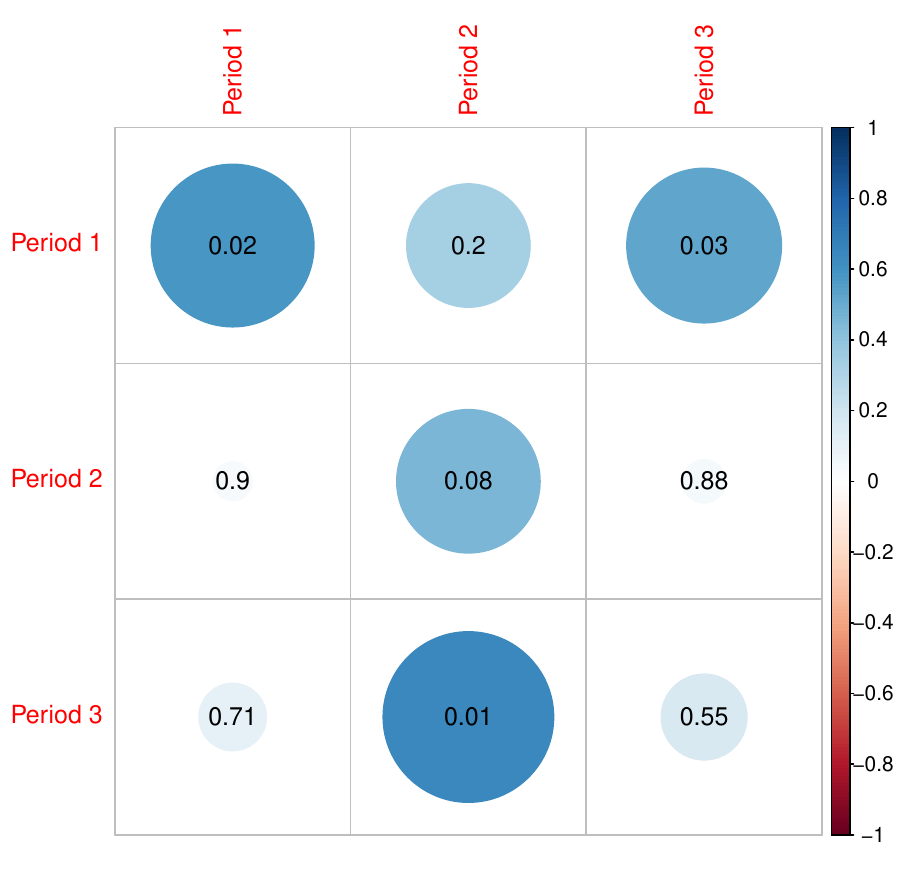}
  \caption{$\left( Gene_1, Gene_2 \right)$}
  \label{fig:g1g2}
\end{subfigure}%
\begin{subfigure}{.3\textwidth}
  \centering
  \includegraphics[height=\linewidth, width=1.0\linewidth]{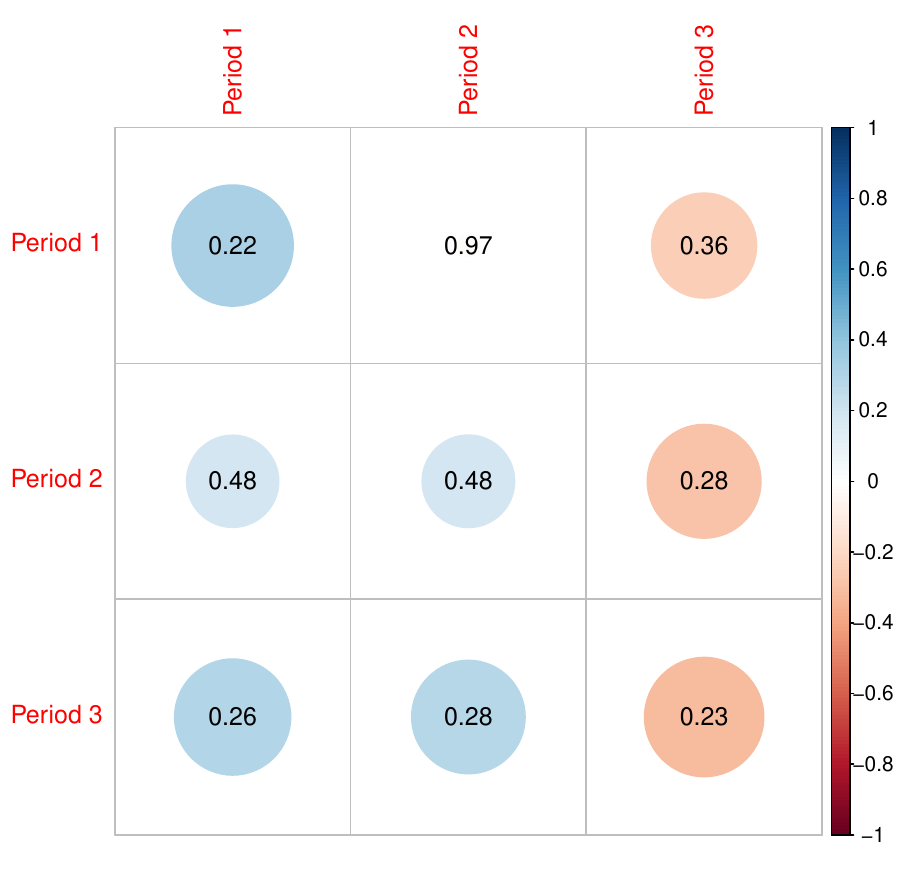}
  \caption{$\left( Gene_1, Gene_3 \right)$}
  \label{fig:g1g3}
\end{subfigure}%
\begin{subfigure}{.3\textwidth}
  \centering
  \includegraphics[height=\linewidth, width=1.0\linewidth]{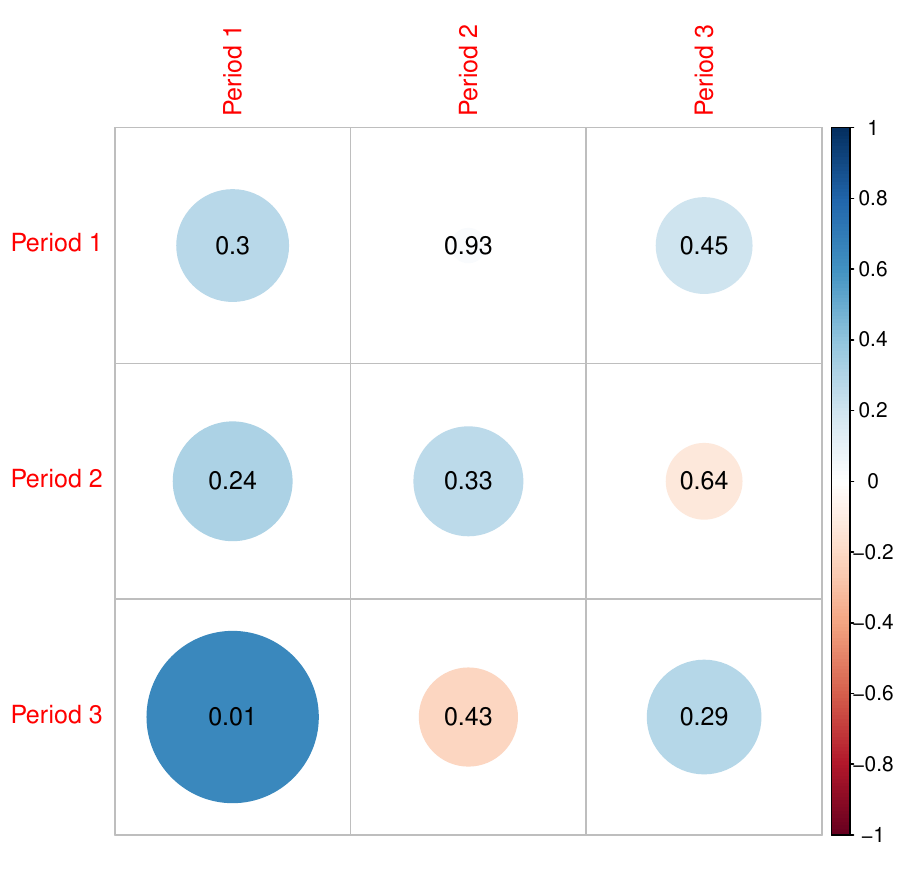}
  \caption{$\left( Gene_2, Gene_3 \right)$}
  \label{fig:g2g3}
\end{subfigure}
\caption{Between response correlation plots corresponding to the $3$ gene pairs. The values are the p-values of the corresponding correlation tests while the colours represent the quantity of correlation present.}
\label{between_gene}
\end{figure}

\noindent Learning from the motivating gene example, in the following section we propose a multivariate model that has varying effects of periods, subjects and treatments corresponding to each response variable, and also includes correlation between and within responses.

\section{A multivariate crossover model}
\label{proposed-model}
We consider the class, $\Omega_{t,n,p}$, of crossover designs with $t$ treatments, $n$ subjects and $p$ periods, where $t,p \geq 3$. It is assumed that measurements are recorded on $g \geq 1$ response variables from every subject in each period. Suppose for a design $d \in \Omega_{t,n,p}$, $Y_{dijk}$ represents the $k^{th}$ response measurement from the $j^{th}$ subject in the  $i^{th}$ period of the experiment, where $1 \leq k \leq g$, $1 \leq i \leq p$ and $1 \leq j \leq n$. We model $Y_{dijk}$ as:
\begin{align}
Y_{dijk} = \mu_k + \alpha_{i,k} + \beta_{j,k} + \tau_{d \left( i,j \right),k} + \rho_{d \left( i-1,j \right),k}  + \eps_{ijk},
\label{model1}
\end{align}
where corresponding to the $k^{th}$ response variable, $\mu_k$ is the intercept effect, $\alpha_{i,k}$ is the $i^{th}$ period effect, $\beta_{j,k}$ is the $j^{th}$ subject effect, $\tau_{s,k}$ and $\rho_{s,k}$ are respectively the direct effect and the first order carryover effect due to the $s^{th}$ treatment where $1 \leq s \leq t$,  and $d(i,j)$ denotes the treatment allocated to the $j^{th}$ subject in the $i^{th}$ period. The assumptions of varying effects are in accordance with the observations made in the gene expression study in Subsection~\ref{motexample}. We assume that \eqref{model1} is a fixed effect model with a zero mean error term, $\eps_{ijk}$, and no carryover effect in the first period. If $g=1$, then the above model becomes an univariate response model similar to those found in the statistical literature \citep[see][]{Bose2009OptimalDesigns}. \par

\noindent Using matrix notations we rewrite the above model as:
\begin{align}
\begin{bmatrix}
 \boldsymbol{Y}^{'}_{d1} &
 \cdots &
\boldsymbol{Y}^{'}_{dg}
\end{bmatrix}^{'}
&= \boldsymbol{Z}_d 
\boldsymbol{\eta} +
\begin{bmatrix}
 \boldsymbol{\eps}^{'}_{1} &
 \cdots &
\boldsymbol{\eps}^{'}_{g}
\end{bmatrix}^{'},
\label{model2}
\end{align}
where
$
\boldsymbol{Y}_{dk} = \left( Y_{d11k}, \cdots, Y_{dp1k}, \cdots, Y_{d1nk}, \cdots, Y_{dpnk} \right)^{'}$ and $
\boldsymbol{\eps}_k = \left( \eps_{11k}, \cdots, \eps_{pnk} \right)^{'},
$ represent the observations and the error terms, respectively, corresponding to the $k^{th}$ response, and
\setcounter{MaxMatrixCols}{20}
$\boldsymbol{\eta} =$ \\$
\begin{bmatrix}
\mu_1 & 
\cdots &
\mu_g &
\boldsymbol{\alpha}^{'}_1 &
\boldsymbol{\beta}^{'}_1 & 
\cdots &
\boldsymbol{\alpha}^{'}_g &
\boldsymbol{\beta}^{'}_g &
\boldsymbol{\tau}^{'}_1 &
\boldsymbol{\rho}^{'}_1 &
\cdots &
\boldsymbol{\tau}^{'}_g &
\boldsymbol{\rho}^{'}_g
\end{bmatrix}^{'}$ is the parameter vector. The design matrix is given by
\begin{align}
\boldsymbol{Z}_d &= 
\begin{bmatrix}
\identity_g \otimes \vecone_{np} & \boldsymbol{Z_1} & \boldsymbol{Z_2}
\end{bmatrix},
\label{Zd}
\end{align}
where $\boldsymbol{Z_1} = \identity_g \otimes \boldsymbol{X_1}$, $\boldsymbol{Z_2} = \identity_g \otimes \boldsymbol{X_2}$, $\boldsymbol{X_1} = 
\begin{bmatrix} 
\boldsymbol{P} & \boldsymbol{U}
\end{bmatrix}$ and $\boldsymbol{X_2} = 
\begin{bmatrix}  
\boldsymbol{T}_d & \boldsymbol{F}_d 
\end{bmatrix}$. For the periods and subjects, we have $\boldsymbol{P}=\vecone_n \otimes \identity_p$, $\boldsymbol{U}=\identity_n \otimes \vecone_p$, while $\boldsymbol{T}_d$ and $\boldsymbol{F}_d$ are the matrices for treatment and carryover effects, respectively. Note that $\boldsymbol{F}_d = \left( \identity_n \otimes \boldsymbol{\psi} \right) \boldsymbol{T}_d$, where $\boldsymbol{\psi} =  
\begin{bmatrix}
\zero^{'}_{p-1 \times 1} & 0\\
\identity_{p-1} & \zero_{p-1 \times 1}
\end{bmatrix}$. The vector of errors is assumed to have multivariate normal distribution with $E(\boldsymbol{\eps})=\zero_{np \times 1}$ and the dispersion matrix as
\begin{align}
\mathbb{D} \left( 
\boldsymbol{\eps}
\right)
&=
\boldsymbol{\mathit{\Sigma}} =
\begin{bmatrix}
\boldsymbol{\mathit{\Sigma}}_{11} & \cdots & \boldsymbol{\mathit{\Sigma}}_{1g}\\
\vdots & \vdots & \vdots\\
\boldsymbol{\mathit{\Sigma}}_{g1} & \cdots & \boldsymbol{\mathit{\Sigma}}_{gg}
\end{bmatrix},
\label{dispepsc}
\end{align} 
where $\boldsymbol{\mathit{\Sigma}}_{k'k} = \boldsymbol{\mathit{\Sigma}}_{kk'}^{'}$, for $k \neq k' = 1, \cdots, g$. Here, the symmetric matrices, $\boldsymbol{\mathit{\Sigma}}_{kk}$ represent the variance of $\boldsymbol{\eps}_k$, while $\boldsymbol{\mathit{\Sigma}}_{kk'}$ represents the covariance between $\boldsymbol{\eps}_k$ and $\boldsymbol{\eps}_{k'}$. The matrix $\boldsymbol{\mathit{\Sigma}}$ is assumed to be positive definite. Also, observations from different subjects are taken to be uncorrelated, as is the usual practice \citep[see][]{Bose2009OptimalDesigns}.\par

\noindent The investigation made regarding the presence of correlation between genes (see \autoref{between_gene}) in Subsection~\ref{motexample}, motivates the use of the covariance structure with cross-covariances. Two different structures are considered to model the covariances between responses, namely, the proportional structure and the generalized Markov-type structure. A generalized version of the Markov-type covariance structure for spatial data \citep[][]{Journel1999955, chiles2012geostatistics} was proposed recently by \citet{dasgupta2022optimal}. The proportional structure is simpler than the Markovian one and is used for the general $g>1$ case. However, in the statistical literature, the Markovian structure has been studied for only the $g=2$ or bivariate response case. We also restrict our search for optimal designs for a Markovian covariance structure to $g=2$ or bivariate responses.
\begin{description}
\item[\namedlabel{structure1}{Proportional Structure}:] 
Under the proportional structure the dispersion matrix $\boldsymbol{\mathit{\Sigma}}$ is taken as follows:
\begin{align*}
\boldsymbol{\mathit{\Sigma}} &= \boldsymbol{\mathit{\Gamma}} \otimes \left( \identity_n \otimes \boldsymbol{V} \right),
\end{align*}
where $\boldsymbol{\mathit{\Gamma}}$ is assumed to be a $g \times g$ positive definite and symmetric matrix having non-zero off-diagonal elements, and $\boldsymbol{V}$ a known positive definite and symmetric matrix. For $g=2$ case, the matrix $\boldsymbol{\mathit{\Gamma}}$ is given as
\begin{align*}
\boldsymbol{\mathit{\Gamma}} = 
\begin{bmatrix}
\gamma_{11} & \gamma_{12}\\
\gamma_{12} & \gamma_{22}
\end{bmatrix},
\end{align*}
where $\gamma_{11}, \gamma_{22} > 0$ and $\gamma_{11} \gamma_{22} > \gamma_{12}^2$. Under the proportional structure, the within and between responses have similar covariance matrices differing by some constants.
\item[\namedlabel{structure2}{Generalized Markov-Type Structure}:]
In the statistical literature \citep[see][]{Journel1999955, chiles2012geostatistics, dasgupta2022optimal}, this structure is available to model the $g=2$ case. Let $\sigma_{11}$ and $\sigma_{22}$, where $\sigma_{11}, \sigma_{22}>0$, be the respective variances of the first and second responses and $\rho$, $0<|\rho|<1$,  denotes the correlation coefficient between the observations from the two responses measured from the same period. Under the generalized Markov-type structure the dispersion matrix $\boldsymbol{\mathit{\Sigma}}$ is as follows:
\begin{align*}
\boldsymbol{\mathit{\Sigma}} &=
\begin{bmatrix}
\boldsymbol{\mathit{\Sigma}}_{11} & \boldsymbol{\mathit{\Sigma}}_{12}\\
\boldsymbol{\mathit{\Sigma}}_{21} & \boldsymbol{\mathit{\Sigma}}_{22}
\end{bmatrix},
\end{align*}
where $\boldsymbol{\mathit{\Sigma}}_{11} = \identity_n \otimes \boldsymbol{V}_1$, $\boldsymbol{\mathit{\Sigma}}_{12} = \boldsymbol{\mathit{\Sigma}}_{21} = \rho \sqrt{\frac{\sigma_{22}}{\sigma_{11}} }\boldsymbol{\mathit{\Sigma}}_{11}$, $\boldsymbol{\mathit{\Sigma}}_{22} = \rho^2 \frac{\sigma_{22}}{\sigma_{11}} \boldsymbol{\mathit{\Sigma}}_{11} + \sigma_{22} (1-\rho^2) \boldsymbol{\mathit{\Sigma}}_{R}$ and $\boldsymbol{\mathit{\Sigma}}_{R} = \identity_n \otimes \boldsymbol{V}_R$. Here, $\boldsymbol{V}_1 = \sigma_{11} \boldsymbol{V}_C$, $\boldsymbol{V}_C$ is a known positive definite and symmetric matrix, and $\boldsymbol{V}_R$ is a known correlation matrix. 
\end{description}

\subsection{Working covariance structures}
\label{working-cov}
The two working covariance structures for the dispersion matrix $\boldsymbol{\mathit{\Sigma}}$, namely, the proportional structure and the generalized Markov-type structure are not at all restrictive, as they allow usage of various covariance functions as shown next. For the proportional case, the only constraint on the matrix $\boldsymbol{V}$ is positive definiteness and symmetricity. Some popular choices of $\boldsymbol{V}$ are shown in \autoref{table_prop}.
\begin{table}
\begin{threeparttable}[b]
\caption{Some popular choices for $\boldsymbol{V}$ in the proportional covariance structure}\label{table_prop}
\centering
\begin{tabular}{p{5cm}|p{10.5cm}}
\toprule%
 Covariance function for $\boldsymbol{V}$ & $\left( \boldsymbol{V} \right)_{i_1, i_2}: \left(i_1, i_2 \right)^{th} \text{ element of } \boldsymbol{V}$\\
\hline
& $0 < r_{(1)} < 1$, and $i_1$ and $i_2$ are period indices, $1 \leq i_1, i_2 \leq p$.\\
\hline
Mat($0.5$) \tnote{1,2}& $r_{(1)}^{|i_1 - i_2|}$\\
Mat($1.5$) \tnote{1}& $\left[ 1 - |i_1 - i_2| log\left(r_{(1)} \right) \right]r_{(1)}^{|i_1 -i_2|}$\\
Mat($\infty$) \tnote{1,2} & $r_{(1)}^{|i_1 - i_2|^2}$\\
\hline
\end{tabular}
\begin{tablenotes}
       \item [1] See \citet{dasgupta2022optimal} and \citet{li2015}. 
       \item [2] Note that $r_{(1)}$ is the within and between response correlation coefficient between consecutive periods.
     \end{tablenotes}
 \end{threeparttable}
\end{table}

\noindent For the generalized Markov-type covariance, we need to choose matrices $\boldsymbol{V}_1$ and $\boldsymbol{V}_R$. Note, $\boldsymbol{V}_1 = \sigma_{11} \boldsymbol{V}_C$, where $\boldsymbol{V}_C$ is a known positive definite and symmetric matrix, $\sigma_{11}>0$ and $\boldsymbol{V}_R$ is a known correlation matrix. Various popular combinations of the matrices $\boldsymbol{V}_1$ and $\boldsymbol{V}_R$ are presented in \autoref{table_markov}. Note that we could also use the same covariance functions for $\boldsymbol{V}_1$ and $\boldsymbol{V}_R$ with different correlation parameters between consecutive periods.
\begin{table}
\begin{threeparttable}[b]
\caption{Some popular combinations of $\boldsymbol{V}_1$ and $\boldsymbol{V}_R$ under the generalized Markov-type structure}\label{table_markov}
\centering
\begin{tabular}{p{4cm}|p{5.65cm}|p{5.5cm}}
\toprule%
Covariance function for $\boldsymbol{V}_1$ and $\boldsymbol{V}_R$ & $\left( \boldsymbol{V}_1 \right)_{i_1, i_2}$: $\left(i_1, i_2 \right)^{th}$ element of $\boldsymbol{V}_1$ & $\left( \boldsymbol{V}_R \right)_{i_1, i_2}$: $\left(i_1, i_2 \right)^{th}$ element of $\boldsymbol{V}_R$\\
\hline
& \multicolumn{2}{c}{$0 < r_{(1)} < 1$, and $i_1$ and $i_2$ are period indices, $1 \leq i_1, i_2 \leq p$.}\\
\hline
& Note that $\sigma_{11}>0$ for all the structures below. &\\
\hline
(Mat($0.5$), Mat($1.5$)) & $\sigma_{11} r_{(1)}^{|i_1 - i_2|}$ & $\left[ 1 - |i_1 - i_2| log\left(r_{(1)} \right) \right]r_{(1)}^{|i_1 -i_2|}$\\
(Mat($0.5$), Mat($\infty$)) & $\sigma_{11} r_{(1)}^{|i_1 - i_2|}$ & $r_{(1)}^{|i_1 - i_2|^2}$\\
(Mat($1.5$), Mat($0.5$)) & $\sigma_{11} \left[ 1 - |i_1 - i_2| log\left(r_{(1)} \right) \right] r_{(1)}^{|i_1 -i_2|}$ & $r_{(1)}^{|i_1 - i_2|}$\\
(Mat($1.5$), Mat($\infty$)) & $\sigma_{11} \left[ 1 - |i_1 - i_2| log\left(r_{(1)} \right) \right]r_{(1)}^{|i_1 -i_2|}$ & $r_{(1)}^{|i_1 - i_2|^2}$\\
(Mat($\infty$), Mat($0.5$)) & $\sigma_{11} r_{(1)}^{|i_1 - i_2|^2}$ & $r_{(1)}^{|i_1 - i_2|}$\\
(Mat($\infty$), Mat($1.5$)) & $\sigma_{11} r_{(1)}^{|i_1 - i_2|^2}$ & $\left[ 1 - |i_1 - i_2| log\left(r_{(1)} \right) \right]r_{(1)}^{|i_1 -i_2|}$\\
(Mat($0.5$), Mat($0.5$)) \tnote{3} & $\sigma_{11} r_{(1)}^{|i_1 - i_2|}$ & $r_{(1)}^{2|i_1 - i_2|}$\\
\hline 
\end{tabular}
\begin{tablenotes}
       \item [3] $\boldsymbol{V}_1$ and $\boldsymbol{V}_R$ have Mat($0.5$) structure with correlation between adjacent periods as $r_{(1)}$ and $r_{(1)}^2$, respectively. This structure has been named as NS1 in \citet{li2015} and \citet{dasgupta2022optimal}.
     \end{tablenotes}
 \end{threeparttable}
\end{table}

\noindent To aid the reader, we use an illustration based on a crossover design with $t=p=n=3$ and  $g=2$. In this case we write $\boldsymbol{\mathit{\Sigma}}$ as
\begin{align*}
\boldsymbol{\mathit{\Sigma}} = 
\begin{bmatrix}
\boldsymbol{G}_{11} & \zero_{3 \times 3} & \zero_{3 \times 3} & \boldsymbol{G}_{12} & \zero_{3 \times 3} & \zero_{3 \times 3}\\
\zero_{3 \times 3} & \boldsymbol{G}_{11} & \zero_{3 \times 3} & \zero_{3 \times 3} & \boldsymbol{G}_{12} & \zero_{3 \times 3}\\
\zero_{3 \times 3} & \zero_{3 \times 3} &  \boldsymbol{G}_{11} & \zero_{3 \times 3} & \zero_{3 \times 3} &  \boldsymbol{G}_{12}\\
\boldsymbol{G}_{12} & \zero_{3 \times 3} & \zero_{3 \times 3} & \boldsymbol{G}_{22} & \zero_{3 \times 3} & \zero_{3 \times 3}\\
\zero_{3 \times 3} & \boldsymbol{G}_{12} & \zero_{3 \times 3} & \zero_{3 \times 3} & \boldsymbol{G}_{22} & \zero_{3 \times 3}\\
\zero_{3 \times 3} & \zero_{3 \times 3} &  \boldsymbol{G}_{12} & \zero_{3 \times 3} & \zero_{3 \times 3} &  \boldsymbol{G}_{22}
\end{bmatrix},
\end{align*}
where under the proportional covariance, for $\boldsymbol{V}$ with Mat($\infty$) type structure, we have
\begin{align*}
\boldsymbol{G}_{lm}= 
\begin{bmatrix}
\gamma_{l,m} &  \gamma_{l,m} r_{(1)} & \gamma_{l,m} r_{(1)}^4\\
 \gamma_{l,m} r_{(1)} & \gamma_{l,m} & \gamma_{l,m} r_{(1)}\\
\gamma_{l,m} r_{(1)}^4 & \gamma_{l,m} r_{(1)} & \gamma_{l,m}
\end{bmatrix}.
\end{align*}
Here, $0<r_{(1)}<1$, $\gamma_{l,m}>0$, $1 \leq l < m \leq 2$, and $\gamma_{11} \gamma_{22} > \gamma^2_{12}$. We can see that the matrices $\boldsymbol{G}_{11}$, $\boldsymbol{G}_{12}$ and $\boldsymbol{G}_{22}$ are proportional to each other. Under the generalized Markov-type covariance, for the NS1 structure of $\boldsymbol{\mathit{\Sigma}}$, we have
\begin{multline*}
\boldsymbol{G}_{11}= \sigma_{11}
\begin{bmatrix}
1 &  r_{(1)} & r_{(1)}^2\\
 r_{(1)} & 1 & r_{(1)}\\
r_{(1)}^2 &  r_{(1)} & 1
\end{bmatrix},~
\boldsymbol{G}_{12}= \rho \sqrt{\sigma_{11} \sigma_{22}}
\begin{bmatrix}
1 &   r_{(1)} &  r_{(1)}^2\\
  r_{(1)} & 1 &  r_{(1)}\\
 r_{(1)}^2 & r_{(1)} & 1
\end{bmatrix} \text{ and}\\
\boldsymbol{G}_{22}= \sigma_{22}
\begin{bmatrix}
1 &   r_{(1)} \left( \rho^2 +r_{(1)} \left(1-\rho^2 \right) \right) &  r_{(1)}^2 \left( \rho^2 +r^2_{(1)} \left(1-\rho^2 \right) \right)\\
  r_{(1)} \left( \rho^2 +r_{(1)} \left(1-\rho^2 \right) \right) & 1 &  r_{(1)} \left( \rho^2 +r_{(1)} \left(1-\rho^2 \right) \right)\\
 r_{(1)}^2 \left( \rho^2 +r^2_{(1)} \left(1-\rho^2 \right) \right) &  r_{(1)} \left( \rho^2 +r_{(1)} \left(1-\rho^2 \right) \right) & 1
\end{bmatrix},
\end{multline*}
where $\sigma_{11}, \sigma_{22}>0$, $0<|\rho|<1$ and $0<r_{(1)}<1$. It can be easily observed that the matrices $\boldsymbol{G}_{11}$ and $\boldsymbol{G}_{12}$ are proportional to each other, but for $r_{(1)} \neq 0$, the matrix $\boldsymbol{G}_{22}$ is not proportional to $\boldsymbol{G}_{11}$.\par

\noindent The choices presented in \autoref{table_prop} and \autoref{table_markov} are not exhaustive. We could have chosen any matrix for $\boldsymbol{V}$ (proportional case) as long as it is positive definite and symmetric, while in the Markovian case, $\boldsymbol{V}_1$ needs to be positive definite and symmetric while $\boldsymbol{V}_R$ can be any valid correlation matrix.

\section{Search for efficient designs}
\label{effi}
We next determine a trace optimal/efficient design for the direct effects in model \eqref{model2}. Using the idea as provided by \citet{bate2006construction}, we say that a design $d^* \in \mathcal{D}$, where $\mathcal{D}$ is a subclass of designs, is a trace optimal/efficient design for the parameters of interest over $\mathcal{D}$ if $d^*$ maximizes the trace of the corresponding information matrix over $\mathcal{D}$.\par

\noindent For any design $d \in \Omega_{t,n,p}$, let $\boldsymbol{C}_{d(s1)}$ and $\boldsymbol{C}_{d(s2)}$ represent the information matrices of the direct effects for the proportional and generalized Markov-type structures, respectively. \autoref{lemma5-c4-i} and \autoref{lemma5-c4} show that the two information matrices have the following representations:
\begin{align}
\boldsymbol{C}_{d(s1)} &=  \boldsymbol{C}_{d(s1)(11)} - \boldsymbol{C}_{{d(s1)}(12)} \boldsymbol{C}^{-}_{{d (s1)}(22)} \boldsymbol{C}_{{d (s1)}(21)},\label{markov-t25c-1}\\
\boldsymbol{C}_{d(s2)} &= \boldsymbol{C}_{{d(s2)}(11)(1)} - \boldsymbol{C}_{{d(s2)}(12)(1)} \boldsymbol{C}^{-}_{{d(s2)}(22)(1)} \boldsymbol{C}_{{d(s2)}(21)(1)}, \label{markov-t26c-1}
\end{align}
where $ \boldsymbol{C}_{d(s1)(11)} = \left( \identity_g \otimes \boldsymbol{T}_d^{'} \right) \boldsymbol{A}^* \left( \identity_g \otimes \boldsymbol{T}_d\right)$, $\boldsymbol{C}_{d(s1)(12)} = \boldsymbol{C}_{d(s1)(21)}^{'} = \left( \identity_g \otimes \boldsymbol{T}_d^{'} \right) \boldsymbol{A}^* \left(\identity_g \otimes \boldsymbol{F}_d \right)$, \\$\boldsymbol{C}_{{d(s1)}(22)} = \left(\identity_g \otimes \boldsymbol{F}^{'}_d\right) \boldsymbol{A}^* \left(\identity_g \otimes \boldsymbol{F}_d \right)$, $\boldsymbol{C}_{{d(s2)}(11)(1)} = \left( \identity_2 \otimes \boldsymbol{T}_d^{'} \right) \boldsymbol{A}^* \left( \identity_2 \otimes \boldsymbol{T}_d\right)$, $\boldsymbol{C}_{{d(s2)}(12)(1)}=$\\$ \boldsymbol{C}^{'}_{{d(s2)}(21)(1)}=\left( \identity_2 \otimes \boldsymbol{T}_d^{'} \right) \boldsymbol{A}^* \left(\identity_2 \otimes \boldsymbol{F}_d \hatmat_t\right)$, $\boldsymbol{C}_{{d(s2)}(22)(1)} = \left(\identity_2 \otimes \boldsymbol{F}^{'}_d \hatmat_t\right) \boldsymbol{A}^* \left(\identity_2 \otimes \boldsymbol{F}_d \hatmat_t\right)$, and the matrix  $\boldsymbol{A}^*$ is  $\boldsymbol{\mathit{\Sigma}}^{-1/2} pr^{\perp} \left(\boldsymbol{\mathit{\Sigma}}^{-1/2} \boldsymbol{Z_1}
\right) \boldsymbol{\mathit{\Sigma}}^{-1/2}$. Here, $\hatmat_t = \identity_t - \frac{1}{t} \vecone_{t} \vecone{t}^{'}$. Note that $\boldsymbol{A}^*$ depends on the structure of $\boldsymbol{\mathit{\Sigma}}$ and $\boldsymbol{C}_{d(s1)}$ is a $gt \times gt$ matrix, while $\boldsymbol{C}_{d(s2)}$ is a $2t \times 2t$ matrix.

\subsection{Proportional covariance structure}
\label{proportional covariance}
Under the proportional covariance, we prove that for the $g>1$ setup, a design represented by $OA_{I} \left( n=\lambda t \left(t-1 \right), p=t, t, 2 \right)$, where $\lambda$ is a positive integer and $t\geq 3$, maximizes the $tr \left(  \boldsymbol{C}_{d(s1)} \right)$ and hence is a trace optimal/efficient design.\par

\noindent We know from \eqref{markov-t25c-1}, that the information matrix for the direct effects depends on the expression of $\boldsymbol{A}^*$. However, till date, there has been no work discussing the explicit expression of $\boldsymbol{A}^*$ for bivariate/multivariate responses in a crossover framework. Using \autoref{prop-lemma3-c4-1}, the expression of the matrix $\boldsymbol{A}^*$ under the proportional structure for the $g>1$ case is obtained as
\begin{align*}
\boldsymbol{A}^* 
&= 
\boldsymbol{\mathit{\Gamma}}^{-1} \otimes \left(\hatmat_n \otimes \boldsymbol{V}^* \right),
\end{align*}
where $\boldsymbol{V}^* = \boldsymbol{V}^{-1} - \left( \vecone_{p}^{'} \boldsymbol{V}^{-1} \vecone_p \right)^{-1} \boldsymbol{V}^{-1} \matone_{p \times p} \boldsymbol{V}^{-1}$, $\matone_{p \times p} = \vecone_{p} \vecone_{p}^{'}$ and $\hatmat_n = \identity_n - \frac{1}{n} \vecone_{n} \vecone_{n}^{'}$. In the next theorem, the expression of $\boldsymbol{A}^*$ is used to obtain the information matrix for the direct effects under the proportional structure for the multivariate setup.
\renewcommand{\subsectionautorefname}{Appendix}
\begin{thm}
\label{prop-lemma4-c4-1}
The information matrix for the direct effects under $g>1$ case can be expressed as
\begin{align}
\boldsymbol{C}_{d(s1)} &= \boldsymbol{\mathit{\Gamma}}^{-1} \otimes \boldsymbol{C}_{d(uni)},
\label{prop-t20c-1}
\end{align}
where $\boldsymbol{C}_{d(uni)} = \boldsymbol{C}_{d(uni)(11)} - \boldsymbol{C}_{d(uni)(12)} \boldsymbol{C}_{d(uni)(22)}^{-} \boldsymbol{C}_{d(uni)(21)}$, $\boldsymbol{C}_{d(uni)(11)} = \boldsymbol{T}_d^{'} \left( \hatmat_n \otimes \boldsymbol{V}^* \right) \boldsymbol{T}_d$, \\$\boldsymbol{C}_{d(uni)(12)} = \boldsymbol{C}_{d(uni)(21)}^{'} = \boldsymbol{T}_d^{'} \left( \hatmat_n \otimes \boldsymbol{V}^* \right) \boldsymbol{F}_d $, and $\boldsymbol{C}_{d(uni)(22)} = \boldsymbol{F}_d^{'} \left( \hatmat_n \otimes \boldsymbol{V}^* \right) \boldsymbol{F}_d $. Here  $\boldsymbol{V}^* = \boldsymbol{V}^{-1} - \left( \vecone_{p}^{'} \boldsymbol{V}^{-1} \vecone_p \right)^{-1} \boldsymbol{V}^{-1} \matone_{p \times p} \boldsymbol{V}^{-1}$.
\end{thm}
\begin{proof}
The details of the proof are given in \autoref{proportion-information matrix}.
\end{proof}

It should be noted that $\boldsymbol{C}_{d(uni)}$ is the information matrix for the direct effects when we have a univariate version of the model \eqref{model1}, i.e., $g=1$. To check if $OA_{I} \left( n=\lambda t \left(t-1 \right), p=t,\right.$\\$\left. t, 2 \right)$, where $\lambda$ is a positive integer and $t \geq 3$, is universally optimal, we explore complete symmetricity of the corresponding information matrix for $g>1$ case.
\begin{remarks}
Let $d^*$ be a design given by $OA_{I} \left( n=\lambda t \left(t-1 \right), p=t, t, 2 \right)$, where $\lambda$ is a positive integer and $t \geq 3$. Then for $g>1$, the information matrix for the direct effects corresponding to $d^*$, $\boldsymbol{C}_{d^*(s1)}$, is not completely symmetric.
\label{prop-remark-1}
\end{remarks}
\begin{proof}
The proof is given in \autoref{proportion-information matrix}.
\end{proof}

\noindent From \autoref{prop-remark-1}, it is clear that the usual sufficient condition of complete symmetricity similar to the one given by \citet{Kiefer1975ConstructionIi} for finding a universally optimal design is not applicable here. Also, to satisfy the more general conditions similar to those given by \citet{yen1986conditions}, for $g>1$ case, we need
\begin{align*}
\begin{split}
\frac{\gamma^{(11)} |\boldsymbol{E} | }{ t e_{22}  }   &= \cdots = \frac{\gamma^{(gg)} |\boldsymbol{E} | }{ t e_{22}  }, \text{ and}\\
\frac{\gamma^{(kk)} |\boldsymbol{E} | }{t e_{22}  } &=  \frac{\gamma^{(kk')} |\boldsymbol{E} | }{t e_{22}  },
\end{split}
\end{align*}
where $\gamma^{(kk')}$ is the $(k,k')^{th}$ element of the matrix $\boldsymbol{\mathit{\Gamma}}^{-1}$. However, this condition is also seen to be violated for the $g>1$ setup (see Proof of \autoref{prop-remark-1}). Note, for a multivariate setup with no correlation between responses and a much simpler covariance structure, \citet{niphadkar2023universally} have successfully obtained universally optimal designs. However, this is not true when the responses are correlated. Hence for the multivariate response case, we resort to identifying an efficient design in the next theorem.

\begin{thm}
Let $d^* \in \mathcal{D}^{(1)}_{t,n=\lambda t (t-1),p=t}$ be a design given by $OA_{I} \left( n=\lambda t \left(t-1 \right), p=t, t, 2 \right)$, where $\mathcal{D}^{(1)}_{t,n=\lambda t (t-1),p=t}$ is a class of binary designs with $p=t$, $\lambda$ is a positive integer and $t \geq 3$. Then $d^*$ is a trace optimal/efficient design for the direct effects over $\mathcal{D}^{(1)}_{t,n=\lambda t (t-1),p=t}$ for the $g>1$ case.
\label{prop-thm-1}
\end{thm}
\begin{proof}
\label{prop-ss3}
Here, $d^* \in \mathcal{D}^{(1)}_{t,n=\lambda t (t-1),p=t}$ is a design given by $OA_{I} \left( n=\lambda t \left(t-1 \right), p=t, t, 2 \right)$, where \\$\mathcal{D}^{(1)}_{t,n=\lambda t (t-1),p=t}$ is a class of binary designs with $p=t$, $\lambda$ is a positive integer and $t \geq 3$. It has been shown by \citet{Kunert2000OptimalityErrors} that $d^*$ is universally optimal for the direct effects over $\mathcal{D}^{(1)}_{t,n=\lambda t (t-1),p=t}$ in the univariate response case. Note that the universal optimality of a design over a subclass of designs implies that the design maximizes the trace of the corresponding information matrix over the same subclass of designs. Hence, $d^*$ maximizes $tr \left( \boldsymbol{C}_{d(uni)} \right)$ over $\mathcal{D}^{(1)}_{t,n=\lambda t (t-1),p=t}$, where $\boldsymbol{C}_{d(uni)}$ is the information matrix of the direct effects for the univariate response case. For $g>1$, using the expression of $\boldsymbol{C}_{d(s1)}$ from \autoref{prop-lemma4-c4-1}, we get
\begin{align}
tr \left( \boldsymbol{C}_{d(s1)} \right) &= tr \left( \boldsymbol{\mathit{\Gamma}}^{-1} \right) tr \left( \boldsymbol{C}_{d(uni)} \right).
\label{markov-t27c-1}
\end{align}
Here, $\boldsymbol{\mathit{\Gamma}}^{-1}$ is a positive definite matrix and thus $tr \left( \boldsymbol{\mathit{\Gamma}}^{-1} \right) > 0$. So from \eqref{markov-t27c-1}, we get that for $g>1$, $d^*$ maximizes $tr \left( \boldsymbol{C}_{d(s1)} \right)$ over $\mathcal{D}^{(1)}_{t,n=\lambda t (t-1),p=t}$. Thus $d^*$ is trace optimal/efficient for the direct effects over $\mathcal{D}^{(1)}_{t,n=\lambda t (t-1),p=t}$ for the $g>1$ case.
\end{proof}

\subsection{Generalized Markov-type covariance structure}
For the Markovian structure, we are able to show that the $OA_{I} \left( n=\lambda t \left(t-1 \right), p=t, t, 2 \right)$, where $\lambda$ is a positive integer and $t\geq 3$, is a highly efficient design. To evaluate the trace optimality/efficiency of these $OA_{I} \left( n=\lambda t \left(t-1 \right), p=t, t, 2 \right)$ designs, we obtain an upper bound of $tr \left( \boldsymbol{C}_{d(s2)} \right)$, and state a necessary and sufficient condition for  a $OA_{I} \left( n=\lambda t \left(t-1 \right), p=t, t, 2 \right)$, to attain this upper bound of $tr \left( \boldsymbol{C}_{d(s2)} \right)$. Note, that if this upper bound is attained then the design maximizes $tr \left(  \boldsymbol{C}_{d(s2)} \right)$ and hence is an efficient design. Also, for  various choices of $\boldsymbol{V}_1$, $\boldsymbol{V}_R$ and $\rho$,  we study the trace of the information matrix for the $OA_{I} \left( n=\lambda t \left(t-1 \right), p=t, t, 2 \right)$ and the derived upper bound.\par

\noindent Proceeding as before, we see that the process of obtaining  $\boldsymbol{A}^*$ is more complicated as compared to the proportional structure. This is mainly because under the Markovian structure, $\boldsymbol{\mathit{\Sigma}}_{22}$ is not proportional to $\boldsymbol{\mathit{\Sigma}}_{11}$. By laborious calculation, (see \autoref{lemma3-c4-1}), we obtain the expression of the matrix $\boldsymbol{A}^*$ under the generalized Markov-type structure as
\begin{align*}
\boldsymbol{A}^* 
&= 
\begin{bmatrix}
\hatmat_n \otimes \boldsymbol{\mathit{\Omega}}_1 & -\hatmat_n \otimes \boldsymbol{\mathit{\Omega}}_2\\
-\hatmat_n \otimes \boldsymbol{\mathit{\Omega}}_2 & \hatmat_n \otimes \boldsymbol{\mathit{\Omega}}_4
\end{bmatrix},
\end{align*}
where $\boldsymbol{\mathit{\Omega}}_1 = \boldsymbol{V}_1^* + \frac{\bar{\rho}^2}{\sigma_{12}} \boldsymbol{V}_R^*$, $\boldsymbol{\mathit{\Omega}}_2 =  \frac{\bar{\rho}}{\sigma_{12}} \boldsymbol{V}_R^*$ and $\boldsymbol{\mathit{\Omega}}_4 = \frac{1}{\sigma_{12}} \boldsymbol{V}_R^*$. Here, $\boldsymbol{V}_1^* = \boldsymbol{V}_1^{-1} - \left( \vecone_{p}^{'} \boldsymbol{V}_1^{-1} \vecone_p \right)^{-1} \boldsymbol{V}_1^{-1} \times$\\$\matone_{p \times p} \boldsymbol{V}_1^{-1}$, $\boldsymbol{V}_R^* = \boldsymbol{V}_R^{-1} - \left( \vecone_{p}^{'} \boldsymbol{V}_R^{-1} \vecone_p \right)^{-1} \boldsymbol{V}_R^{-1} \matone_{p \times p} \boldsymbol{V}_R^{-1}$, $\hatmat_n = \identity_n - \frac{1}{n} \vecone_{n} \vecone_{n}^{'}$, $\matone_{p \times p} = \vecone_{p} \vecone_{p}^{'}$, $\bar{\rho} = \rho \sqrt{\frac{\sigma_{22}}{\sigma_{11}}}$ and $\sigma_{12} = \sigma_{22} \left(1 
- \rho^2 \right)$.\par

\noindent Let $d^*$ be a design given by $OA_{I} \left( n=\lambda t \left(t-1 \right), p=t, t, 2 \right)$, where $\lambda$ is a positive integer and $t \geq 3$. Using the above expression of $\boldsymbol{A}^*$, the information matrix for the direct effects and its trace under the generalized Markov-type structure, corresponding to $d^*$, are obtained.
\begin{thm}
\label{thm5-c4}
Let $d^*$ be a design given by $OA_{I} \left( n=\lambda t \left(t-1 \right), p=t, t, 2 \right)$, where $\lambda$ is a positive integer and $t \geq 3$. Then under the generalized Markov-type structure, the information matrix for the direct effects can be expressed as
\begin{align}
\boldsymbol{{C}}_{d^*(s2)} &= \frac{n}{t-1} 
\begin{bmatrix}
\boldsymbol{\mathit{\Lambda}}_1 & \boldsymbol{\mathit{\Lambda}}_2\\
\boldsymbol{\mathit{\Lambda}}_2 & \boldsymbol{\mathit{\Lambda}}_4
\end{bmatrix},
\label{p37c}
\end{align}
where the matrices $\boldsymbol{\mathit{\Lambda}}_1$, $\boldsymbol{\mathit{\Lambda}}_2$ and $\boldsymbol{\mathit{\Lambda}}_4$ are given as 
\begin{align}
\begin{split}
\boldsymbol{\mathit{\Lambda}}_1 &= \left( tr \left( \boldsymbol{V}_1^* \right) + \frac{\bar{\rho}^2}{\sigma_{12}} tr \left( \boldsymbol{V}_R^* \right) - \frac{\left(  tr \left( \boldsymbol{V}_1^* \boldsymbol{\psi} \right) \right)^2}{tr \left( \hatmat_p \boldsymbol{\psi}^{'} \boldsymbol{V}_1^* \boldsymbol{\psi} \right) } - \frac{ \bar{\rho}^2}{\sigma_{12}} \frac{\left(tr \left( \boldsymbol{V}_R^* \boldsymbol{\psi} \right) \right)^2}{tr \left( \hatmat_p \boldsymbol{\psi}^{'} \boldsymbol{V}_R^* \boldsymbol{\psi} \right)} \right) \hatmat_t,\\
\boldsymbol{\mathit{\Lambda}}_2 &= - \frac{ \bar{\rho} }{\sigma_{12}} \left( tr \left(\boldsymbol{V}_R^* \right) -  \frac{\left(tr \left( \boldsymbol{V}_R^* \boldsymbol{\psi} \right) \right)^2}{tr \left( \hatmat_p \boldsymbol{\psi}^{'} \boldsymbol{V}_R^* \boldsymbol{\psi} \right)} \right) \hatmat_t, \text{ and}\\
\boldsymbol{\mathit{\Lambda}}_4 &= \frac{1}{\sigma_{12}} \left(  tr \left( \boldsymbol{V}_R^* \right) - \frac{\left(tr \left( \boldsymbol{V}_R^* \boldsymbol{\psi} \right) \right)^2}{tr \left( \hatmat_p \boldsymbol{\psi}^{'} \boldsymbol{V}_R^* \boldsymbol{\psi} \right)} \right) \hatmat_t.
\end{split}
\label{markov-t30c-1}
\end{align}
Further, we have
\begin{align}
tr \left( \boldsymbol{{C}}_{d^*(s2)} \right) &= n \left( tr \left( \boldsymbol{V}_1^* \right) + \frac{1 + \bar{\rho}^2}{\sigma_{12}} tr \left( \boldsymbol{V}_R^* \right) -\frac{\left(  tr \left( \boldsymbol{V}_1^* \boldsymbol{\psi} \right) \right)^2}{tr \left( \hatmat_p \boldsymbol{\psi}^{'} \boldsymbol{V}_1^* \boldsymbol{\psi} \right) } - \frac{1 + \bar{\rho}^2}{\sigma_{12}} \frac{\left(tr \left( \boldsymbol{V}_R^* \boldsymbol{\psi} \right) \right)^2}{tr \left( \hatmat_p \boldsymbol{\psi}^{'} \boldsymbol{V}_R^* \boldsymbol{\psi} \right)} \right).
\label{markov-t19c-1}
\end{align}
\end{thm}
\begin{proof}
The details of the proof are given in \autoref{markov-information matrix}.
\end{proof}

\noindent As in Subsection~\ref{proportional covariance}, we show next that the sufficient conditions for universal optimality as similar to those discussed in \citet{Kiefer1975ConstructionIi} and \citet{yen1986conditions} are not applicable.  
\renewcommand{\subsectionautorefname}{Appendix}
\begin{remarks}
Let $d^*$ be a design given by $OA_{I} \left( n=\lambda t \left(t-1 \right), p=t, t, 2 \right)$, where $\lambda$ is a positive integer and $t \geq 3$. Then the information matrix for the direct effects corresponding to a design $d^*$, $\boldsymbol{{C}}_{d^*(s2)}$, is not completely symmetric.  
\label{markov-remark-2}
\end{remarks}
\begin{proof}
For the proof see \autoref{markov-information matrix}.
\end{proof}

\noindent From \autoref{markov-remark-2}, it is clear that as in the proportional covariance case, here also we cannot apply the usual sufficient condition for a universally optimal design. For conditions similar to \citet{yen1986conditions}, we require
\begin{multline*}
\frac{1}{t} \left( tr \left( \boldsymbol{V}_1^* \right) + \frac{\bar{\rho}^2}{\sigma_{12}} tr \left( \boldsymbol{V}_R^* \right) - \frac{\left(  tr \left( \boldsymbol{V}_1^* \boldsymbol{\psi} \right) \right)^2}{tr \left( \hatmat_p \boldsymbol{\psi}^{'} \boldsymbol{V}_1^* \boldsymbol{\psi} \right) } - \frac{ \bar{\rho}^2}{\sigma_{12}} \frac{\left(tr \left( \boldsymbol{V}_R^* \boldsymbol{\psi} \right) \right)^2}{tr \left( \hatmat_p \boldsymbol{\psi}^{'} \boldsymbol{V}_R^* \boldsymbol{\psi} \right)} \right) =\\ \frac{1}{\sigma_{12} t} \left(  tr \left( \boldsymbol{V}_R^* \right) - \frac{\left(tr \left( \boldsymbol{V}_R^* \boldsymbol{\psi} \right) \right)^2}{tr \left( \hatmat_p \boldsymbol{\psi}^{'} \boldsymbol{V}_R^* \boldsymbol{\psi} \right)} \right)  = - \frac{ \bar{\rho} }{\sigma_{12} t } \left( tr \left(\boldsymbol{V}_R^* \right) -  \frac{\left(tr \left( \boldsymbol{V}_R^* \boldsymbol{\psi} \right) \right)^2}{tr \left( \hatmat_p \boldsymbol{\psi}^{'} \boldsymbol{V}_R^* \boldsymbol{\psi} \right)} \right).
\end{multline*}
Since both conditions are not satisfied (see proof of \autoref{markov-remark-2}), we instead focus on efficient designs. In order to check whether $d^*$ is an efficient design, we obtain an upper bound of $tr \left( \boldsymbol{C}_{d(s2)} \right)$. If $d^*$ attains this upper bound then we say that $d^*$ is an efficient design. However, obtaining an upper bound for $tr \left( \boldsymbol{C}_{d(s2)} \right)$ is computationally challenging, so we instead consider $\boldsymbol{\tilde{C}}_{d(s2)}$ (see \autoref{lemma5-c4-1}), the information matrix for the direct effects with no period effects, and show that $\boldsymbol{{C}}_{d(s2)} \leq \boldsymbol{\tilde{C}}_{d(s2)}$ with equality if $d$ is uniform on periods (see \autoref{lemma5-c4-1} and \autoref{thm3-c4}).\par

\noindent Note, $d^* \in \mathcal{D}^{(1)}_{t,n=\lambda t (t-1),p=t}$, where 
\begin{align*}
\mathcal{D}^{(1)}_{t,n=\lambda t (t-1),p=t} = \{ d: d \in \Omega_{t,n=\lambda t (t-1) ,p=t}, \text{ } t \geq 3 \text{ and } d \text{ is a binary design} \}.
\end{align*}
Also, for $p=t$, a binary design is a design uniform on subjects and thus $\mathcal{D}^{(1)}_{t,n=\lambda t (t-1),p=t}$ is indeed a class of designs with $p=t$ which are uniform on subjects.

\begin{lemmas}
\label{thm4-c4}
Let $d \in \mathcal{D}^{(1)}_{t,n=\lambda t (t-1),p=t}$. Then,
\begin{align}
tr \left( \boldsymbol{{C}}_{d(s2)} \right) \leq u \left( t, n, p, \sigma_{11}, \sigma_{22}, \rho, \boldsymbol{V}_1, \boldsymbol{V}_R \right),
\label{p26c}
\end{align}
where $u \left( t, n, p, \sigma_{11}, \sigma_{22}, \rho, \boldsymbol{V}_1, \boldsymbol{V}_R \right) \\= n \Bigg( tr \left(  \boldsymbol{V}_1^*  \right) + \frac{1+\bar{\rho}^2}{\sigma_{12}} tr \left( \boldsymbol{V}_R^*  \right) -
\left( tr \left(  \boldsymbol{V}_1^* \boldsymbol{\psi} \right)  + \frac{1+\bar{\rho}^2}{\sigma_{12}} tr \left(   \boldsymbol{V}_R^* \boldsymbol{\psi} \right) \right)^2 \times$ \\ $ \left( tr \left( \hatmat_p \boldsymbol{\psi}^{'} \boldsymbol{V}_1^* \boldsymbol{\psi} \right) + \frac{1+\bar{\rho}^2}{\sigma_{12}} tr \left( \hatmat_p \boldsymbol{\psi}^{'} \boldsymbol{V}_R^* \boldsymbol{\psi} \right) \right)^{-1} \Bigg)$. Here $\boldsymbol{V}_1^*$, $\boldsymbol{V}_R^*$, $\bar{\rho}$ and $\sigma_{12}$ are as defined before.
\end{lemmas}
\begin{proof}
Note that a design $d \in \mathcal{D}^{(1)}_{t,n=\lambda t (t-1),p=t}$, is a binary design. In \autoref{thm3-c4}, it has been shown that $\boldsymbol{{C}}_{d(s2)} \leq \boldsymbol{\tilde{C}}_{d(s2)}$, so
\begin{align}
tr \left( \boldsymbol{{C}}_{d(s2)} \right) \leq tr \left( \boldsymbol{\tilde{C}}_{d(s2)} \right).
\label{p27c}
\end{align}
Using the expression of $\boldsymbol{\tilde{C}}_{d(s2)}$ from \autoref{lemma5-c4-1} and applying Lemma 5.1 from \citet{kushner1997optimal}, we get
\begin{align}
tr \left( \boldsymbol{\tilde{C}}_{d(s2)} \right) \leq tr \left( \boldsymbol{\tilde{C}}_{d(s2)(11)} \right) - \left(tr \left( \boldsymbol{\tilde{C}}_{d(s2)(12)} \right) \right)^2 \left(tr \left( \boldsymbol{\tilde{C}}_{d(s2)(22)} \right) \right)^{+},
\label{p28c}
\end{align}
where
\begin{align}
\begin{split}
tr \left( \boldsymbol{\tilde{C}}_{d(s2)(11)} \right) &= tr \left( \boldsymbol{T}_{d}^{'} \left( \identity_n \otimes \boldsymbol{\mathit{\Omega}}_1 \right) \boldsymbol{T}_{d} \right) + tr \left( \boldsymbol{T}_{d}^{'} \left( \identity_n \otimes \boldsymbol{\mathit{\Omega}}_4 \right) \boldsymbol{T}_{d} \right), \\
tr \left( \boldsymbol{\tilde{C}}_{d(s2)(12)} \right) &= tr \left( \boldsymbol{T}_{d}^{'} \left( \identity_n \otimes \boldsymbol{\mathit{\Omega}}_1 \boldsymbol{\psi} \right) \boldsymbol{T}_{d} \hatmat_t  \right) + tr \left( \boldsymbol{T}_{d}^{'} \left( \identity_n \otimes \boldsymbol{\mathit{\Omega}}_4 \boldsymbol{\psi} \right) \boldsymbol{T}_{d} \hatmat_t \right), \text{ and}\\
tr \left( \boldsymbol{\tilde{C}}_{d(s2)(22)} \right) &= tr \left( \hatmat_t \boldsymbol{T}_{d}^{'} \left( \identity_n \otimes \boldsymbol{\psi}^{'} \boldsymbol{\mathit{\Omega}}_1 \boldsymbol{\psi} \right) \boldsymbol{T}_{d} \hatmat_t  \right) + tr \left( \hatmat_t \boldsymbol{T}_{d}^{'} \left( \identity_n \otimes \boldsymbol{\psi}^{'} \boldsymbol{\mathit{\Omega}}_4 \boldsymbol{\psi} \right) \boldsymbol{T}_{d} \hatmat_t \right),
\end{split}
\label{p29c}
\end{align}
from \autoref{lemma7-c4-1}. Using \autoref{remark1-c4}, it is observed that $\vecone_{p}^{'} \boldsymbol{\mathit{\Omega}}_1 = \vecone_{p}^{'} \boldsymbol{\mathit{\Omega}}_2 = \vecone_{p}^{'} \boldsymbol{\mathit{\Omega}}_4 = \zero_{1 \times p} $, and
\begin{align}
\begin{split}
tr \left( \boldsymbol{T}_{d}^{'} \left( \identity_n \otimes \boldsymbol{\mathit{\Omega}}_1 \right) \boldsymbol{T}_{d} \right) &= n \times  tr \left( 
\boldsymbol{\mathit{\Omega}}_1
\right),\\
tr \left( \boldsymbol{T}_{d}^{'} \left( \identity_n \otimes \boldsymbol{\mathit{\Omega}}_4 \right) \boldsymbol{T}_{d} \right) &= n \times tr \left( 
\boldsymbol{\mathit{\Omega}}_4
\right),\\
tr \left( \boldsymbol{T}_{d}^{'} \left( \identity_n \otimes \boldsymbol{\mathit{\Omega}}_1 \boldsymbol{\psi} \right) \boldsymbol{T}_{d} \hatmat_t \right) &= n \times tr \left( 
\boldsymbol{\mathit{\Omega}}_1
\boldsymbol{\psi} 
\right),\\
tr \left( \boldsymbol{T}_{d}^{'} \left( \identity_n \otimes \boldsymbol{\mathit{\Omega}}_4 \boldsymbol{\psi} \right) \boldsymbol{T}_{d} \hatmat_t \right) &= n \times tr \left( 
\boldsymbol{\mathit{\Omega}}_4
\boldsymbol{\psi}
\right),\\
tr \left( \hatmat_t \boldsymbol{T}_{d}^{'} \left( \identity_n \otimes \boldsymbol{\psi}^{'} \boldsymbol{\mathit{\Omega}}_1 \boldsymbol{\psi} \right) \boldsymbol{T}_{d} \hatmat_t \right) &= n \times tr \left( \hatmat_p
\boldsymbol{\psi}^{'} 
\boldsymbol{\mathit{\Omega}}_1
\boldsymbol{\psi}
\right), \text{ and}\\
tr \left( \hatmat_t \boldsymbol{T}_{d}^{'} \left( \identity_n \otimes \boldsymbol{\psi}^{'} \boldsymbol{\mathit{\Omega}}_4 \boldsymbol{\psi} \right) \boldsymbol{T}_{d} \hatmat_t \right) &= n \times tr \left( \hatmat_p
\boldsymbol{\psi}^{'} 
\boldsymbol{\mathit{\Omega}}_4
\boldsymbol{\psi}
\right).
\end{split}
\label{p30c}
\end{align}
Suppose we define, $c_{11} = n \left( c_{11(11)} + c_{11(22)} \right)$, $c_{12} = n \left(c_{12(11)} + c_{12(22)} \right)$ and $c_{22} = n \big( c_{22(11)} +$\\$ c_{22(22)} \big)$, where $
c_{11(11)} = tr \left( 
\boldsymbol{\mathit{\Omega}}_1
\right)$, $c_{11(22)} = tr \left( 
\boldsymbol{\mathit{\Omega}}_4
\right)$, $c_{12(11)} = tr \left( 
\boldsymbol{\mathit{\Omega}}_1
\boldsymbol{\psi}
\right)$, $c_{12(22)} = tr \left( 
\boldsymbol{\mathit{\Omega}}_4
\boldsymbol{\psi} 
\right)$, $c_{22(11)} = tr \left( \hatmat_p
\boldsymbol{\psi}^{'}
\boldsymbol{\mathit{\Omega}}_1
\boldsymbol{\psi}
\right)$ and $c_{22(22)} = tr \left( \hatmat_p
\boldsymbol{\psi}^{'}
\boldsymbol{\mathit{\Omega}}_4
\boldsymbol{\psi}
\right)$. From \autoref{lemma3-c4-1}, $\boldsymbol{\mathit{\Omega}}_1 = \boldsymbol{V}_1^* + \frac{\bar{\rho}^2}{\sigma_{12}} \boldsymbol{V}_R^*$ and $\boldsymbol{\mathit{\Omega}}_4 = \frac{1}{\sigma_{12}} \boldsymbol{V}_R^*$, where $\bar{\rho} = \rho \sqrt{\frac{\sigma_{22}}{\sigma_{11}}}$ and $\sigma_{12} = \sigma_{22} (1 - \rho^2)$. Using these expressions in $c_{22(11)}$ and $c_{22(22)}$, we obtain $c_{22(11)} = tr \left( \hatmat_p
\boldsymbol{\psi}^{'}
\boldsymbol{V}_1^* \boldsymbol{\psi} \right) + \frac{\bar{\rho}^2}{\sigma_{12}} tr \left( \hatmat_p
\boldsymbol{\psi}^{'}
\boldsymbol{V}_R^*
\boldsymbol{\psi}
\right)$ and $c_{22(22)} = \frac{1}{\sigma_{12}} tr \left( \hatmat_p
\boldsymbol{\psi}^{'}
\boldsymbol{V}_R^*
\boldsymbol{\psi}
\right)$. Using \autoref{lemma-markov-1}, $| \bar{\rho} | >0$ and $\sigma_{12}>0$ gives $c_{22(11)}$, $c_{22(22)}>0$, implying $c_{22} = n \left( c_{22(11)} + c_{22(22)}  \right) > 0$. Combining \eqref{p27c}, \eqref{p28c}, \eqref{p29c} and \eqref{p30c}, we obtain 
\begin{align}
tr \left( \boldsymbol{{C}}_{d(s2)} \right) &\leq  c_{11} - c_{12}^2c_{22}^{-1}.
\label{markov-t15c-1}
\end{align}
Note that since $c_{22}>0$, $c_{22}^{+}$ has been replaced by $c_{22}^{-1}$. Using the expression of $c_{11}$, $c_{12}$ and $c_{22}$, and the expressions of $\boldsymbol{\mathit{\Omega}}_1$ and $\boldsymbol{\mathit{\Omega}}_4$ in \eqref{markov-t15c-1}, we prove \eqref{p26c}.
\end{proof}

\begin{remarks}
Let $d^* \in \mathcal{D}^{(1)}_{t,n=\lambda t (t-1),p=t}$ be a design given by $OA_{I} \left( n=\lambda t \left(t-1 \right), p=t, t, 2 \right)$, where $\lambda$ is a positive integer and $t \geq 3$. Then $d^*$ attains $u \left( t, n, p, \sigma_{11}, \sigma_{22}, \rho, \boldsymbol{V}_1, \boldsymbol{V}_R \right)$ as given in \autoref{thm4-c4}, where $n=\lambda t \left(t-1 \right)$ and $p=t$, if and only if $tr \left( \boldsymbol{V}_1^* \boldsymbol{\psi} \right) tr \left( \hatmat_p \boldsymbol{\psi}^{'} \boldsymbol{V}_R^* \boldsymbol{\psi} \right) = tr \left( \boldsymbol{V}_R^* \boldsymbol{\psi} \right) tr \left( \hatmat_p \boldsymbol{\psi}^{'} \boldsymbol{V}_1^* \boldsymbol{\psi} \right)$. Here, $\boldsymbol{V}_1^*$ and $\boldsymbol{V}_R^*$ are as defined before.
\label{markov-lemma-1}
\end{remarks}
\begin{proof}
The proof of the above remark is given in \ref{markov-efficiency}.
\end{proof}

\noindent From \autoref{markov-lemma-1}, we note that if $\boldsymbol{V}_1$ is proportional to $\boldsymbol{V}_R$, then the upper bound is attained, implying the attainment of the upper bound for the proportional structure.

\subsubsection{Study of design efficiency for the Markovian covariance structure}
\label{nearly-efficient}
In this section, for various choices for $\boldsymbol{V}_1$ and $\boldsymbol{V}_R$ we evaluate $d^*$ using the upper bound of $tr \left( 
\boldsymbol{C}_{d(s2)} \right)$. For the purpose of simplicity, we denote the upper bound given in \autoref{thm4-c4} by $u$. We list $7$ cases based on different combinations of  $\boldsymbol{V}_1$ and $\boldsymbol{V}_R$ from \autoref{table_markov}:

\begin{minipage}[t]{0.5\textwidth}
\begin{enumerate}
\item[\namedlabel{c1}{Case $1$}.] $\boldsymbol{V}_1$: Mat($0.5$) and $\boldsymbol{V}_R$: Mat($1.5$)
\item[\namedlabel{c2}{Case $2$}.] $\boldsymbol{V}_1$: Mat($0.5$) and $\boldsymbol{V}_R$: Mat($\infty$)
\item[\namedlabel{c3}{Case $3$}.] $\boldsymbol{V}_1$: Mat($1.5$) and $\boldsymbol{V}_R$: Mat($0.5$)
\item[\namedlabel{c4}{Case $4$}.] $\boldsymbol{V}_1$: Mat($1.5$) and $\boldsymbol{V}_R$: Mat($\infty$)
\end{enumerate}
\end{minipage}
\begin{minipage}[t]{0.5\textwidth}
\begin{enumerate}
 \item[\namedlabel{c5}{Case $5$}.] $\boldsymbol{V}_1$: Mat($\infty$) and $\boldsymbol{V}_R$: Mat($0.5$)
 \item[\namedlabel{c6}{Case $6$}.] $\boldsymbol{V}_1$: Mat($\infty$) and $\boldsymbol{V}_R$: Mat($1.5$)
 \item[\namedlabel{c7}{Case $7$}.] $\boldsymbol{V}_1$: Mat($0.5$) and $\boldsymbol{V}_R$: Mat($0.5$)
\end{enumerate}
\end{minipage}

\noindent Note in case $7$,  $\boldsymbol{\mathit{\Sigma}}$ resembles the NS1 structure, where both $\boldsymbol{V}_1$ and $\boldsymbol{V}_R$ have a Mat($0.5$) like structure with correlation parameters $r_{(1)}$ and $r_{(1)}^2$; $0 < r_{(1)}<1$, respectively.\par

\noindent For our illustrations, $p=t \in \{ 3,4 \}$ is used. To apply \autoref{markov-lemma-1}, the expressions for $\boldsymbol{V}_1^*$ and $\boldsymbol{V}_R^*$ are needed. In case $7$ for $p=t=3$, 
\begin{align*}
\boldsymbol{V}_1^* &= \boldsymbol{V}_1^{-1} - \left( \vecone_{p}^{'} \boldsymbol{V}_1^{-1} \vecone_p \right)^{-1} \boldsymbol{V}_1^{-1} \matone_{p \times p} \boldsymbol{V}_1^{-1} \\
&=
 \frac{\sigma_{11}^{-1}}{r_{(1)}^2 -4 r_{(1)}+3}
\begin{bmatrix}
 \frac{2 \left( r_{(1)}^2 -4 r_{(1)}+3 \right)}{r_{(1)}^3-3 r_{(1)}^2-r_{(1)} +3} & -1 &\frac{ r_{(1)}^2 -4 r_{(1)}+3}{r_{(1)}^2 -2 r_{(1)}-3} \\
-1 & 2  & -1 \\
\frac{ r_{(1)}^2 -4 r_{(1)}+3}{r_{(1)}^2 -2 r_{(1)}-3}  & -1 & \frac{2 \left( r_{(1)}^2 -4 r_{(1)}+3 \right)}{r_{(1)}^3-3 r_{(1)}^2-r_{(1)} +3}
\end{bmatrix}
\end{align*}
and
\begin{align*}
\boldsymbol{V}_R^* &= \boldsymbol{V}_R^{-1} - \left( \vecone_{p}^{'} \boldsymbol{V}_R^{-1} \vecone_p \right)^{-1} \boldsymbol{V}_R^{-1} \matone_{p \times p} \boldsymbol{V}_R^{-1}\\
&= 
\frac{1}{r_{(1)}^4 -4 r_{(1)}^2+3}
\begin{bmatrix}
\frac{2 \left( r_{(1)}^4 -4 r_{(1)}^2+3 \right)}{r_{(1)}^6 - 3 r_{(1)}^4 -r_{(1)}^2+3} & -1 & \frac{r_{(1)}^4 -4 r_{(1)}^2+3}{r_{(1)}^4 -2 r_{(1)}^2 - 3}\\
-1 & 2 & -1\\
\frac{r_{(1)}^4 -4 r_{(1)}^2+3}{r_{(1)}^4 -2 r_{(1)}^2 - 3} & -1 & \frac{2 \left( r_{(1)}^4 -4 r_{(1)}^2+3 \right)}{r_{(1)}^6 - 3 r_{(1)}^4 -r_{(1)}^2+3}
\end{bmatrix},
\end{align*}
where $\boldsymbol{V}_1 = \sigma_{11}  
\begin{bmatrix}
1 & r_{(1)} & r_{(1)}^2\\
r_{(1)} & 1 & r_{(1)}\\
r_{(1)}^2 & r_{(1)} & 1
\end{bmatrix}$ and $\boldsymbol{V}_R = \begin{bmatrix}
1 & r_{(1)}^2 & r_{(1)}^4\\
r_{(1)}^2 & 1 & r_{(1)}^2\\
r_{(1)}^4 & r_{(1)}^2 & 1
\end{bmatrix}$. Here, $\sigma_{11} >0$ and $0 < r_{(1)} < 1$. The constraint, $r_{(1)} <1$, is required for the matrix $\boldsymbol{V}_1$ to be positive definite. To check the necessary and sufficient condition from \autoref{markov-lemma-1}, we evaluate $tr \left( \boldsymbol{V}_1^* \boldsymbol{\psi} \right) = -\frac{2\sigma_{11}^{-1}}{r_{(1)}^2 -4 r_{(1)}+3}$, $tr \left( \boldsymbol{V}_R^* \boldsymbol{\psi} \right) = -\frac{2}{r_{(1)}^4 - 4 r_{(1)}^2 + 3}$, $tr \left( \hatmat_p \boldsymbol{\psi}^{'} \boldsymbol{V}_1^* \boldsymbol{\psi} \right) =  \frac{2 \sigma_{11}^{-1} \left(3r_{(1)} +5 \right)}{3 \left( r_{(1)}^3-3 r_{(1)}^2-r_{(1)} +3 \right)}$ and $tr \left( \hatmat_p \boldsymbol{\psi}^{'} \boldsymbol{V}_R^* \boldsymbol{\psi} \right) =$\\$ \frac{2 \left(3r_{(1)}^2 +5 \right)}{3 \left( r_{(1)}^6-3 r_{(1)}^4-r_{(1)}^2 +3 \right)}$. Note, these expressions of $\boldsymbol{V}_1^*$, $\boldsymbol{V}_R^*$ and the traces have been cross checked by using MATLAB \citep{MATLAB}. Since $tr \left( \boldsymbol{V}_1^* \boldsymbol{\psi} \right) tr \left( \hatmat_p \boldsymbol{\psi}^{'} \boldsymbol{V}_R^* \boldsymbol{\psi} \right)$ and $tr \left( \boldsymbol{V}_R^* \boldsymbol{\psi} \right) \times$\\$ tr \left( \hatmat_p \boldsymbol{\psi}^{'} \boldsymbol{V}_1^* \boldsymbol{\psi} \right)$ are unequal, we may conclude that for $0< r_{(1)} < 1$, $d^*$ does not attain the upper bound of $tr \left( \boldsymbol{{C}}_{d(s2)} \right)$ over $d \in \mathcal{D}^{(1)}_{t,n=\lambda t (t-1),p=t}$. We find that similar results hold for the other cases when $p=t=3$, and for all the cases when $p=t=4$. Thus we may state that for $p=t \in \{3,4\}$, and all cases of $\boldsymbol{V}_1$ and $\boldsymbol{V}_R$ considered,  $d^*$ fails to attain the upper bound $u$.\par

\noindent We then evaluate the performance of $d^*$ for the various combinations of $\boldsymbol{V}_1$ and $\boldsymbol{V}_R$ by the relative difference between $tr \left( \boldsymbol{{C}}_{d^*(s2)} \right)$ and the upper bound $u$. The performance of $d^*$ is also compared with a uniform and a balanced uniform design having the same number of subjects as $d^*$. Note that a uniform and balanced uniform design for $p=t$ are also binary designs and hence belong to $\mathcal{D}^{(1)}_{t,n=\lambda t (t-1),p=t}$.\par

\noindent For any design $d^{(0)} \in \mathcal{D}^{(1)}_{t,n=\lambda t (t-1),p=t}$, where $\lambda$ is a positive integer and $t \geq 3$, $RD_{d^{(0)}}$, the relative difference between $tr \left( \boldsymbol{{C}}_{d^{(0)}(s2)} \right)$ and $u$ is defined as
\begin{align}
RD_{d^{(0)}} &= 1 - \frac{tr \left( \boldsymbol{{C}}_{d^{(0)}(s2)} \right)}{ u }.
\label{markov-t21c-1}
\end{align}
Here, $RD_{d^{(0)}}$ is a function of $t$, $n$, $p$, $\sigma_{11}$, $\sigma_{22}$, $\rho$, $\boldsymbol{V}_1$ and $\boldsymbol{V}_R$, and its values lie between $0$ and $1$. For $RD_{d^{(0)}} =0$, the upper bound is attained by the design $d^{(0)}$, thus proving $d^{(0)}$ is an efficient design. Hence, lower values of $RD_{d^{(0)}}$ are preferred as they indicate that $tr \left( \boldsymbol{{C}}_{d^{(0)}(s2)} \right)$ is closer to the upper bound $u$, and $d^{(0)}$ is a highly efficient design.\par

\noindent Let $d_1$ be a uniform design and $d_2$ be a balanced uniform design. For illustration purpose, when $p=t=3$, we consider
\begin{align*}
d_1 = 
\begin{matrix}
1 & 2 & 3 & 1 & 2 & 3\\
2 & 3 & 1 & 2 & 3 & 1\\
3 & 1 & 2 & 3 & 1 & 2
\end{matrix}
\quad \text{and} \quad 
d^* =
\begin{matrix}
1 & 2 & 3 & 1 & 2 & 3\\
2 & 3 & 1 & 3 & 1 & 2\\
3 & 1 & 2 & 2 & 3 & 1
\end{matrix}.
\end{align*}
While for $p=t=4$, we consider
\begin{align*}
d_1 =
\begin{matrix}
1 & 4 & 3 & 2 & 1 & 4 & 3 & 2 & 1 & 4 & 3 & 2\\
2 & 1 & 4 & 3 & 2 & 1 & 4 & 3 & 2 & 1 & 4 & 3\\
3 & 2 & 1 & 4 & 3 & 2 & 1 & 4 & 3 & 2 & 1 & 4\\
4 & 3 & 2 & 1 & 4 & 3 & 2 & 1 & 4 & 3 & 2 & 1
\end{matrix}, \quad 
d_2 =
\begin{matrix}
1 & 2 & 3 & 4 & 1 & 2 & 3 & 4 & 1 & 2 & 3 & 4\\
4 & 1 & 2 & 3 & 4 & 1 & 2 & 3 & 4 & 1 & 2 & 3\\
2 & 3 & 4 & 1 & 2 & 3 & 4 & 1 & 2 & 3 & 4 & 1\\
3 & 4 & 1 & 2 & 3 & 4 & 1 & 2 & 3 & 4 & 1 & 2
\end{matrix} 
\end{align*}
and
\begin{align*}
d^* =
\begin{matrix}
1 & 1 & 1 & 2 & 2 & 2 & 3 & 3 & 3 & 4 & 4 & 4\\
2 & 3 & 4 & 1 & 3 & 4 & 1 & 2 & 4 & 1 & 2 & 3\\
3 & 4 & 2 & 4 & 1 & 3 & 2 & 4 & 1 & 3 & 1 & 2\\
4 & 2 & 3 & 3 & 4 & 1 & 4 & 1 & 2 & 2 & 3 & 1
\end{matrix}. 
\end{align*}
The balanced uniform design $d_2$ is constructed using the method described by \citet{Bose2009OptimalDesigns}. For $p=t=3$, a balanced uniform design with the same number of subjects as $d^*$ is also an orthogonal array, so it is not considered as a candidate for comparison with $d^*$. A graphical procedure is used to evaluate the designs $d^{(0)} \in \{d_1, d_2, d^*\}$. In our illustrations we assume $\sigma_{11} = \sigma_{22}$. The next remark shows that if $\sigma_{11} = \sigma_{22}$, then for $p=t \geq 3$, the relative difference does not depend on the value of $\sigma_{11}$.\par

\begin{remarks}
\label{markov-lemma-sigma}
Let $d^{(0)} \in \mathcal{D}^{(1)}_{t,n=\lambda t \left( t-1 \right),p=t}$ be a design, where $\lambda$ is a positive integer and $p=t \geq 3$. Then for $\sigma_{11} = \sigma_{22}$, $RD_{d^{(0)}}$ is independent of the value of $\sigma_{11}$, where $RD_{d^{(0)}}$ is as given in \eqref{markov-t21c-1}.
\end{remarks}
\begin{proof}
Details are given in \ref{markov-efficiency}.
\end{proof}

\noindent From our choices of $\boldsymbol{V}_1$ and $\boldsymbol{V}_R$, it is clear that a $p \times p$ matrix $\boldsymbol{V}_1$ is a function of $\sigma_{11}$ and $r_{(1)}$, while a $p \times p$ matrix $\boldsymbol{V}_R$ is only a function of $r_{(1)}$. Thus from \autoref{markov-lemma-sigma}, we can see that for any $d^{(0)} \in \mathcal{D}^{(1)}_{t,n=\lambda t \left( t-1 \right),p=t}$, $p=t \geq 3$ and $\sigma_{11} = \sigma_{22}$, the relative difference given as $RD_{d^{(0)}}$ is a function of $d^{(0)}$, $p$, $n$, $\rho$ and $r_{(1)}$.\par

\noindent For $p=t \in \{ 3,4 \}$, $n=t \left(t-1 \right)$, $\sigma_{11} = \sigma_{22}$,  and $d^{(0)} \in \{ d_1, d_2, d^* \}$, we plot $min_{0< |\rho | <1} ~ RD_{d^{(0)}}$ and $max_{0< |\rho | <1} ~ RD_{d^{(0)}}$, corresponding to various positive permissible values of $r_{(1)}$. Here by permissible values, we mean those values for which the matrices $\boldsymbol{V}_1$ and $\boldsymbol{V}_R$ are positive definite. If in any case $min_{0< |\rho | <1} ~ RD_{d^{(0)}}$ attains $0$, then this implies that design $d^{(0)}$ maximizes $tr \left( \boldsymbol{{C}}_{d^{(0)}(s2)} \right)$ and  $d^{(0)}$ is a efficient design. If $min_{0< |\rho | <1} ~ RD_{d^{(0)}}$ does not attain $0$, we plot $max_{0< |\rho | <1} ~ RD_{d^{(0)}}$ to get an understanding of the worst scenario and quantify the departure of the design under consideration from an efficient design. Lower values of $max_{0< |\rho | <1} ~ RD_{d^{(0)}}$ implies better performance of the design in terms of efficiency. For $d^{(0)} \in \{ d_1, d^* \}$, where $p=t=3$ and $n=6$, \autoref{min-max-p3t3-1} and \autoref{min-max-p3t3-2} display the plots of $min_{0< |\rho | <1} ~ RD_{d^{(0)}}$ and $max_{0< |\rho | <1} ~ RD_{d^{(0)}}$, for cases $1$-$7$. While, for $p=t=4$, $n=12$, \autoref{min-max-p4t4-1} and \autoref{min-max-p4t4-2} show the design performances.\par

\noindent From \autoref{min-max-p3t3-1} and \autoref{min-max-p3t3-2}, for $p=t=3$ and $n=6$, we see that $min_{0< |\rho | <1} ~ RD_{d_1}$ never attains $0$ and hence $d_1$, like $d^*$, fails to attain the upper bound $u$. Also, we note that for cases $1$-$7$,  the values of $max_{0< |\rho | <1} ~ RD_{d_1}>>max_{0< |\rho | <1} ~ RD_{d^*}$, implying that $d^*$ performs much better when compared to $d_1$. Considering all seven cases, the maximum value of $RD_{d^*}$ over $\rho \neq 0$ is at most around $0.39\%$ indicating that $d^*$ is a highly efficient design. In the $p=t=4$ and $n=12$ case, we observe from \autoref{min-max-p4t4-1} and \autoref{min-max-p4t4-2} that for none of the designs $d^{(0)} \in \{ d_1, d_2, d^* \}$, $min_{0< |\rho | <1} ~ RD_{d^{(0)}}$ attain $0$ and  $max_{0< |\rho | <1} ~ RD_{d^*} < max_{0< |\rho | <1} ~ RD_{d_2} < max_{0< |\rho | <1} ~ RD_{d_1}$. So $d^*$ performs the best among the three designs being compared, followed by $d_2$. Since $max_{0< |\rho | <1} ~ RD_{d_1}$ is far away from $0$, we reject the candidature of $d_1$ as a highly efficient design. For some values of $r_{(1)}$, we note $max_{0< |\rho | <1} ~ RD_{d_2}$ is quite close to $max_{0< |\rho | <1} ~ RD_{d^*}$. The value of $max_{0< |\rho | <1} ~ RD_{d^*}$ is found to be at most $2.02\%$, while most of the values of $max_{0< |\rho | <1} ~ RD_{d_2}$ are much larger. Thus, we may conclude that $d^*$ is the best design among the three designs being compared and also is a highly efficient design. 

\begin{figure}
\centering
\begin{subfigure}{.5\textwidth}
  \centering
  \includegraphics[height=0.4\linewidth, width=1.0\linewidth]{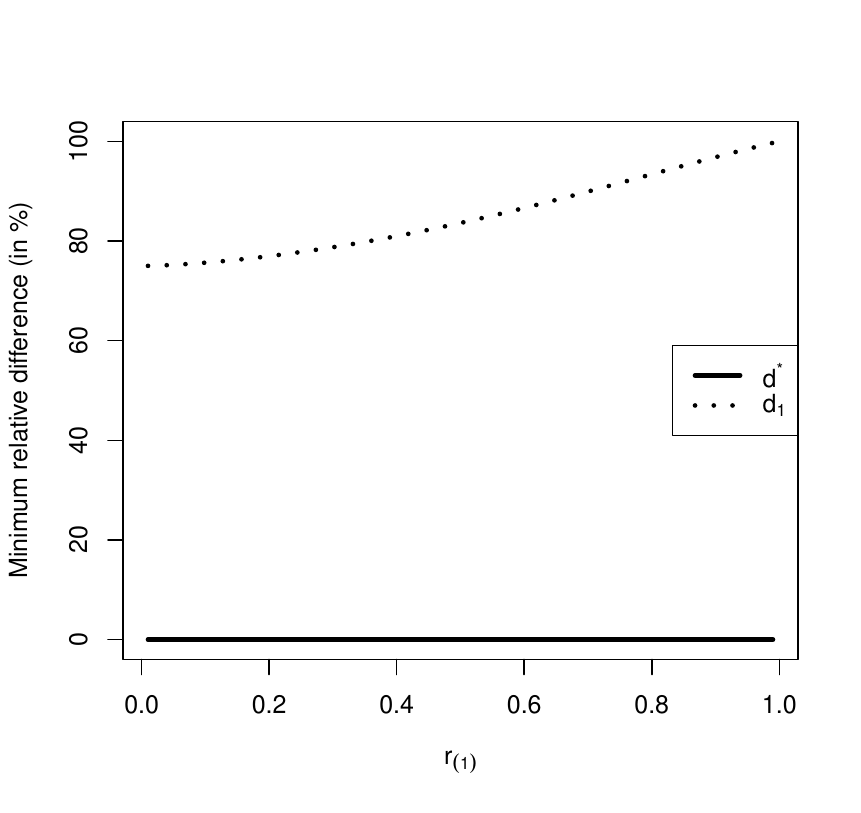}
  \caption{$min_{0< |\rho | <1}~ RD_{d^{(0)}}$ for $d^{(0)} \in \{ d_1, d^* \}$ \\($\boldsymbol{V}_1$: Mat($0.5$), $\boldsymbol{V}_R$: Mat($1.5$))}
  \label{fig:minimum_case1_p3t3}
\end{subfigure}%
\begin{subfigure}{.5\textwidth}
  \centering
  \includegraphics[height=0.4\linewidth, width=1.0\linewidth]{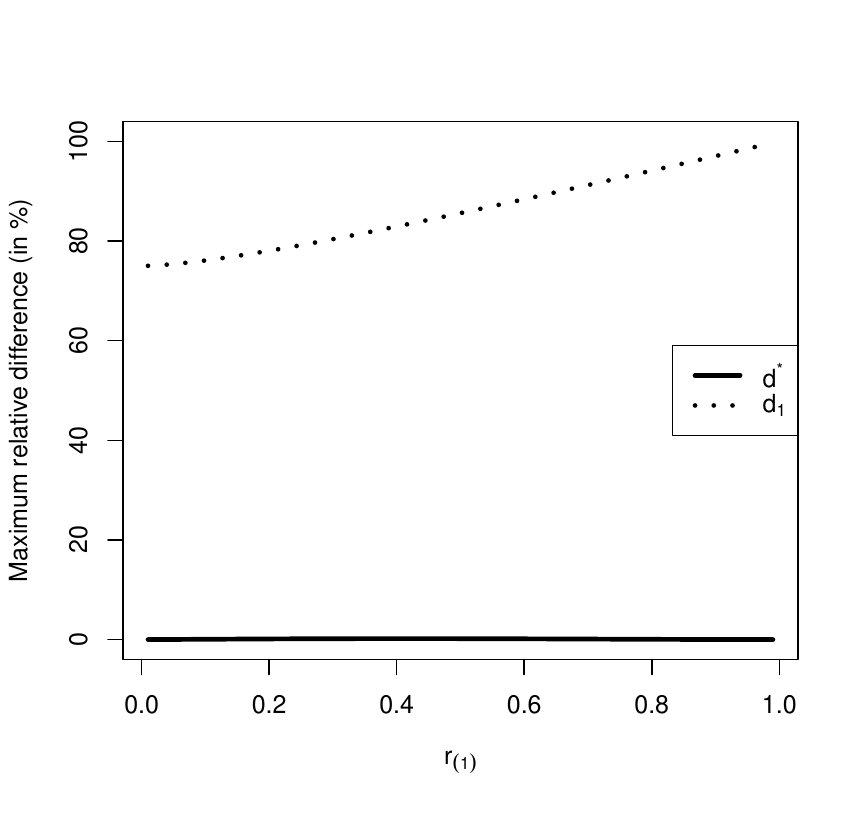}
  \caption{$max_{0< |\rho | <1}~ RD_{d^{(0)}}$ for $d^{(0)} \in \{ d_1, d^* \}$ \\($\boldsymbol{V}_1$: Mat($0.5$), $\boldsymbol{V}_R$: Mat($1.5$))}
  \label{fig:maximum_case1_p3t3}
\end{subfigure}\\
\begin{subfigure}{.5\textwidth}
  \centering
  \includegraphics[height=0.4\linewidth, width=1.0\linewidth]{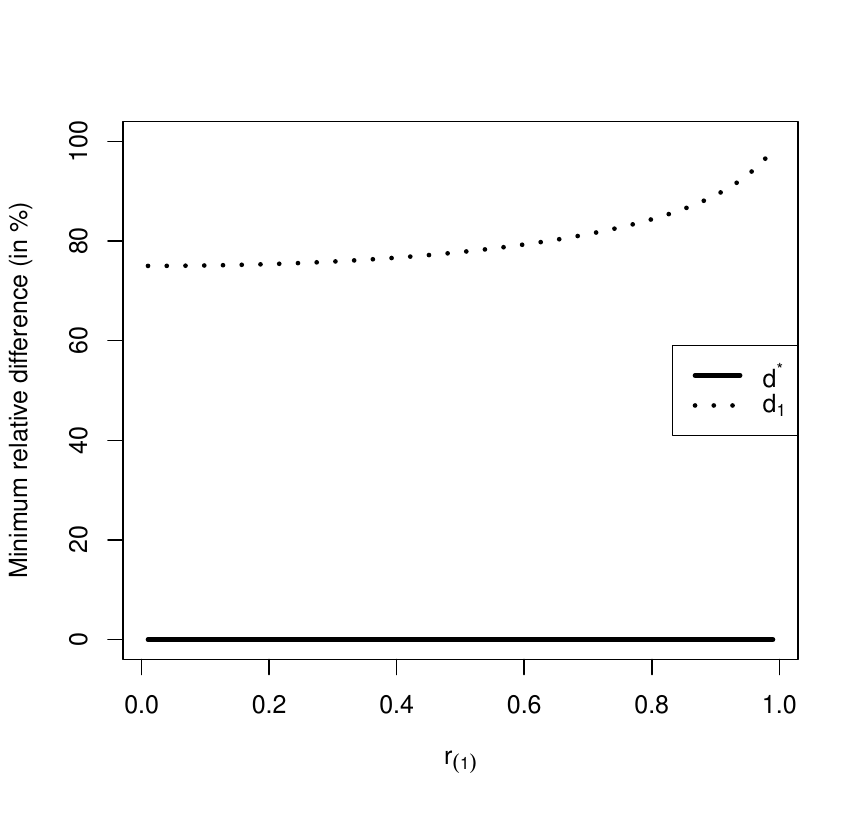}
  \caption{$min_{0< |\rho | <1}~ RD_{d^{(0)}}$ for $d^{(0)} \in \{ d_1, d^* \}$ \\($\boldsymbol{V}_1$: Mat($0.5$), $\boldsymbol{V}_R$: Mat($\infty$))}
  \label{fig:minimum_case2_p3t3}
\end{subfigure}%
\begin{subfigure}{.5\textwidth}
  \centering
  \includegraphics[height=0.4\linewidth, width=1.0\linewidth]{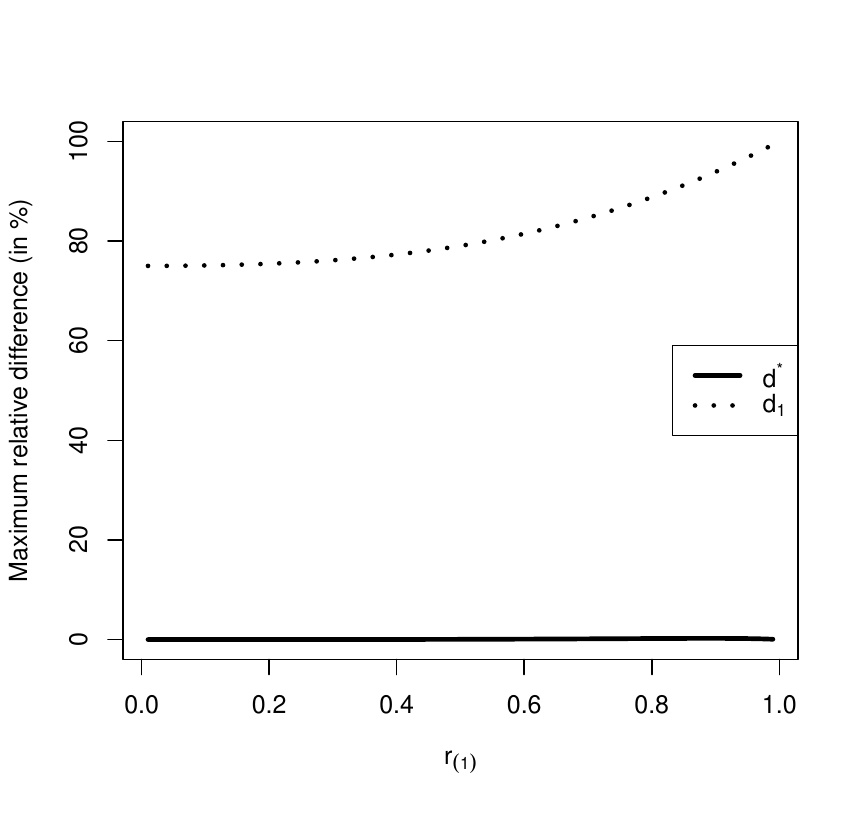}
  \caption{$max_{0< |\rho | <1}~ RD_{d^{(0)}}$ for $d^{(0)} \in \{ d_1, d^* \}$ \\($\boldsymbol{V}_1$: Mat($0.5$), $\boldsymbol{V}_R$: Mat($\infty$))}
  \label{fig:maximum_case2_p3t3}
\end{subfigure}\\
\begin{subfigure}{.5\textwidth}
  \centering
  \includegraphics[height=0.4\linewidth, width=1.0\linewidth]{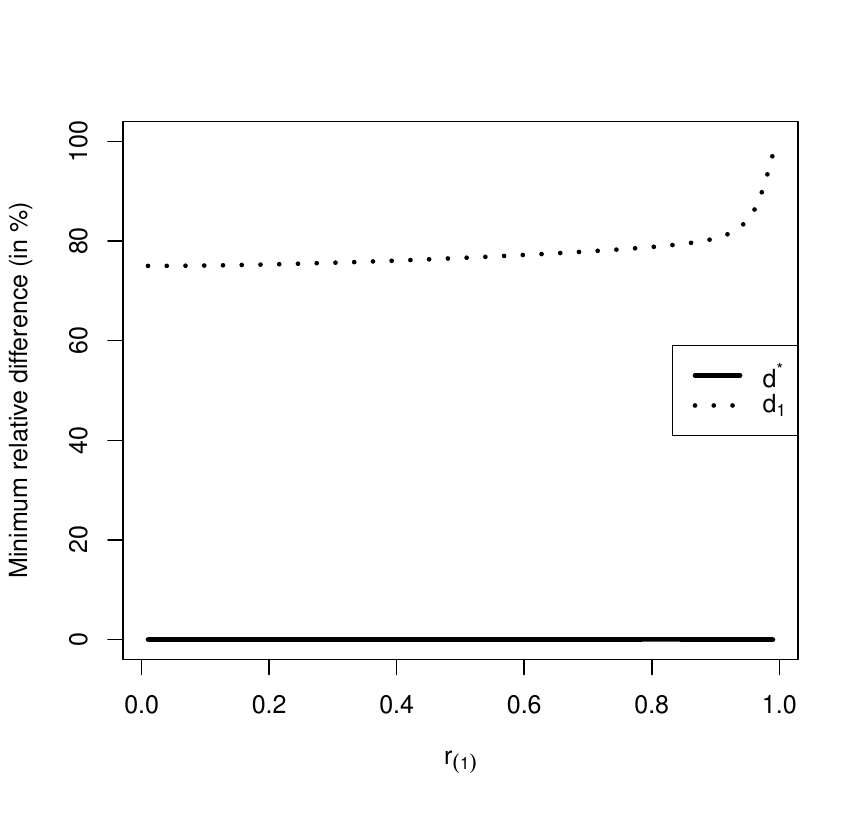}
  \caption{$min_{0< |\rho | <1}~ RD_{d^{(0)}}$ for $d^{(0)} \in \{ d_1, d^* \}$ \\($\boldsymbol{V}_1$: Mat($1.5$), $\boldsymbol{V}_R$: Mat($0.5$))}
  \label{fig:minimum_case3_p3t3}
\end{subfigure}%
\begin{subfigure}{.5\textwidth}
  \centering
  \includegraphics[height=0.4\linewidth, width=1.0\linewidth]{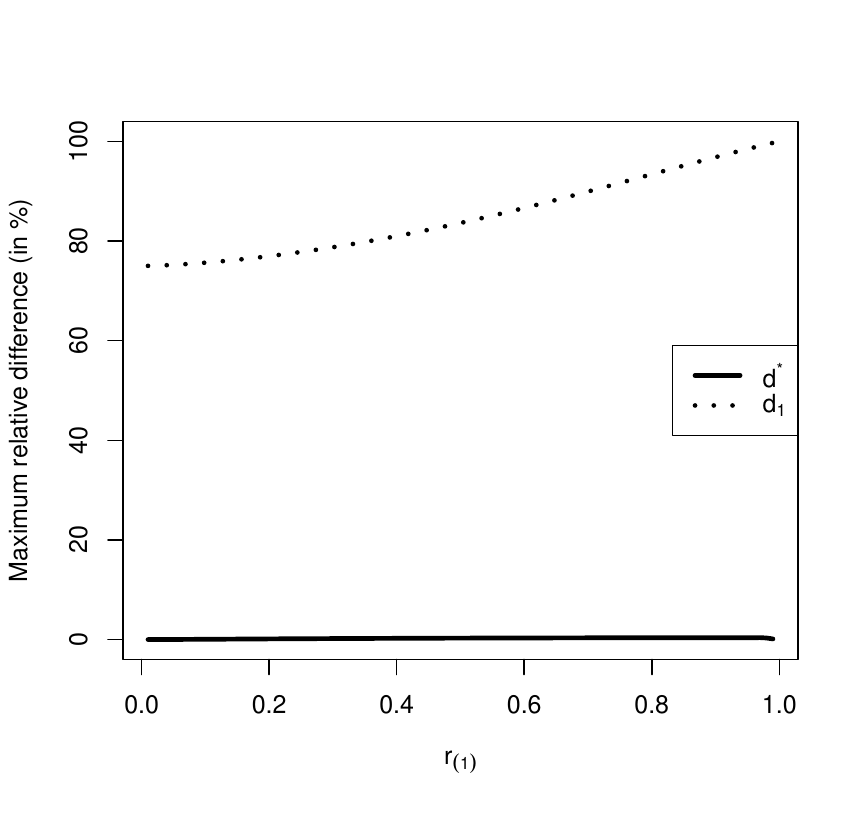}
  \caption{$max_{0< |\rho | <1}~ RD_{d^{(0)}}$ for $d^{(0)} \in \{ d_1, d^* \}$ \\($\boldsymbol{V}_1$: Mat($1.5$), $\boldsymbol{V}_R$: Mat($0.5$))}
  \label{fig:maximum_case3_p3t3}
\end{subfigure}\\
\begin{subfigure}{.5\textwidth}
  \centering
  \includegraphics[height=0.4\linewidth, width=1.0\linewidth]{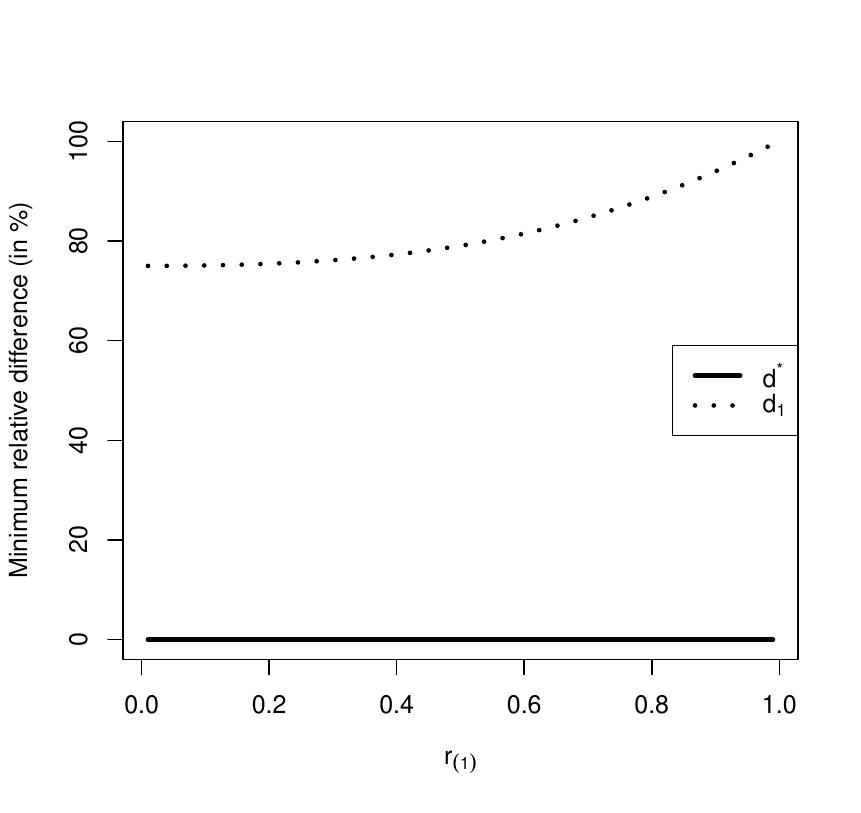}
  \caption{$min_{0< |\rho | <1}~ RD_{d^{(0)}}$ for $d^{(0)} \in \{ d_1, d^* \}$ \\($\boldsymbol{V}_1$: Mat($1.5$), $\boldsymbol{V}_R$: Mat($\infty$))}
  \label{fig:minimum_case4_p3t3}
\end{subfigure}%
\begin{subfigure}{.5\textwidth}
  \centering
  \includegraphics[height=0.4\linewidth, width=1.0\linewidth]{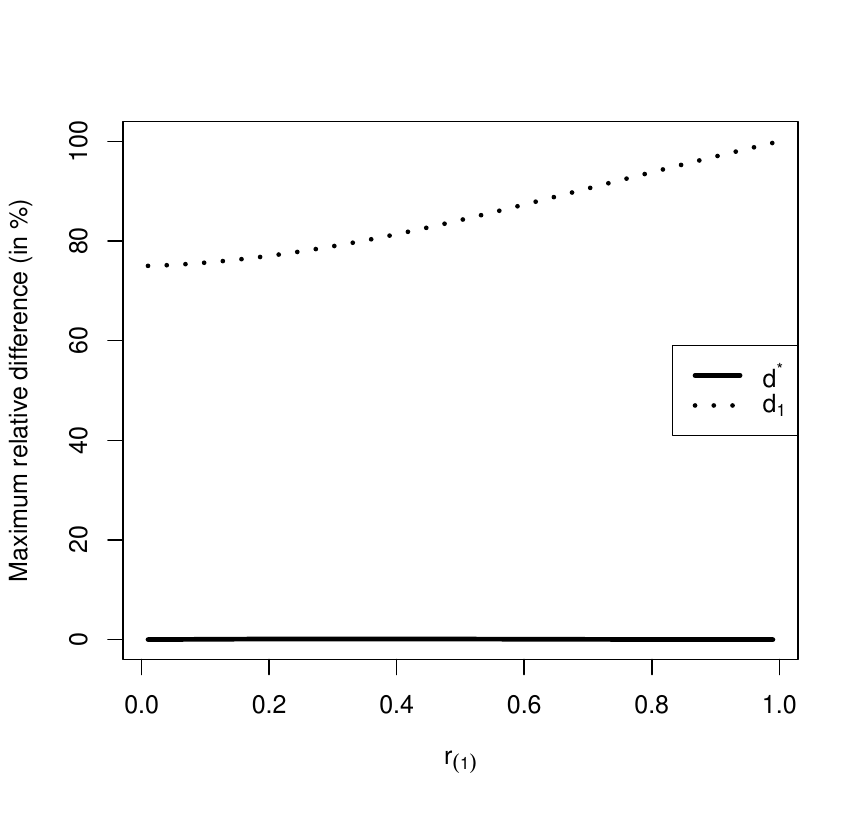}
  \caption{$max_{0< |\rho | <1}~ RD_{d^{(0)}}$ for $d^{(0)} \in \{ d_1, d^* \}$ \\($\boldsymbol{V}_1$: Mat($1.5$), $\boldsymbol{V}_R$: Mat($\infty$))}
  \label{fig:maximum_case4_p3t3}
\end{subfigure}
\caption{Plots of $min_{0< |\rho | <1}~ RD_{d^{(0)}}$ and $max_{0< |\rho | <1}~ RD_{d^{(0)}}$ for $d^{(0)} \in \{ d_1, d^* \}$,  $p=t=3$, $n=6$ and $\sigma_{11} = \sigma_{22}$.}
\label{min-max-p3t3-1}
\end{figure}

\begin{figure}
\centering
\begin{subfigure}{.5\textwidth}
  \centering
  \includegraphics[height=0.4\linewidth, width=1.0\linewidth]{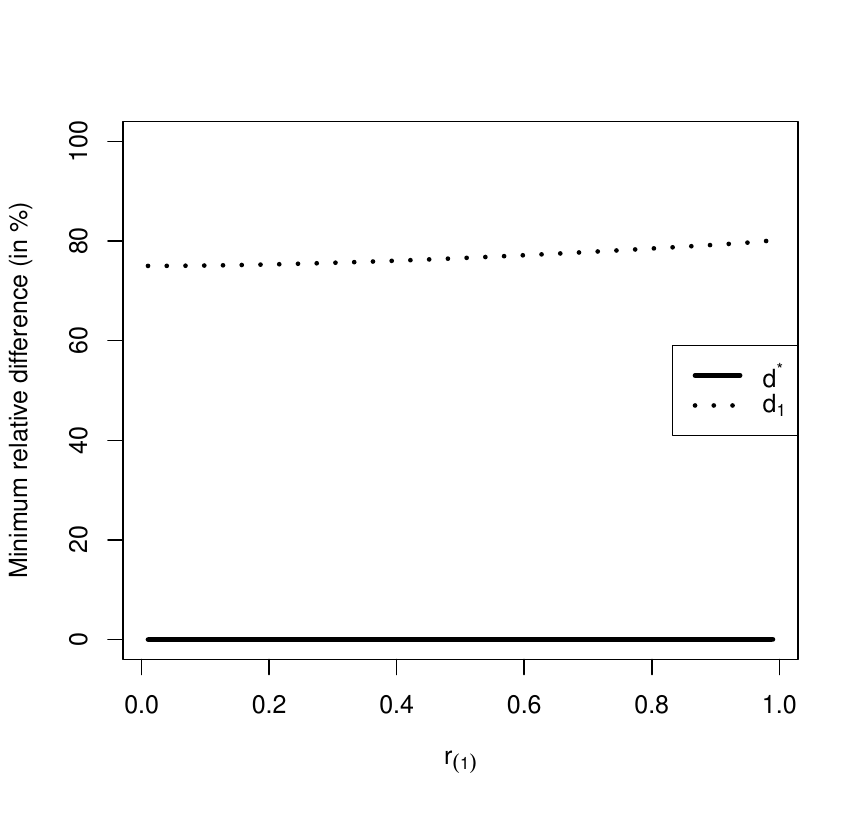}
  \caption{$min_{0< |\rho | <1}~ RD_{d^{(0)}}$ for $d^{(0)} \in \{ d_1, d^* \}$ \\($\boldsymbol{V}_1$: Mat($\infty$), $\boldsymbol{V}_R$: Mat($0.5$))}
  \label{fig:minimum_case5_p3t3}
\end{subfigure}%
\begin{subfigure}{.5\textwidth}
  \centering
  \includegraphics[height=0.4\linewidth, width=1.0\linewidth]{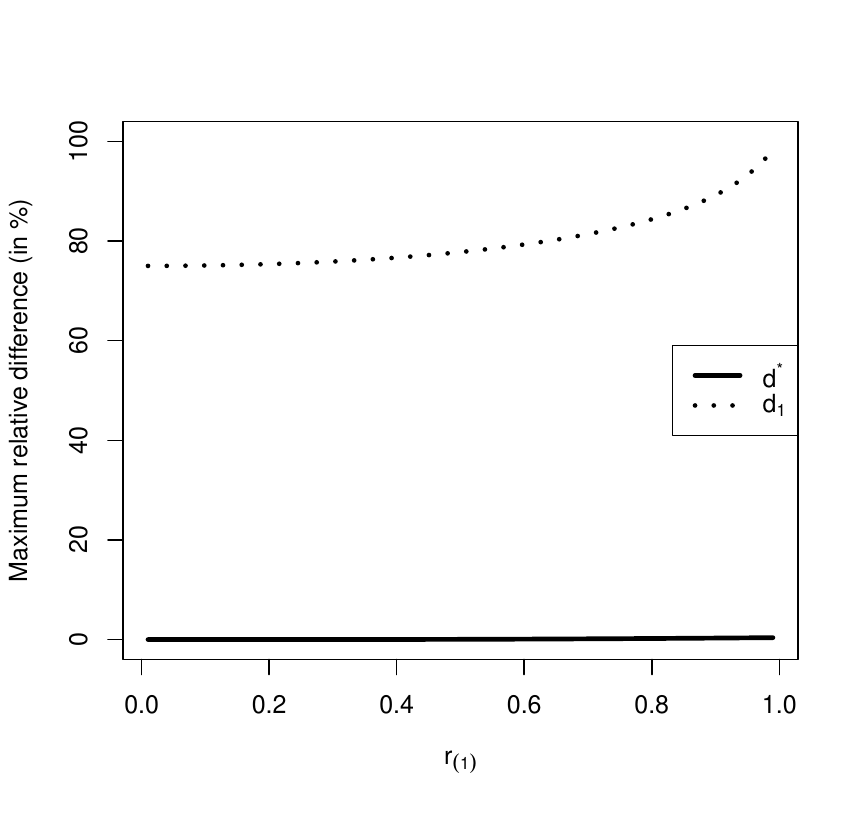}
  \caption{$max_{0< |\rho | <1}~ RD_{d^{(0)}}$ for $d^{(0)} \in \{ d_1, d^* \}$ \\($\boldsymbol{V}_1$: Mat($\infty$), $\boldsymbol{V}_R$: Mat($0.5$))}
  \label{fig:maximum_case5_p3t3}
\end{subfigure}\\
\begin{subfigure}{.5\textwidth}
  \centering
  \includegraphics[height=0.4\linewidth, width=1.0\linewidth]{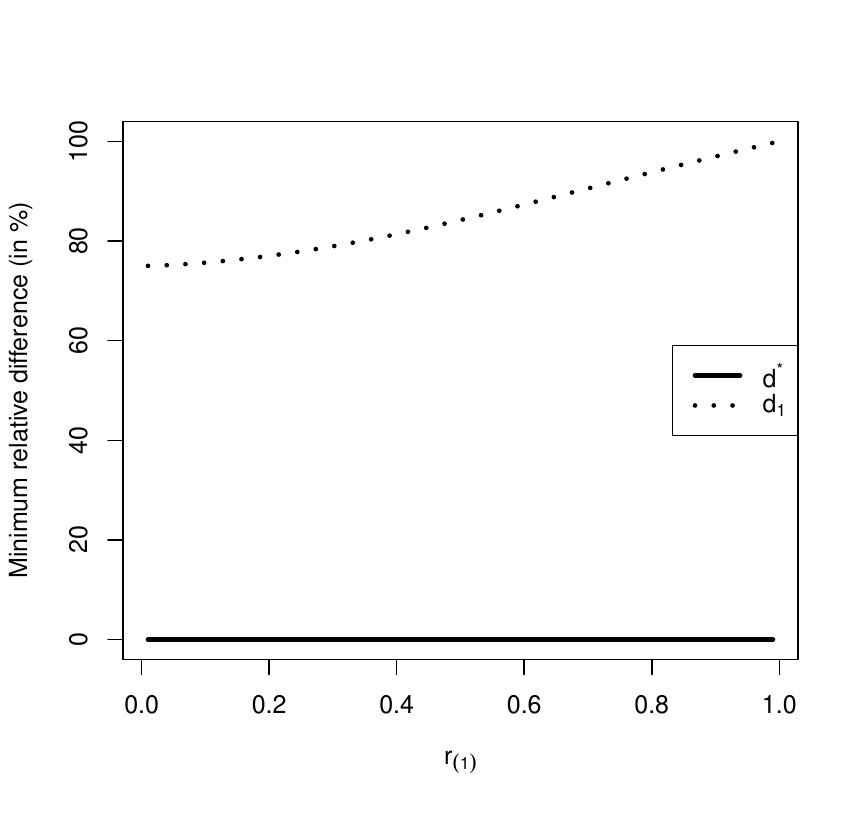}
  \caption{$min_{0< |\rho | <1}~ RD_{d^{(0)}}$ for $d^{(0)} \in \{ d_1, d^* \}$ \\($\boldsymbol{V}_1$: Mat($\infty$), $\boldsymbol{V}_R$: Mat($1.5$))}
  \label{fig:minimum_case6_p3t3}
\end{subfigure}%
\begin{subfigure}{.5\textwidth}
  \centering
  \includegraphics[height=0.4\linewidth, width=1.0\linewidth]{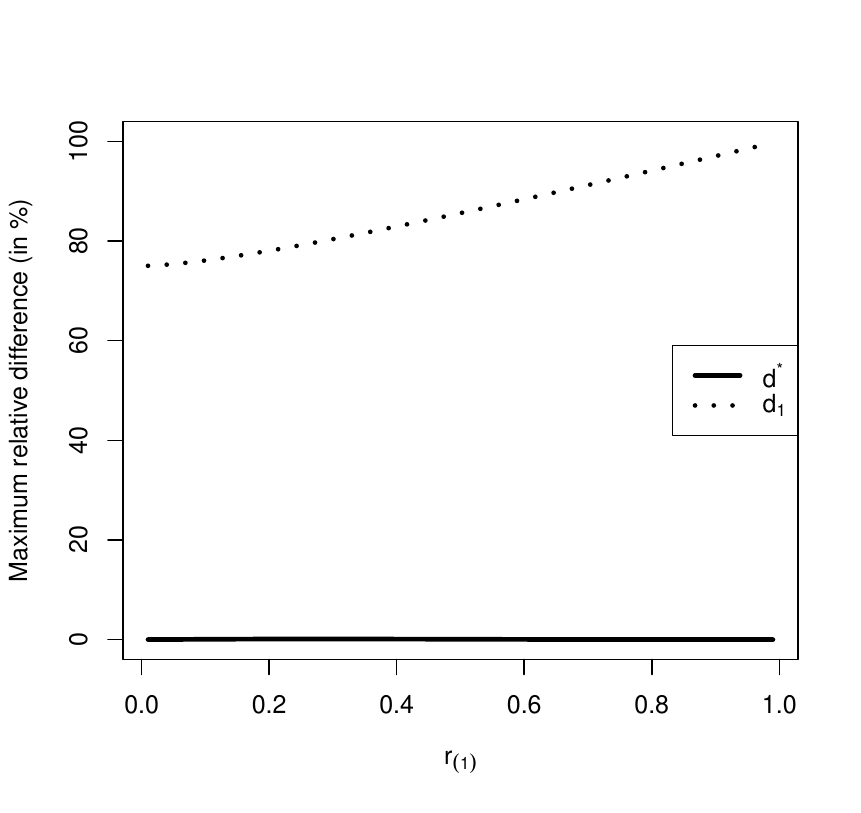}
  \caption{$max_{0< |\rho | <1}~ RD_{d^{(0)}}$ for $d^{(0)} \in \{ d_1, d^* \}$ \\($\boldsymbol{V}_1$: Mat($\infty$), $\boldsymbol{V}_R$: Mat($1.5$))}
  \label{fig:maximum_case6_p3t3}
\end{subfigure}\\
\begin{subfigure}{.5\textwidth}
  \centering
  \includegraphics[height=0.4\linewidth, width=1.0\linewidth]{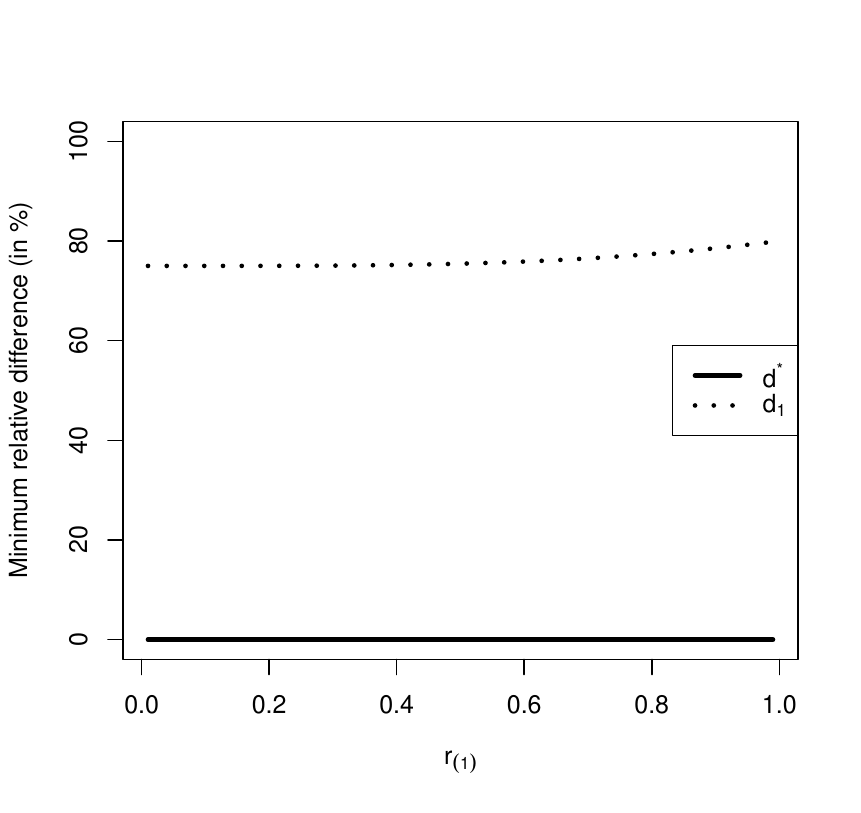}
  \caption{$min_{0< |\rho | <1}~ RD_{d^{(0)}}$ for $d^{(0)} \in \{ d_1, d^* \}$ \\($\boldsymbol{V}_1$: Mat($0.5$), $\boldsymbol{V}_R$: Mat($0.5$))}
  \label{fig:minimum_case7_p3t3}
\end{subfigure}%
\begin{subfigure}{.5\textwidth}
  \centering
  \includegraphics[height=0.4\linewidth, width=1.0\linewidth]{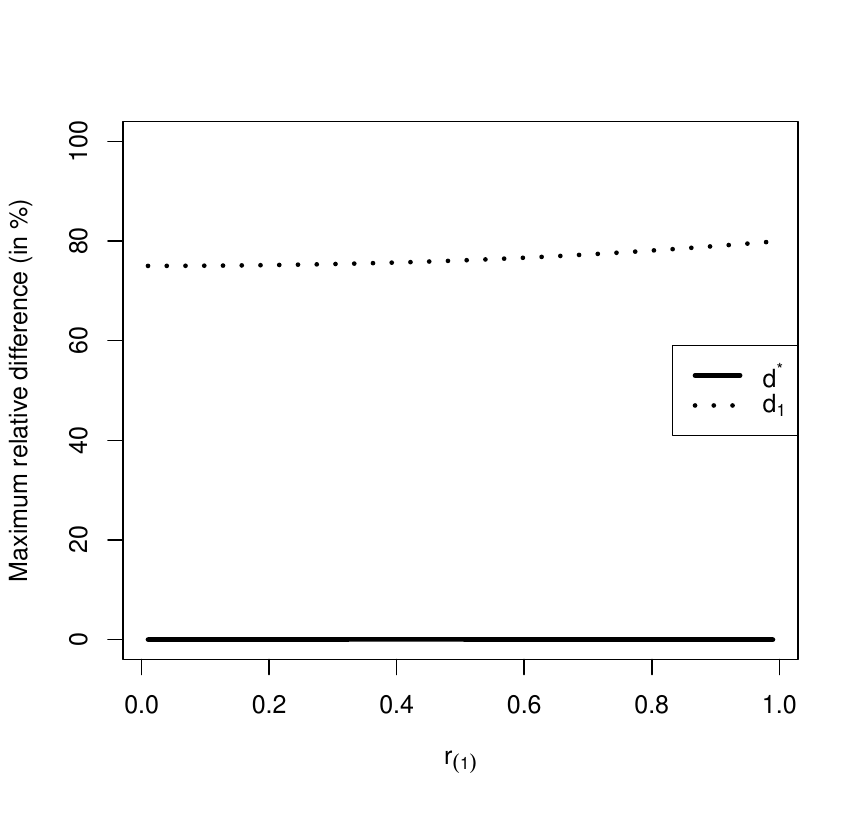}
  \caption{$max_{0< |\rho | <1}~ RD_{d^{(0)}}$ for $d^{(0)} \in \{ d_1, d^* \}$ \\($\boldsymbol{V}_1$: Mat($0.5$), $\boldsymbol{V}_R$: Mat($0.5$))}
  \label{fig:maximum_case7_p3t3}
\end{subfigure}
\caption{Plots of $min_{0< |\rho | <1}~ RD_{d^{(0)}}$ and $max_{0< |\rho | <1}~ RD_{d^{(0)}}$ for $d^{(0)} \in \{ d_1, d^* \}$, $p=t=3$, $n=6$ and $\sigma_{11} = \sigma_{22}$.}
\label{min-max-p3t3-2}
\end{figure}

\begin{figure}
\centering
\begin{subfigure}{.5\textwidth}
  \centering
  \includegraphics[height=0.4\linewidth, width=1.0\linewidth]{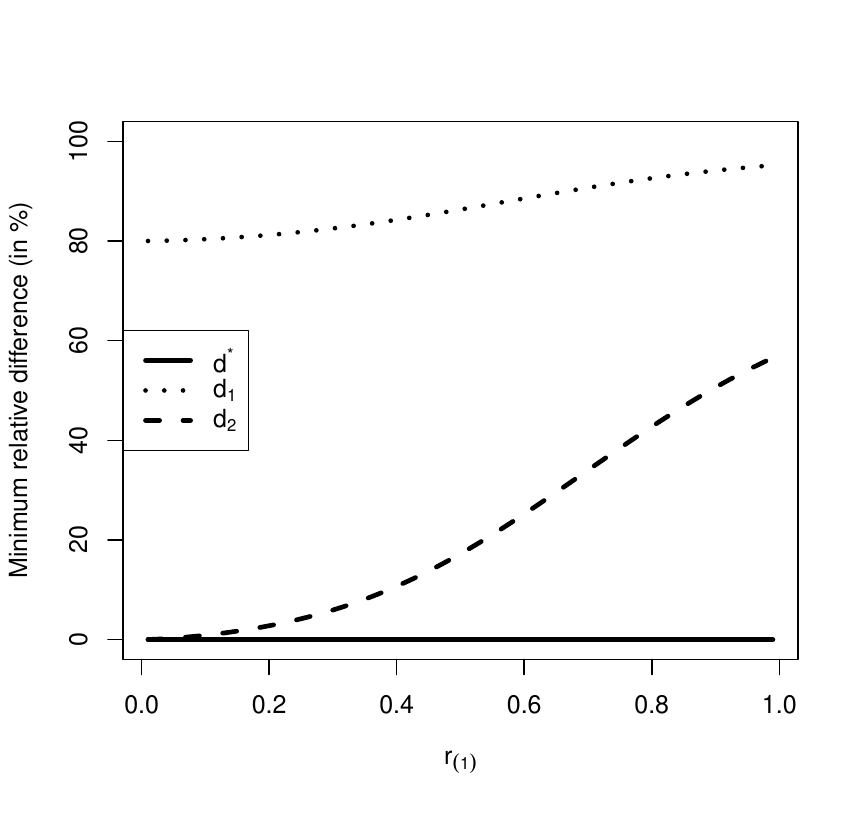}
  \caption{$min_{0< |\rho | <1}~ RD_{d^{(0)}}$ for $d^{(0)} \in \{ d_1, d_2, d^* \}$ \\($\boldsymbol{V}_1$: Mat($0.5$), $\boldsymbol{V}_R$: Mat($1.5$))}
  \label{fig:minimum_case1_p4t4}
\end{subfigure}%
\begin{subfigure}{.5\textwidth}
  \centering
  \includegraphics[height=0.4\linewidth, width=1.0\linewidth]{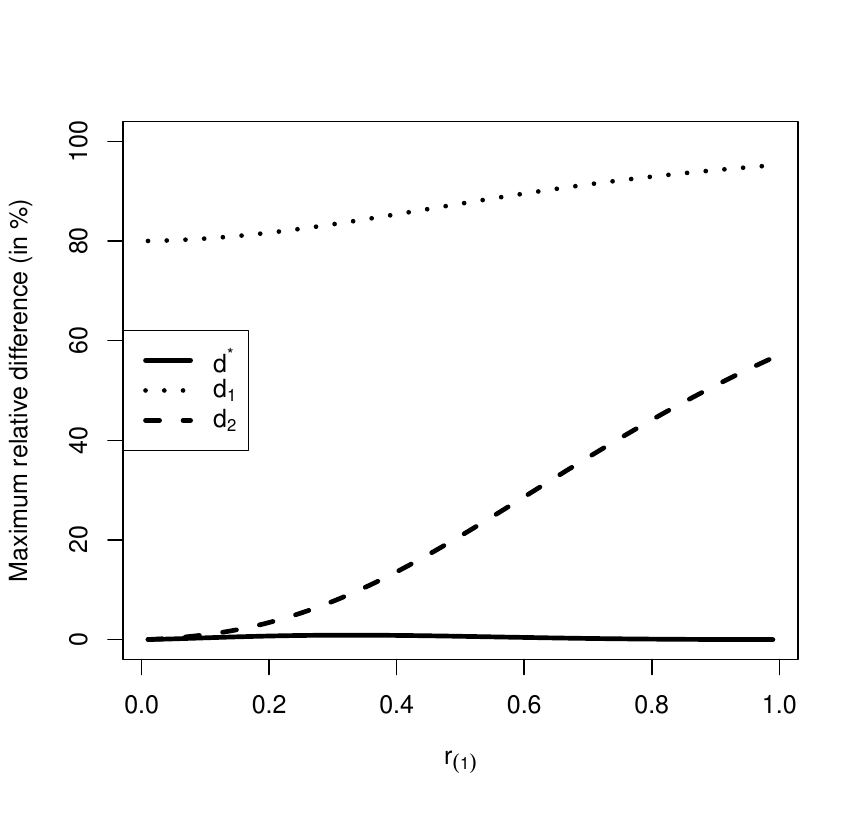}
  \caption{$max_{0< |\rho | <1}~ RD_{d^{(0)}}$ for $d^{(0)} \in \{ d_1, d_2, d^* \}$ \\($\boldsymbol{V}_1$: Mat($0.5$), $\boldsymbol{V}_R$: Mat($1.5$))}
  \label{fig:maximum_case1_p4t4}
\end{subfigure}\\
\begin{subfigure}{.5\textwidth}
  \centering
  \includegraphics[height=0.4\linewidth, width=1.0\linewidth]{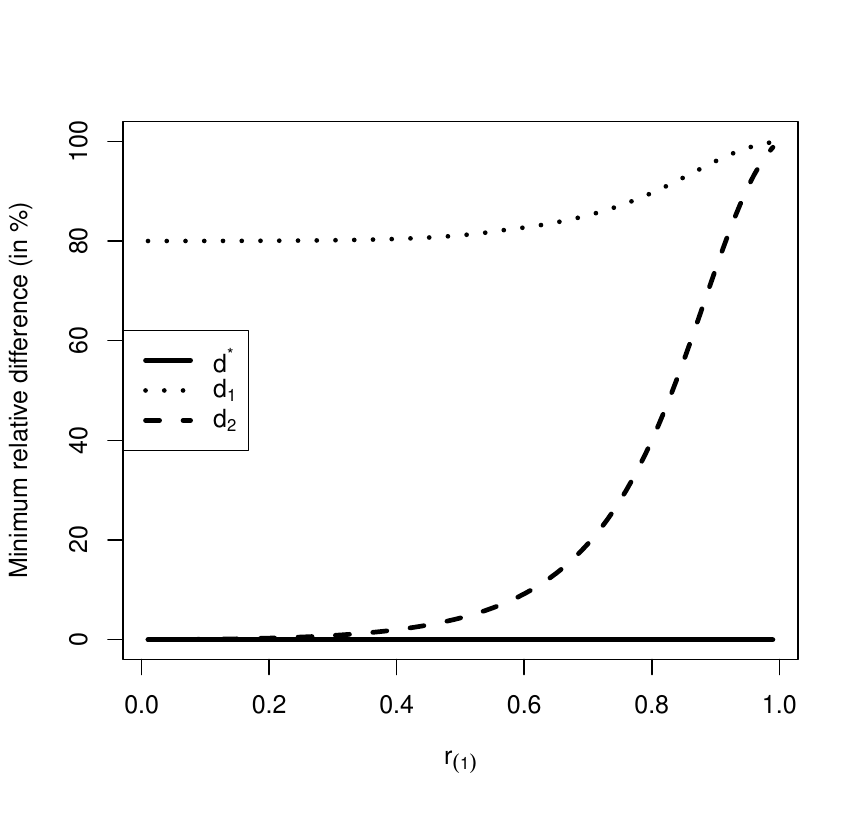}
  \caption{$min_{0< |\rho | <1}~ RD_{d^{(0)}}$ for $d^{(0)} \in \{ d_1, d_2, d^* \}$ \\($\boldsymbol{V}_1$: Mat($0.5$), $\boldsymbol{V}_R$: Mat($\infty$))}
  \label{fig:minimum_case2_p4t4}
\end{subfigure}%
\begin{subfigure}{.5\textwidth}
  \centering
  \includegraphics[height=0.4\linewidth, width=1.0\linewidth]{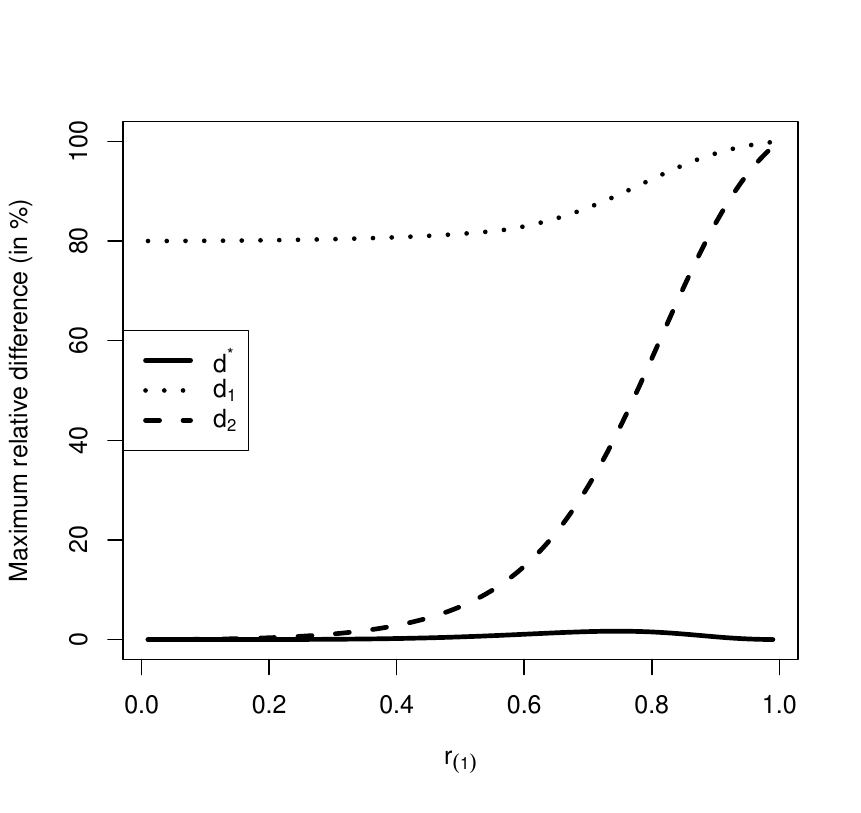}
  \caption{$max_{0< |\rho | <1}~ RD_{d^{(0)}}$ for $d^{(0)} \in \{ d_1, d_2, d^* \}$ \\($\boldsymbol{V}_1$: Mat($0.5$), $\boldsymbol{V}_R$: Mat($\infty$))}
  \label{fig:maximum_case2_p4t4}
\end{subfigure}\\
\begin{subfigure}{.5\textwidth}
  \centering
  \includegraphics[height=0.4\linewidth, width=1.0\linewidth]{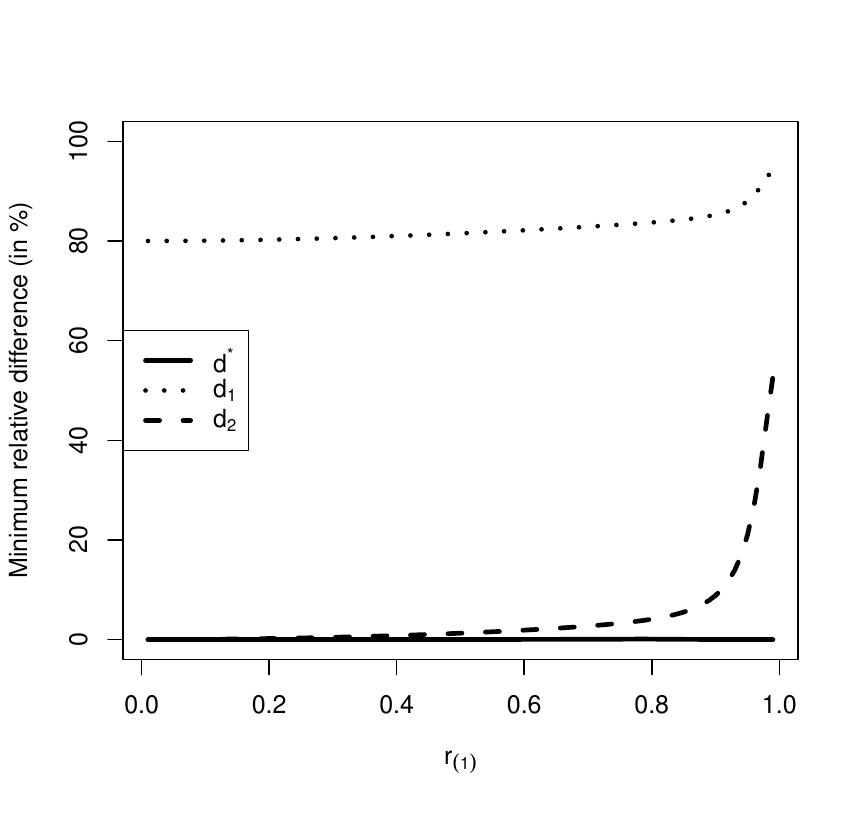}
  \caption{$min_{0< |\rho | <1}~ RD_{d^{(0)}}$ for $d^{(0)} \in \{ d_1, d_2, d^* \}$ \\($\boldsymbol{V}_1$: Mat($1.5$), $\boldsymbol{V}_R$: Mat($0.5$))}
  \label{fig:minimum_case3_p4t4}
\end{subfigure}%
\begin{subfigure}{.5\textwidth}
  \centering
  \includegraphics[height=0.4\linewidth, width=1.0\linewidth]{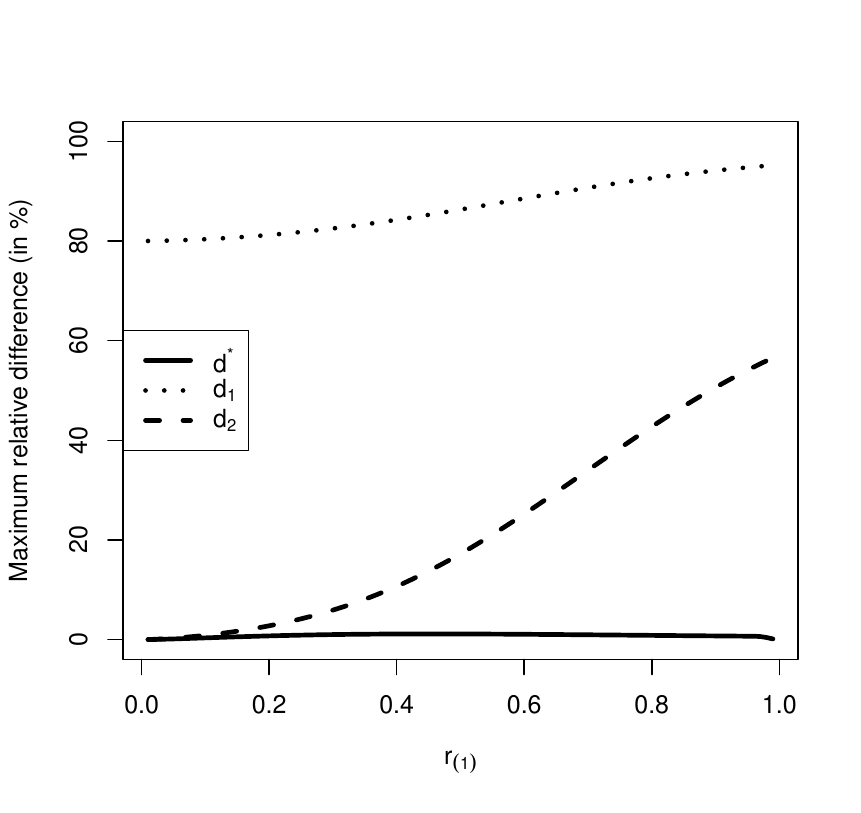}
  \caption{$max_{0< |\rho | <1}~ RD_{d^{(0)}}$ for $d^{(0)} \in \{ d_1, d_2, d^* \}$ \\($\boldsymbol{V}_1$: Mat($1.5$), $\boldsymbol{V}_R$: Mat($0.5$))}
  \label{fig:maximum_case3_p4t4}
\end{subfigure}\\
\begin{subfigure}{.5\textwidth}
  \centering
  \includegraphics[height=0.4\linewidth, width=1.0\linewidth]{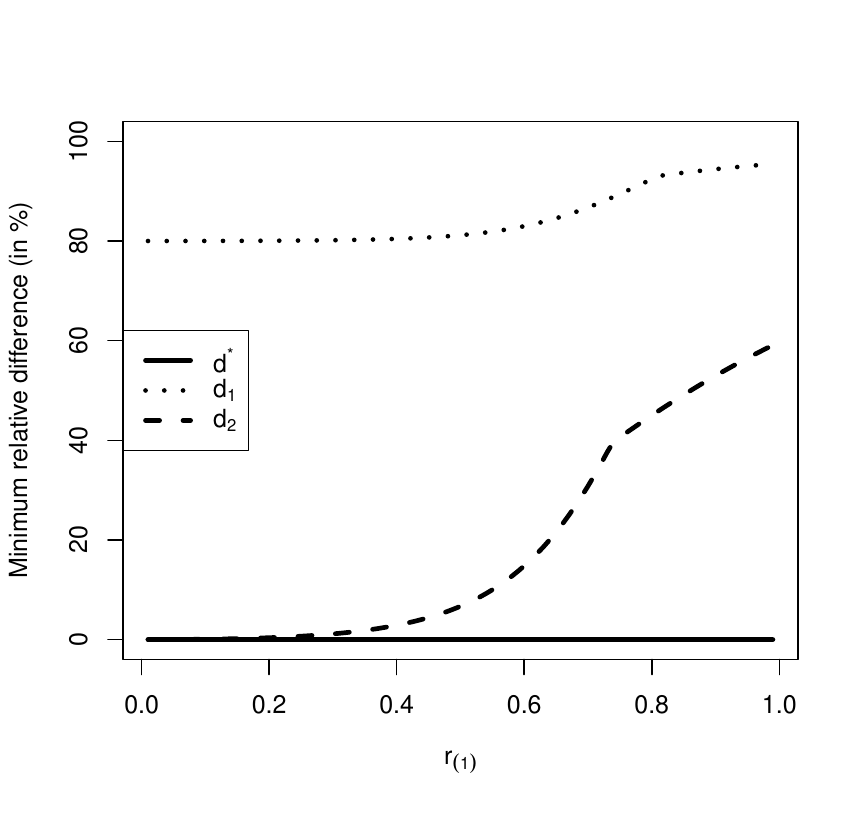}
  \caption{$min_{0< |\rho | <1}~ RD_{d^{(0)}}$ for $d^{(0)} \in \{ d_1, d_2, d^* \}$ \\($\boldsymbol{V}_1$: Mat($1.5$), $\boldsymbol{V}_R$: Mat($\infty$))}
  \label{fig:minimum_case4_p4t4}
\end{subfigure}%
\begin{subfigure}{.5\textwidth}
  \centering
  \includegraphics[height=0.4\linewidth, width=1.0\linewidth]{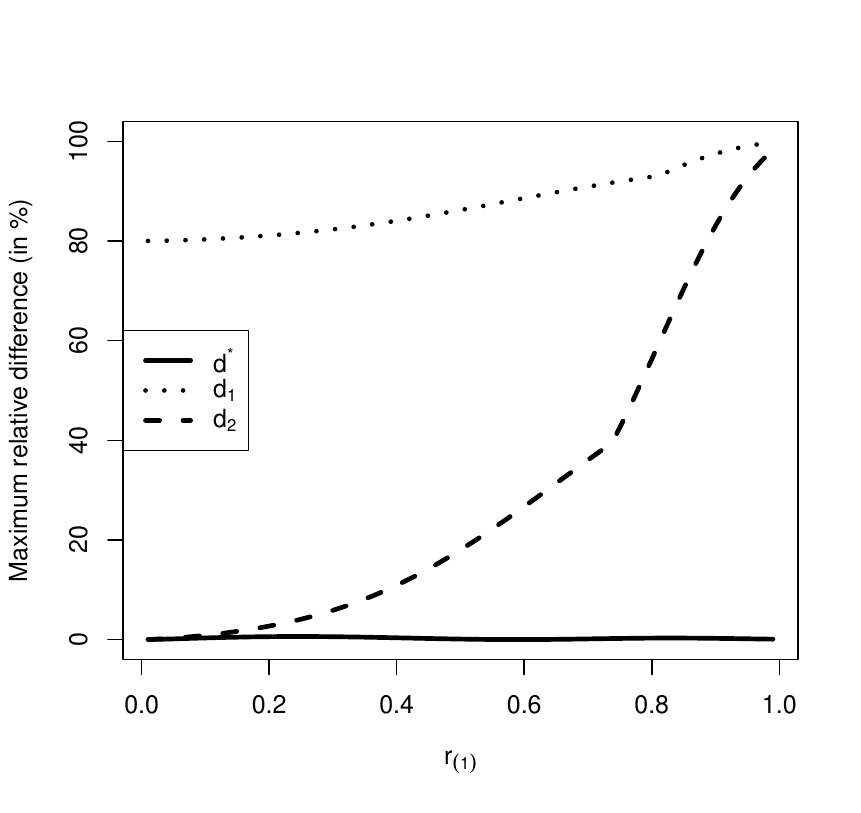}
  \caption{$max_{0< |\rho | <1}~ RD_{d^{(0)}}$ for $d^{(0)} \in \{ d_1, d_2, d^* \}$ \\($\boldsymbol{V}_1$: Mat($1.5$), $\boldsymbol{V}_R$: Mat($\infty$))}
  \label{fig:maximum_case4_p4t4}
\end{subfigure}
\caption{Plots of $min_{0< |\rho | <1}~ RD_{d^{(0)}}$ and $max_{0< |\rho | <1}~ RD_{d^{(0)}}$ for $d^{(0)} \in \{ d_1, d_2, d^* \}$, $p=t=4$, $n=12$ and $\sigma_{11} = \sigma_{22}$.}
\label{min-max-p4t4-1}
\end{figure}

\begin{figure}
\centering
\begin{subfigure}{.5\textwidth}
  \centering
  \includegraphics[height=0.4\linewidth, width=1.0\linewidth]{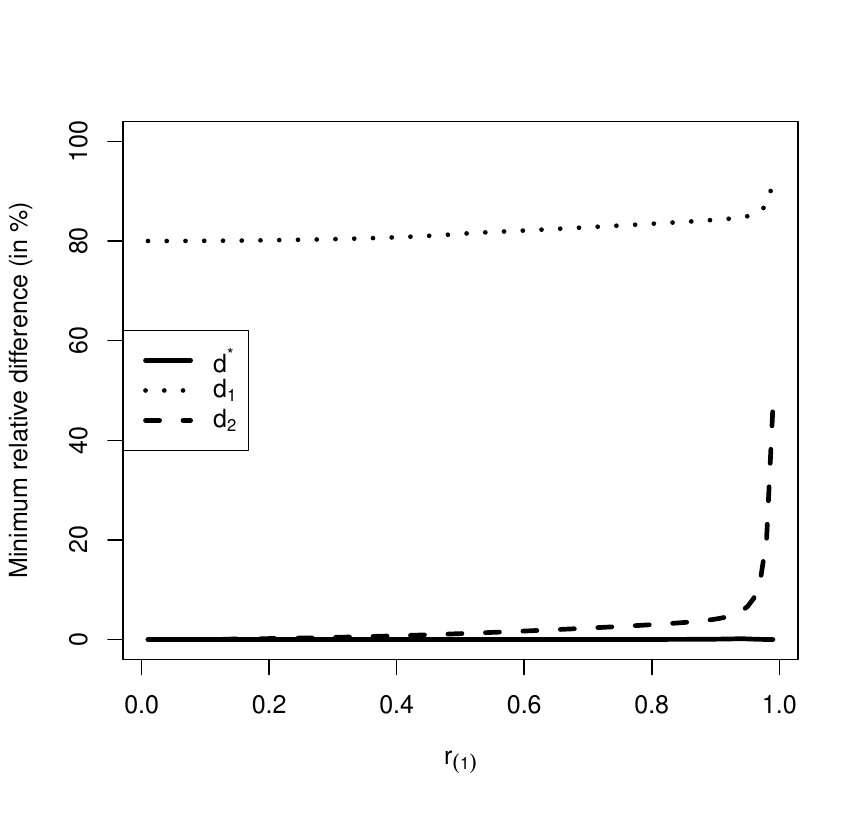}
  \caption{$min_{0< |\rho | <1}~ RD_{d^{(0)}}$ for $d^{(0)} \in \{ d_1, d_2, d^* \}$ \\($\boldsymbol{V}_1$: Mat($\infty$), $\boldsymbol{V}_R$: Mat($0.5$))}
  \label{fig:minimum_case5_p4t4}
\end{subfigure}%
\begin{subfigure}{.5\textwidth}
  \centering
  \includegraphics[height=0.4\linewidth, width=1.0\linewidth]{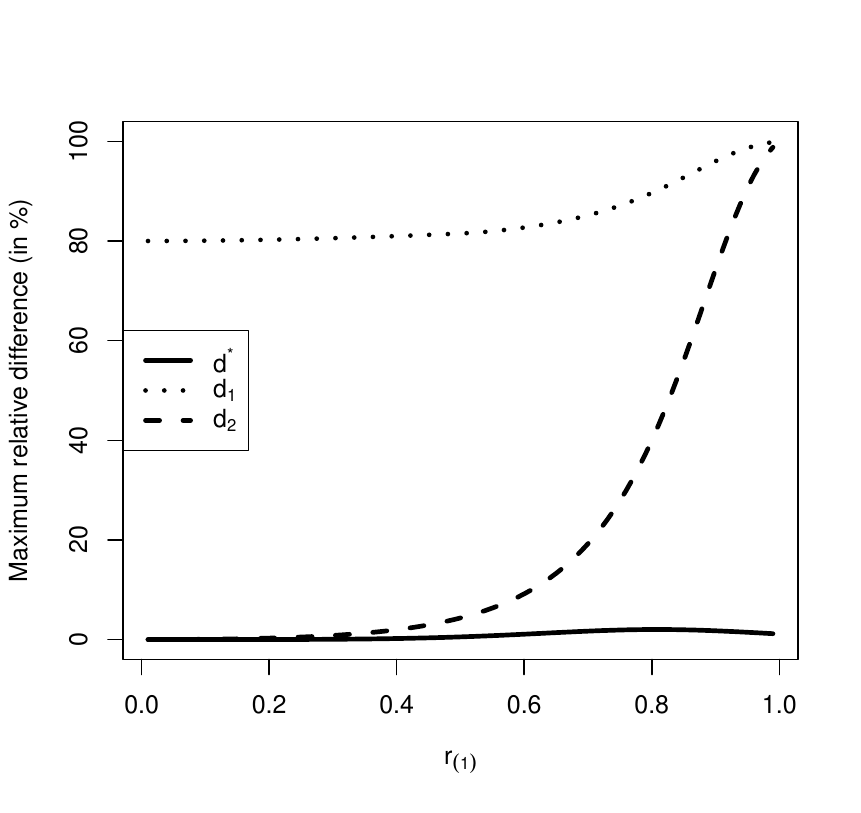}
  \caption{$max_{0< |\rho | <1}~ RD_{d^{(0)}}$ for $d^{(0)} \in \{ d_1, d_2, d^* \}$ \\($\boldsymbol{V}_1$: Mat($\infty$), $\boldsymbol{V}_R$: Mat($0.5$))}
  \label{fig:maximum_case5_p4t4}
\end{subfigure}\\
\begin{subfigure}{.5\textwidth}
  \centering
  \includegraphics[height=0.4\linewidth, width=1.0\linewidth]{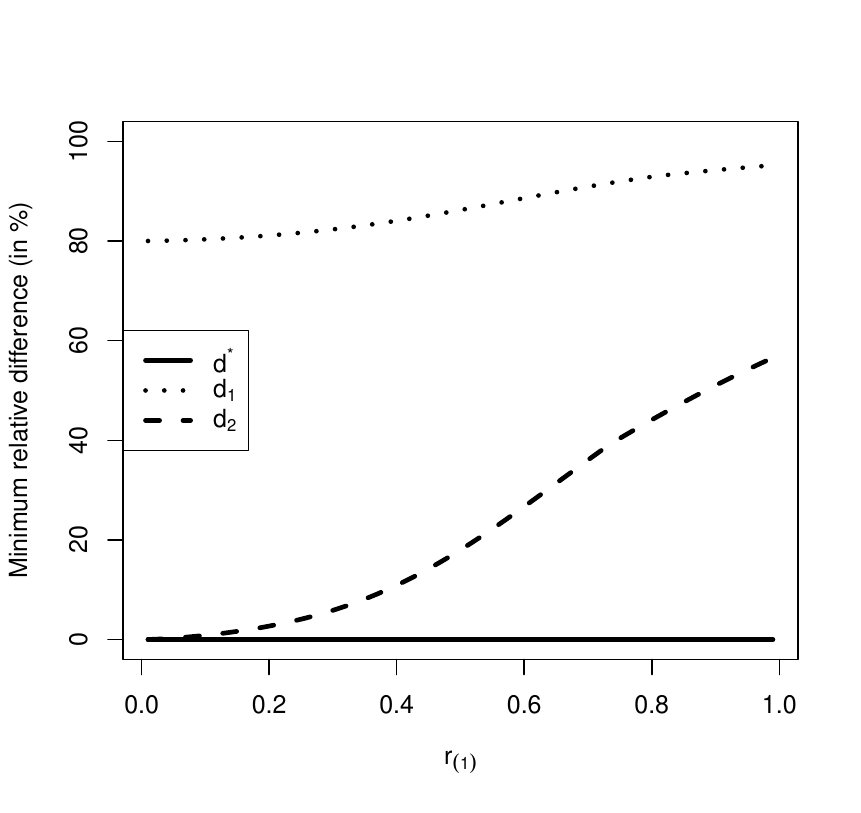}
  \caption{$min_{0< |\rho | <1}~ RD_{d^{(0)}}$ for $d^{(0)} \in \{ d_1, d_2, d^* \}$ \\($\boldsymbol{V}_1$: Mat($\infty$), $\boldsymbol{V}_R$: Mat($1.5$))}
  \label{fig:minimum_case6_p4t4}
\end{subfigure}%
\begin{subfigure}{.5\textwidth}
  \centering
  \includegraphics[height=0.4\linewidth, width=1.0\linewidth]{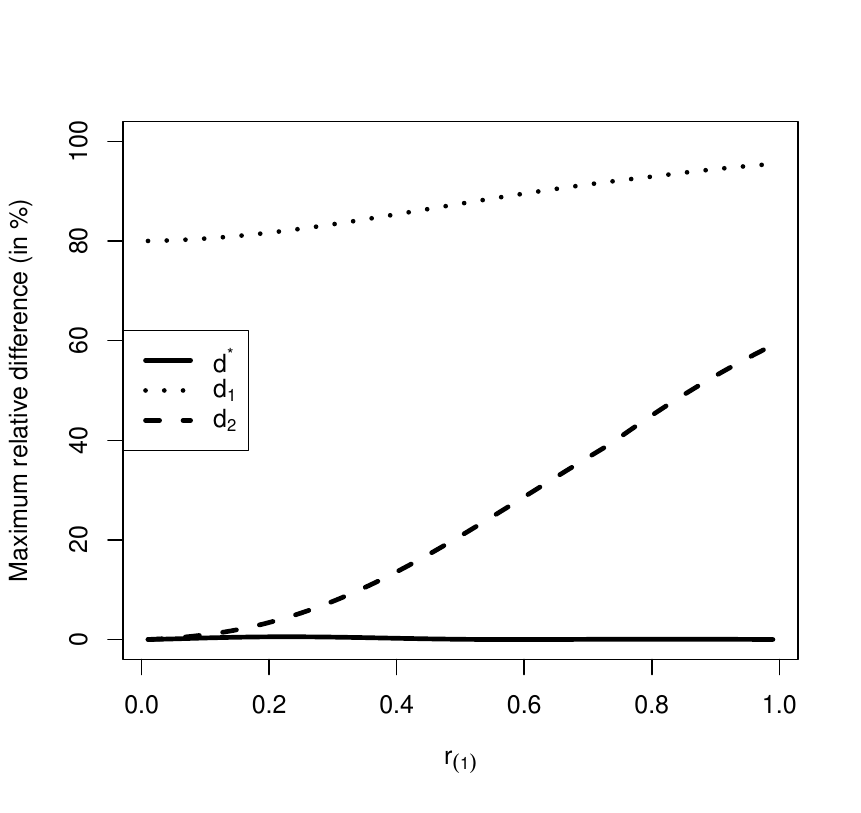}
  \caption{$max_{0< |\rho | <1}~ RD_{d^{(0)}}$ for $d^{(0)} \in \{ d_1, d_2, d^* \}$ \\($\boldsymbol{V}_1$: Mat($\infty$), $\boldsymbol{V}_R$: Mat($1.5$))}
  \label{fig:maximum_case6_p4t4}
\end{subfigure}\\
\begin{subfigure}{.5\textwidth}
  \centering
  \includegraphics[height=0.4\linewidth, width=1.0\linewidth]{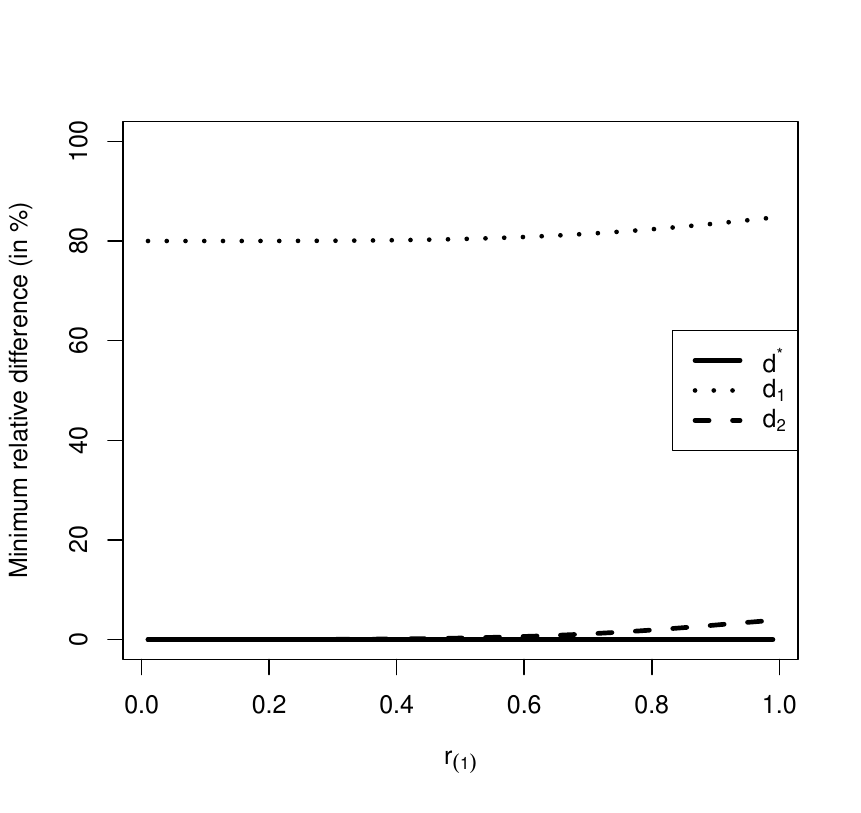}
  \caption{$min_{0< |\rho | <1}~ RD_{d^{(0)}}$ for $d^{(0)} \in \{ d_1, d_2, d^* \}$ \\($\boldsymbol{V}_1$: Mat($0.5$), $\boldsymbol{V}_R$: Mat($0.5$))}
  \label{fig:minimum_case7_p4t4}
\end{subfigure}%
\begin{subfigure}{.5\textwidth}
  \centering
  \includegraphics[height=0.4\linewidth, width=1.0\linewidth]{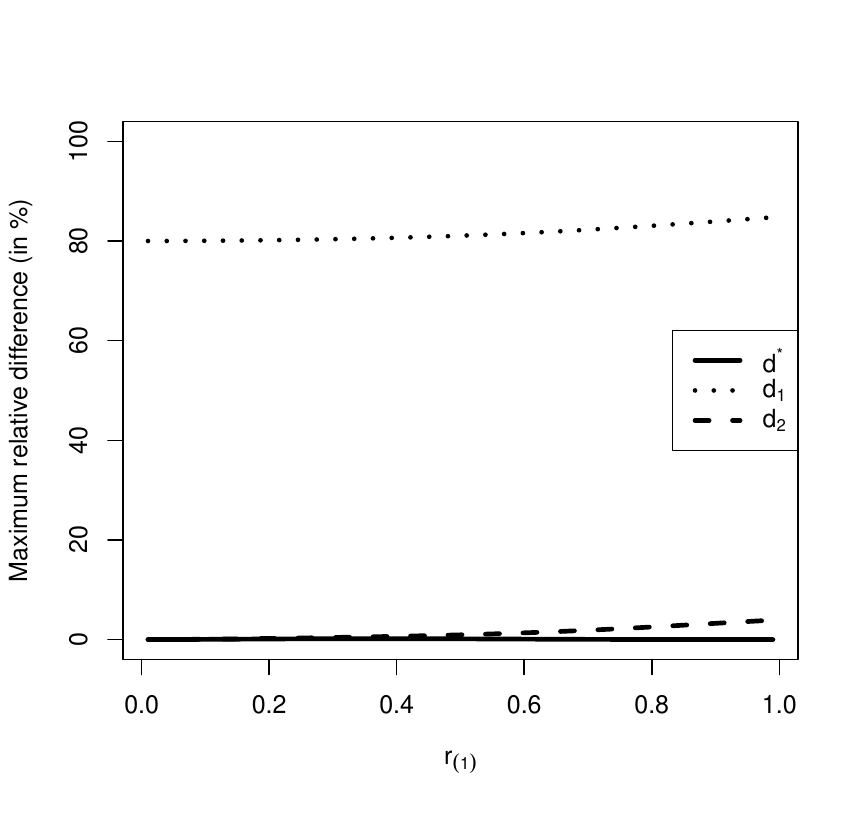}
  \caption{$max_{0< |\rho | <1}~ RD_{d^{(0)}}$ for $d^{(0)} \in \{ d_1, d_2, d^* \}$ \\($\boldsymbol{V}_1$: Mat($0.5$), $\boldsymbol{V}_R$: Mat($0.5$))}
  \label{fig:maximum_case7_p4t4}
\end{subfigure}
\caption{Plots of $min_{0< |\rho | <1}~ RD_{d^{(0)}}$ and $max_{0< |\rho | <1}~ RD_{d^{(0)}}$ for $d^{(0)} \in \{ d_1, d_2, d^* \}$, $p=t=4$, $n=12$ and $\sigma_{11} = \sigma_{22}$.}
\label{min-max-p4t4-2}
\end{figure}

\section{Illustration using the gene data}
\label{illustration}
We go back to use the genetic data example from \citet{leaker2017nasal} and evaluate the design used under different covariance structures. In particular, we focus on the $3 \times 3$ crossover design used in the gene data example, involving $18$ subjects and $3$ treatment sequences, namely $ABC$, $CAB$ and $BCA$. Subjects $1$-$6$ receive treatment sequence $ABC$, $7$-$12$ receive $CAB$, while the remaining receive $BCA$ and multiple gene expression values are measured from each subject in three time periods. Note the study design is a uniform design with $p=t=3$. We name this design $d_0$, where $d_0$ belongs to $\mathcal{D}^{(1)}_{t,n=\lambda t (t-1),p=t}$, with $t=3$ and $\lambda=3$.

\subsection{Proportional structure}
Under the proportional structure, to investigate trace optimality/efficiency of $d_0$ we consider three different structures of matrix $\boldsymbol{V}$ as discussed in \autoref{table_prop}. Corresponding to each case, for $p=3$, the matrix $\boldsymbol{V}$ is as follows:
\renewcommand{\sectionautorefname}{Appendix}
\begin{enumerate}[1.]
\item \label{a1}
$\boldsymbol{V} = 
\begin{bmatrix}
1 & r_{(1)} & r_{(1)}^2\\
r_{(1)} & 1 & r_{(1)}\\
r_{(1)}^2 & r_{(1)} & 1
\end{bmatrix}$, where $0 < r_{(1)} < 1$;
\item \label{b1}
$\boldsymbol{V} =
\begin{bmatrix}
1 & \left[1+log \left(r_{(1)} \right) \right] r_{(1)} & \left[1+2 log \left(r_{(1)} \right) \right] r_{(1)}^2\\
\left[1+log \left(r_{(1)} \right) \right] r_{(1)} & 1 & \left[1+log \left(r_{(1)} \right) \right] r_{(1)}\\
\left[1+2 log \left(r_{(1)} \right) \right] r_{(1)}^2 & \left[1+log \left(r_{(1)} \right) \right] r_{(1)} & 1
\end{bmatrix}$, where $0 < r_{(1)} < 1$;
\item \label{c11}
$\boldsymbol{V} = 
\begin{bmatrix}
1 & r_{(1)} & r_{(1)}^4\\
r_{(1)} & 1 & r_{(1)}\\
r_{(1)}^4 & r_{(1)} & 1
\end{bmatrix}$, where $0 < r_{(1)} < 1$.
\end{enumerate}

\noindent From \autoref{prop-thm-1} in Subsection~\ref{proportional covariance}, we know that under the proportional structure, a design $d^* \in \mathcal{D}^{(1)}_{t,n=\lambda t (t-1),p=t}$ represented by $OA_{I} \left( n=\lambda t \left(t-1 \right), p=t, t, 2 \right)$, where $\mathcal{D}^{(1)}_{t,n=\lambda t (t-1),p=t}$ is a class of binary designs with $p=t$, $\lambda$ is a positive integer and $t \geq 3$, is trace optimal/efficient for the direct effects over $\mathcal{D}^{(1)}_{t,n=\lambda t (t-1),p=t}$ for the multivariate response case. For any binary design $d \in \mathcal{D}^{(1)}_{t,n=\lambda t (t-1),p=t}$, with $t$ and $n$ same as that of $d^*$, let $e$ represent the efficiency of design $d$, defined as \citep{kunert1991cross}
\begin{align*}
e &= \frac{tr \left( \boldsymbol{C}_{d(s1)} \right)}{tr \left( \boldsymbol{C}_{d^*(s1)} \right)}.
\end{align*}
Since $d^*$ maximizes $tr \left( \boldsymbol{C}_{d(s1)} \right)$ over $\mathcal{D}^{(1)}_{t,n=\lambda t (t-1),p=t}$, $e$ lies between $0$ and $1$, and values of $e$ close to $1$ indicate that $d$ is a highly efficient design.\par

\noindent Using the expression of the information matrix for the direct effects under the proportional structure from \autoref{prop-lemma4-c4-1}, we get
\begin{align}
e &= \frac{tr \left( \boldsymbol{\mathit{\Gamma}}^{-1} \right) tr \left( \boldsymbol{C}_{d(uni)} \right)}{ tr \left( \boldsymbol{\mathit{\Gamma}}^{-1} \right) tr \left( \boldsymbol{C}_{d^*(uni)} \right)} = \frac{ tr \left( \boldsymbol{C}_{d(uni)} \right)}{  tr \left( \boldsymbol{C}_{d^*(uni)} \right)},
\label{prop-eqn-efficiency-1}
\end{align}
where $\boldsymbol{C}_{d(uni)} = \boldsymbol{C}_{d(uni)(11)} - \boldsymbol{C}_{d(uni)(12)} \boldsymbol{C}_{d(uni)(22)}^{-} \boldsymbol{C}_{d(uni)(21)}$, $\boldsymbol{C}_{d(uni)(11)} = \boldsymbol{T}_d^{'} \left( \hatmat_n \otimes \boldsymbol{V}^* \right) \boldsymbol{T}_d$, \\$\boldsymbol{C}_{d(uni)(12)} = \boldsymbol{C}_{d(uni)(21)}^{'} = \boldsymbol{T}_d^{'} \left( \hatmat_n \otimes \boldsymbol{V}^* \right) \boldsymbol{F}_d $, and $\boldsymbol{C}_{d(uni)(22)} = \boldsymbol{F}_d^{'} \left( \hatmat_n \otimes \boldsymbol{V}^* \right) \boldsymbol{F}_d $.\par

\noindent From \eqref{prop-eqn-efficiency-1}, it is clear that for a particular design $d \in \mathcal{D}^{(1)}_{t,n=\lambda t (t-1),p=t}$, and fixed values of $t$, $n$ and $p$, $e$ depends only on the matrix $\boldsymbol{V}$, thus only on $r_{(1)}$. Our computations for $0<r_{(1)}<1$, show that the maximum value of $e$ under the three structures of $\boldsymbol{V}$ is around $2.78\%$, indicating that $d_0$ performs very poorly in comparison to the efficient design $d^*$.

\subsection{Generalized Markov-type structure}
To check the trace optimality of design $d_0$ and compare its performance with $d^*$ under the Markovian structure, we consider cases $5$, $6$ and $7$ from Subsubsection~\ref{nearly-efficient}. Note that here $d^*$ is a design represented by an orthogonal array of Type $I$ and strength $2$ with $p=t=3$ and $n=18$. We assume that $\sigma_{11} = \sigma_{22}$. \autoref{min-max-example-2} displays $min_{0< |\rho | <1} ~ RD_{d_0}$, $min_{0< |\rho | <1} ~ RD_{d^*}$, $max_{0< |\rho | <1} ~ RD_{d_0}$ and $max_{0< |\rho | <1} ~ RD_{d^*}$ under the $3$ cases considered. From \autoref{min-max-example-2}, we note $min_{0< |\rho | <1} ~ RD_{d_0}$ and $min_{0< |\rho | <1} ~ RD_{d^*}$ never attain $0$, and $max_{0< |\rho | <1} ~ RD_{d_0}$ and $max_{0< |\rho | <1} ~ RD_{d^*}$ are around $99.70\%$ and $0.38\%$, respectively, indicating the poor performance of $d_0$ when compared to $d^*$. From these comparisons, we may conclude that if the experiment for the gene study had used an orthogonal array based design then he/she would have scored much in terms of efficiency. The poor choice of designs may also affect the parameter estimation for the gene dataset.

\begin{figure}
\centering
\begin{subfigure}{.5\textwidth}
  \centering
  \includegraphics[height=0.4\linewidth, width=1.0\linewidth]{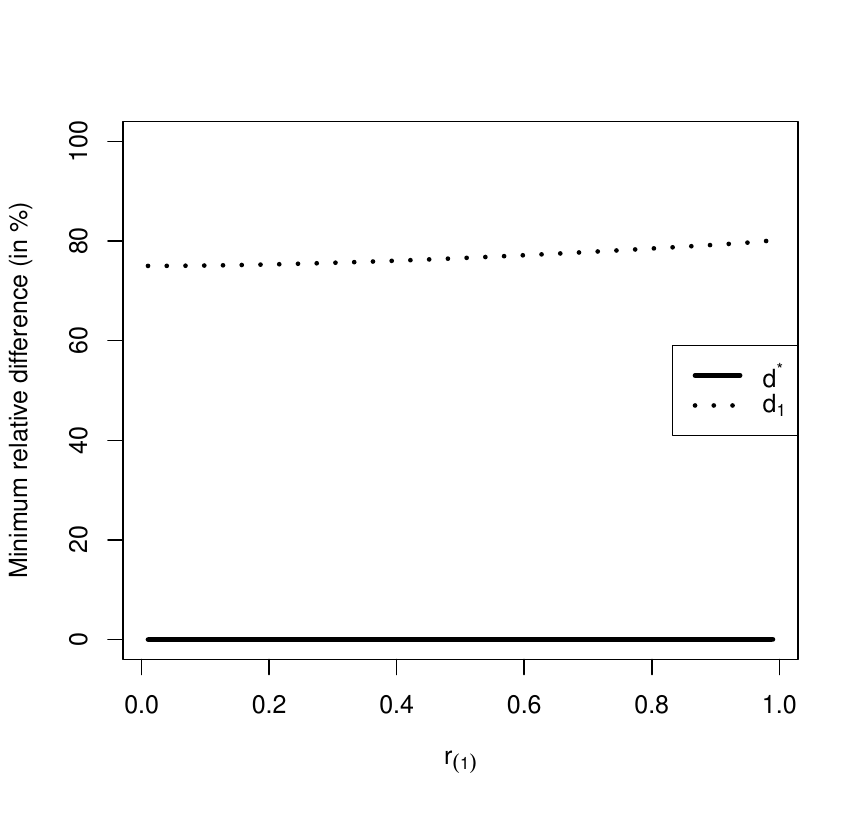}
  \caption{$min_{0< |\rho | <1}~ RD_{d^{(0)}}$ for $d^{(0)} \in \{ d_0, d^* \}$ \\($\boldsymbol{V}_1$: Mat($\infty$), $\boldsymbol{V}_R$: Mat($0.5$))}
  \label{fig:minimum_case5_example}
\end{subfigure}%
\begin{subfigure}{.5\textwidth}
  \centering
  \includegraphics[height=0.4\linewidth, width=1.0\linewidth]{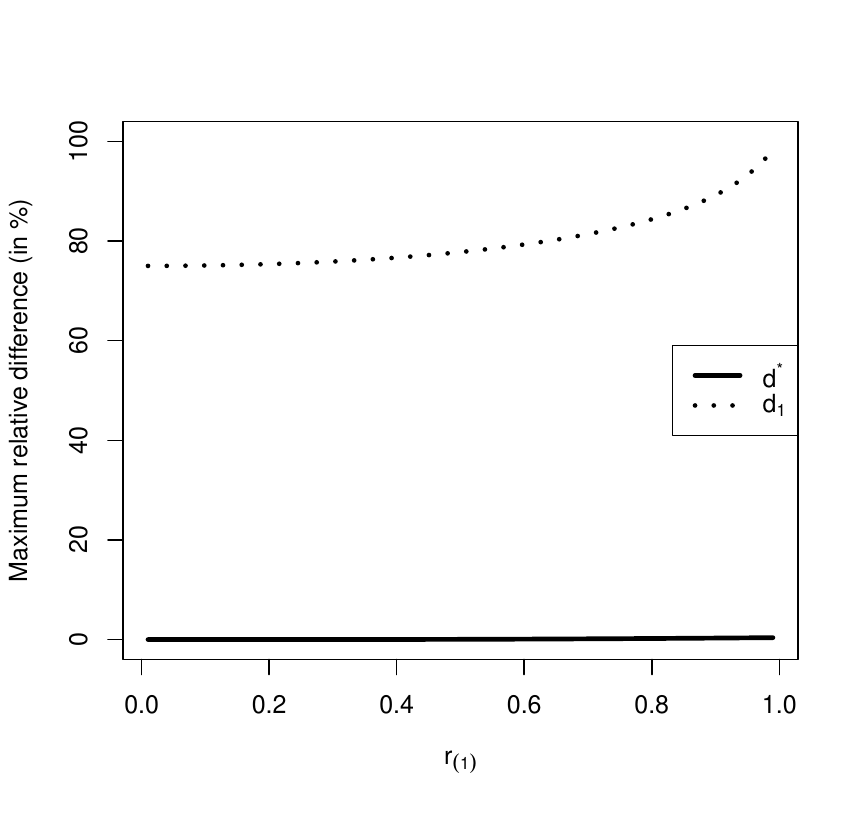}
  \caption{$max_{0< |\rho | <1}~ RD_{d^{(0)}}$ for $d^{(0)} \in \{ d_0, d^* \}$ \\($\boldsymbol{V}_1$: Mat($\infty$), $\boldsymbol{V}_R$: Mat($0.5$))}
  \label{fig:maximum_case5_example}
\end{subfigure}\\
\begin{subfigure}{.5\textwidth}
  \centering
  \includegraphics[height=0.4\linewidth, width=1.0\linewidth]{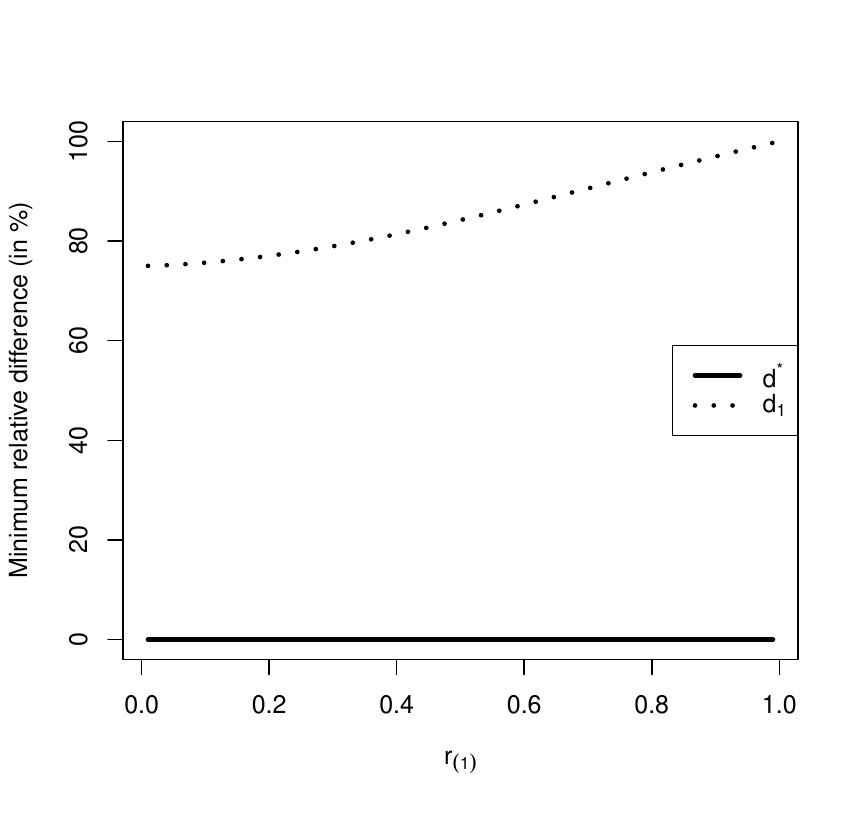}
  \caption{$min_{0< |\rho | <1}~ RD_{d^{(0)}}$ for $d^{(0)} \in \{ d_0, d^* \}$ \\($\boldsymbol{V}_1$: Mat($\infty$), $\boldsymbol{V}_R$: Mat($1.5$))}
  \label{fig:minimum_case6_example}
\end{subfigure}%
\begin{subfigure}{.5\textwidth}
  \centering
  \includegraphics[height=0.4\linewidth, width=1.0\linewidth]{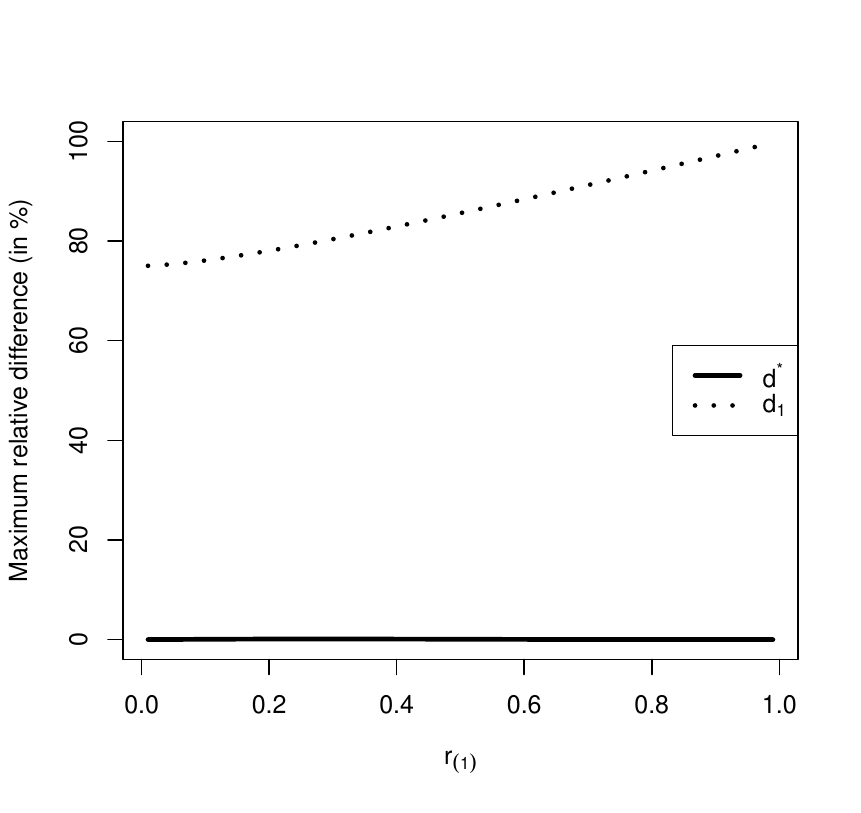}
  \caption{$max_{0< |\rho | <1}~ RD_{d^{(0)}}$ for $d^{(0)} \in \{ d_0, d^* \}$ \\($\boldsymbol{V}_1$: Mat($\infty$), $\boldsymbol{V}_R$: Mat($1.5$))}
  \label{fig:maximum_case6_example}
\end{subfigure}\\
\begin{subfigure}{.5\textwidth}
  \centering
  \includegraphics[height=0.4\linewidth, width=1.0\linewidth]{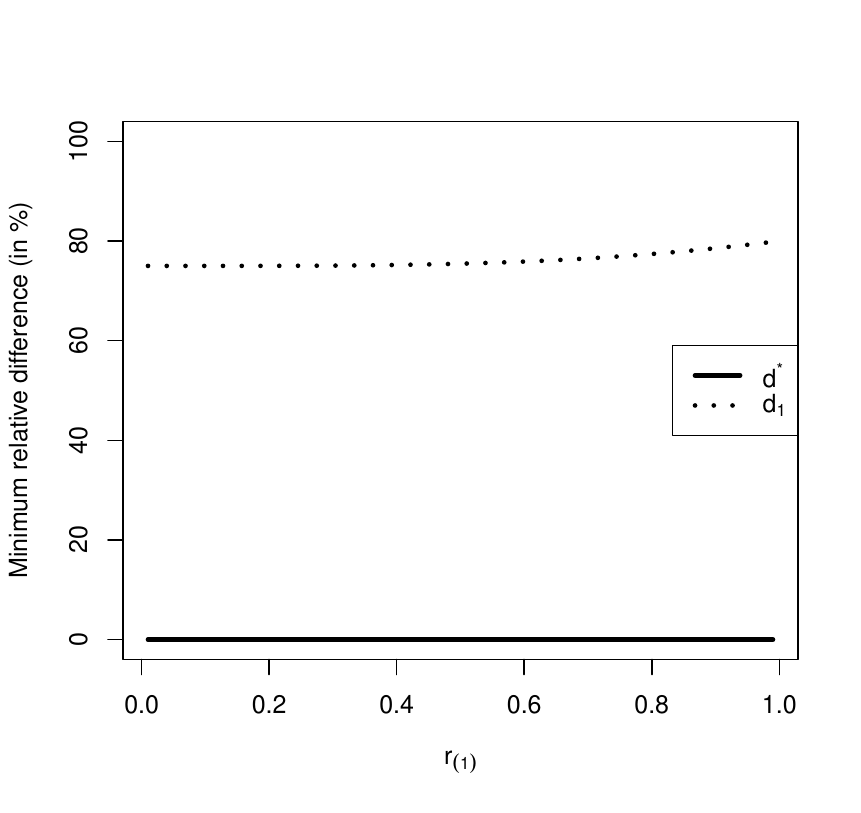}
  \caption{$min_{0< |\rho | <1}~ RD_{d^{(0)}}$ for $d^{(0)} \in \{ d_0, d^* \}$ \\($\boldsymbol{V}_1$: Mat($0.5$), $\boldsymbol{V}_R$: Mat($0.5$))}
  \label{fig:minimum_case7_example}
\end{subfigure}%
\begin{subfigure}{.5\textwidth}
  \centering
  \includegraphics[height=0.4\linewidth, width=1.0\linewidth]{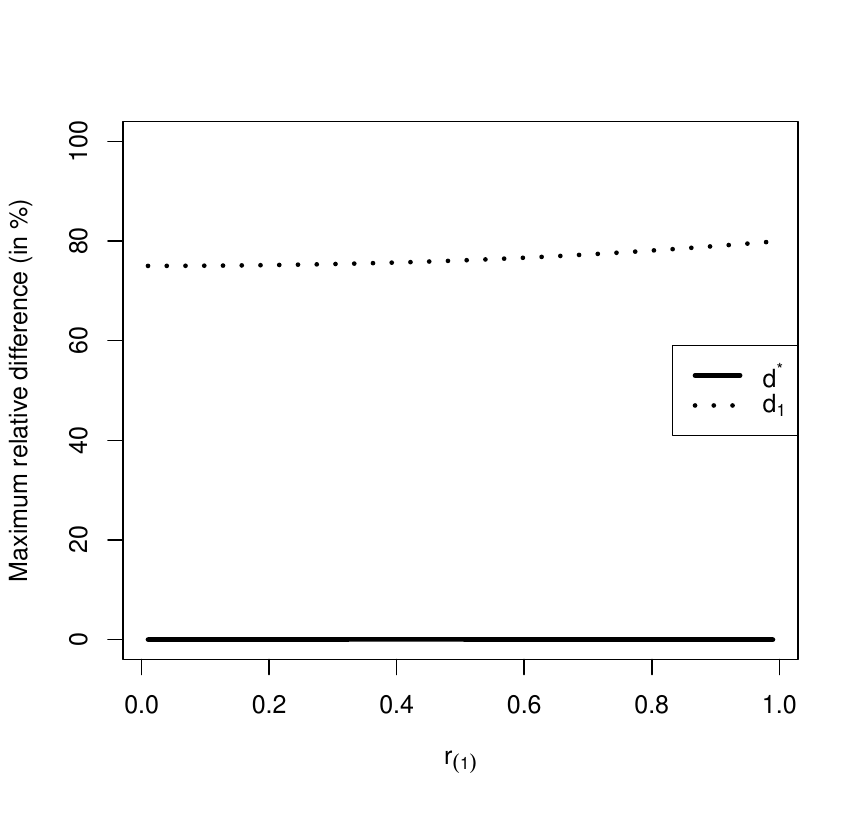}
  \caption{$max_{0< |\rho | <1}~ RD_{d^{(0)}}$ for $d^{(0)} \in \{ d_0, d^* \}$ \\($\boldsymbol{V}_1$: Mat($0.5$), $\boldsymbol{V}_R$: Mat($0.5$))}
  \label{fig:maximum_case7_example}
\end{subfigure}
\caption{Plots of $min_{0< |\rho | <1}~ RD_{d^{(0)}}$ and $max_{0< |\rho | <1}~ RD_{d^{(0)}}$ for $d^{(0)} \in \{ d_0, d^* \}$, where $p=t=3$, $n=18$ and $\sigma_{11} = \sigma_{22}$.}
\label{min-max-example-2}
\end{figure}

\section{Computational details}
\label{computation}
The exploratory analysis for the gene data regarding the presence of varying effects, and within and between response correlation (see Section~\ref{motexample}), was illustrated through line plots and correlation plots with the p-values of the correlation test, respectively. Also, the comparison of designs (see Subsubsection~\ref{nearly-efficient} and Section~\ref{illustration}) is represented using line plots. For the above purpose, we used R programming \citep[][]{rr} with R version 4.3.1 under Windows 11 (64-bit), with a 13th Gen Intel(R) Core(TM) I5-1335U processor and 8GB of RAM. We used \textit{base}, \textit{MASS}, \textit{pracma}, \textit{raster}, \textit{rapportools}, \textit{writexl}, \textit{graphics}, \textit{dplyr}, \textit{tidyr}, \textit{ggplot2}, \textit{readxl}, \textit{snpar}, \textit{stats}, \textit{QuantPsyc}, \textit{corrplot} packages. The R codes (.R files), the dataset on the $3$ genes (.zip file), and the MATLAB code (.m file) used for cross-checking in Subsubsection~\ref{nearly-efficient} are available at \url{https://github.com/rsphd/Efficient-designs-for-multivariate-crossover-trials}. 

\section{Conclusion}
\label{conclusion}
In this article, we studied the trace optimality/efficiency of an orthogonal array of Type $I$ and strength $2$ when measurements are recorded on $g$ responses, where $g > 1$. To account for correlation between responses, we considered two different types of structure for the error dispersion matrix, namely, the proportional structure and the generalized Markov-type structure. Under the proportional covariance structure, we considered within response covariance structures, i.e, structures of $\boldsymbol{V}_1$ to be Mat($0.5$), Mat($1.5$) and Mat($\infty$). Under the generalized Markov-type covariance structure, we considered $7$ different choices for structure of ($\boldsymbol{V}_1$, $\boldsymbol{V}_R$), namely (Mat($0.5$), Mat($1.5$)), (Mat($0.5$), Mat($\infty$)), (Mat($1.5$), Mat($0.5$)), (Mat($1.5$), Mat($\infty$)), (Mat($\infty$), Mat($0.5$)) and (Mat($\infty$), Mat($1.5$)). The main conclusions of this paper are:
\begin{enumerate}
\item Under the proportional structure, for $g>1$ and $p=t \geq 3$, a design given as an orthogonal array of type $I$ and strength $2$ is trace optimal/efficient for the direct effects over the class of binary designs.
\item Under the generalized Markov-type structure, for the $7$ different cases depending on the choices of matrices $\boldsymbol{V}_1$ and $\boldsymbol{V}_R$, when $p=t=3$ and $n=6$, and $p=t=4$ and $n=12$, an orthogonal array of type $I$ and strength $2$ is highly efficient for the direct effects over the class of binary designs.
\end{enumerate}

\noindent Although in this article, we have only considered universal optimality and trace optimality criteria, other criteria can also be investigated. In the future, for multivariate crossover designs, we plan to investigate $A$-, $D$-, and $E$-optimal designs, under the proportional and Markovian structures of the dispersion matrix.

\begin{appendices}
\numberwithin{equation}{section}
\section{Information matrix}
\renewcommand{\thesection}{A}

\subsection{Proportional structure}
\label{proportion-information matrix}

\renewcommand{\thesection}{A}
\renewcommand{\thesubsection}{A.1}
\numberwithin{equation}{section}

\begin{proof}[Proof of \autoref{prop-lemma4-c4-1}]
\label{prop-ss1}
From \citet{Bose2009OptimalDesigns}, we know that for the $g=1$ case with the dispersion matrix of error terms as a positive definite and symmetric matrix $\identity_n \otimes \boldsymbol{V}$, the information matrix for the direct effects is given as
\begin{align}
\boldsymbol{C}_{d(uni)} &= \boldsymbol{C}_{d(uni)(11)} - \boldsymbol{C}_{d(uni)(12)} \boldsymbol{C}_{d(uni)(22)}^{-} \boldsymbol{C}_{d(uni)(21)},
\label{prop-t21c-1}
\end{align}
where $\boldsymbol{C}_{d(uni)(11)} = \boldsymbol{T}_d^{'} \left( \hatmat_n \otimes \boldsymbol{V}^* \right) \boldsymbol{T}_d$, $\boldsymbol{C}_{d(uni)(12)} = \boldsymbol{C}_{d(uni)(21)}^{'} = \boldsymbol{T}_d^{'} \left( \hatmat_n \otimes \boldsymbol{V}^* \right) \boldsymbol{F}_d $ and \\ $\boldsymbol{C}_{d(uni)(22)} = \boldsymbol{F}_d^{'} \left( \hatmat_n \otimes \boldsymbol{V}^* \right) \boldsymbol{F}_d $. Here, the matrix $\boldsymbol{V}^* = \boldsymbol{V}^{-1} - \left( \vecone_{p}^{'} \boldsymbol{V}^{-1} \vecone_p \right)^{-1} \boldsymbol{V}^{-1} \matone_{p \times p} \boldsymbol{V}^{-1}$.\par

\noindent Using the expression of $\boldsymbol{A}^*$ as given in \autoref{prop-lemma3-c4-1} and the expression of the information matrix for the direct effects as given in \autoref{lemma5-c4-i}, we get that under the proportional structure for $g > 1$,
\begin{multline}
\boldsymbol{C}_{d(s1)} = \boldsymbol{\mathit{\Gamma}}^{-1} \otimes \\
\left[  \boldsymbol{T}_d^{'} \left( \hatmat_n \otimes \boldsymbol{V}^* \right) \boldsymbol{T}_d - \boldsymbol{T}_d^{'} \left( \hatmat_n \otimes \boldsymbol{V}^* \right) \boldsymbol{F}_d \left( \boldsymbol{F}_d^{'} \left( \hatmat_n \otimes \boldsymbol{V}^* \right) \boldsymbol{F}_d \right)^{-} \boldsymbol{F}_d^{'} \left( \hatmat_n \otimes \boldsymbol{V}^* \right) \boldsymbol{T}_d \right],
\label{prop-t22c-1}
\end{multline}
where $\boldsymbol{\mathit{\Gamma}}$ is a $g \times g$ positive definite and symmetric matrix with non-zero off-diagonal elements. Thus from \eqref{prop-t21c-1} and \eqref{prop-t22c-1}, we can prove \eqref{prop-t20c-1}.
\end{proof}

\begin{proof}[Proof of \autoref{prop-remark-1}]
\label{prop-ss4}
Here, $d^*$ is a design given by $OA_{I} \left( n=\lambda t \left(t-1 \right), p=t, t, 2 \right)$, where $\lambda$ is a positive integer and $t \geq 3$. From \autoref{prop-lemma4-c4-1}, we get that for $g>1$, the information matrix for the direct effects corresponding to $d^*$ is given as
\begin{align}
\boldsymbol{C}_{d^*(s1)} &= \boldsymbol{\mathit{\Gamma}}^{-1} \otimes \boldsymbol{C}_{d^*(uni)},
\label{prop-t23c-1}
\end{align}
where $\boldsymbol{C}_{d^*(uni)} = \boldsymbol{C}_{d^*(uni)(11)} - \boldsymbol{C}_{d^*(uni)(12)} \boldsymbol{C}_{d^*(uni)(22)}^{-} \boldsymbol{C}_{d^*(uni)(21)}$, and $\boldsymbol{\mathit{\Gamma}}$ is a $g \times g$ positive definite and symmetric matrix with non-zero off-diagonal elements. Here, $\boldsymbol{C}_{d^*(uni)(11)} = \boldsymbol{T}_{d^*}^{'} \left( \hatmat_n \otimes \boldsymbol{V}^* \right) \boldsymbol{T}_{d^*}$, $\boldsymbol{C}_{d^*(uni)(12)} = \boldsymbol{C}_{d^*(uni)(21)}^{'} = \boldsymbol{T}_{d^*}^{'} \left( \hatmat_n \otimes \boldsymbol{V}^* \right) \boldsymbol{F}_{d^*} $, $\boldsymbol{C}_{d^*(uni)(22)}=$ \\ $ \boldsymbol{F}_{d^*}^{'} \left( \hatmat_n \otimes \boldsymbol{V}^* \right) \boldsymbol{F}_{d^*} $ and the matrix $\boldsymbol{V}^* = \boldsymbol{V}^{-1} - \left( \vecone_{p}^{'} \boldsymbol{V}^{-1} \vecone_p \right)^{-1} \boldsymbol{V}^{-1} \matone_{p \times p} \boldsymbol{V}^{-1}$.\par

\noindent Since $\boldsymbol{F}_{d^*} = \left(\identity_n \otimes \boldsymbol{\psi} \right) \boldsymbol{T}_{d^*}$, where $\boldsymbol{\psi} =  
\begin{bmatrix}
\zero^{'}_{p-1 \times 1} & 0\\
\identity_{p-1} & \zero_{p-1 \times 1}
\end{bmatrix}$, and $\boldsymbol{V}$ is a known positive definite and symmetric matrix, from \citet{Martin1998Variance-balancedObservations}, we know
\begin{align}
\boldsymbol{C}_{d^*(uni)} &= \left( | \boldsymbol{E} |/e_{22} \right) \hatmat_t,
\label{prop-t24c-1}
\end{align}
where $\boldsymbol{E} = \frac{n}{t-1} \begin{bmatrix}
e_{11} & e_{12}\\
e_{12} & e_{22}
\end{bmatrix}$, $e_{11} = tr \left(  \boldsymbol{T}^{'}_{d^*1} \boldsymbol{V}^{*} \boldsymbol{T}_{d^*1}\right)$, $e_{12} = tr \left(  \boldsymbol{T}^{'}_{d^*1} \boldsymbol{V}^{*} \boldsymbol{\psi} \boldsymbol{T}_{d^*1}\right)$, $e_{22} = $\\$tr \left(  \boldsymbol{T}^{'}_{d^*1} \boldsymbol{\psi}^{'} \boldsymbol{V}^{*} \boldsymbol{\psi} \boldsymbol{T}_{d^*1}\right) -  \frac{\left(\boldsymbol{V}^{*} \right)_{1,1}}{t}$ and $| \boldsymbol{E} | \neq 0$. Here $\left(\boldsymbol{V}^{*} \right)_{1,1}$ is the element corresponding to the first row and first column of the matrix $\boldsymbol{V}^{*}$, and $\boldsymbol{T}_{d^*1}$ is the treatment incidence matrix for the first subject corresponding to $d^*$. For $g>1$, let us denote
\begin{align*}
\boldsymbol{\mathit{\Gamma}}^{-1} = 
\begin{bmatrix}
\gamma^{(11)} & \cdots & \gamma^{(1g)}\\
\vdots & \vdots & \vdots\\
\gamma^{(g1)} & \cdots & \gamma^{(gg)}
\end{bmatrix}.
\end{align*}
Since $\boldsymbol{\mathit{\Gamma}}$ is a $g \times g$ positive definite, symmetric and non-diagonal matrix, the same holds for $\boldsymbol{\mathit{\Gamma}}^{-1}$. So from \eqref{prop-t23c-1} and \eqref{prop-t24c-1}, we get that $\boldsymbol{C}_{d^*(s1)}$ is completely symmetric for $g>1$ case if and only if 
\begin{align}
\begin{split}
\frac{\gamma^{(11)} |\boldsymbol{E} | }{ t e_{22}  }   &= \cdots = \frac{\gamma^{(gg)} |\boldsymbol{E} | }{ t e_{22}  },\\
\frac{\gamma^{(kk)} |\boldsymbol{E} | }{t e_{22}  } &=  \frac{\gamma^{(kk')} |\boldsymbol{E} | }{t e_{22}  }, \text{ and}\\
\frac{\gamma^{(kk')} |\boldsymbol{E} | }{t e_{22}  } &= \frac{\gamma^{(kk')} |\boldsymbol{E} | \left(t-1 \right) }{t  e_{22}  },
\end{split}
\label{prop-t25c-1}
\end{align}
where $\gamma^{(kk')}$ is the $(k,k')^{th}$ element of the matrix $\boldsymbol{\mathit{\Gamma}}^{-1}$. Here, due to positive definiteness, symmetricity and non-diagonality of $\boldsymbol{\mathit{\Gamma}}^{-1}$, $\gamma^{(kk)}>0$ and for atleast one pair $(k,k')$, $\gamma^{(kk')} = \gamma^{(k'k)} \neq 0$. Thus \eqref{prop-t25c-1} does not hold. Hence, $\boldsymbol{C}_{d^*(s1)}$ is not a completely symmetric matrix for $g>1$ case.
\end{proof}

\renewcommand{\thesubsection}{A.2}
\subsection{Generalized Markov-type structure}
\label{markov-information matrix}

\renewcommand{\thesection}{A}
\renewcommand{\thesubsection}{A.2}

\begin{proof}[Proof of \autoref{thm5-c4}]
\label{ss6}
Here $d^* \in \mathcal{D}^{(1)}_{t,n=\lambda t (t-1),p=t}$ be a design given by $OA_{I} \big( n=\lambda t \left(t-1 \right), p=t, t, 2 \big)$, where $\lambda$ is a positive integer and $t \geq 3$, and thus $d^*$ is a design uniform on periods. Since from \autoref{remark1-c4}, we get that $\vecone_p^{'} \boldsymbol{\mathit{\Omega}}_1 = \vecone_p^{'} \boldsymbol{\mathit{\Omega}}_2 = \vecone_p^{'} \boldsymbol{\mathit{\Omega}}_4 = \zero_{1 \times p}$, $\vecone_{p}^{'} \boldsymbol{\psi}^{'} \boldsymbol{\mathit{\Omega}}_x \boldsymbol{\psi} \vecone_p$ is the entry corresponding to the first row and first column of the matrix $\boldsymbol{\mathit{\Omega}}_x$, for $x=1,2,4$. Here $\boldsymbol{\mathit{\Omega}}_1$, $\boldsymbol{\mathit{\Omega}}_2$ and $\boldsymbol{\mathit{\Omega}}_4$ are as given in \autoref{lemma3-c4-1}. Using the expression of $\boldsymbol{{C}}_{d^*(s2)}$ as given in \autoref{lemma5-c4} and the expression of $\boldsymbol{A}^*$ in \autoref{lemma3-c4-1}, we get
\begin{align}
\boldsymbol{{C}}_{d^*(s2)} &= \boldsymbol{{C}}_{d^*(s2)(11)} - \boldsymbol{{C}}_{d^*(s2)(12)} \boldsymbol{{C}}_{d^*(s2)(22)}^{-} \boldsymbol{{C}}_{d^*(s2)(21)},
\label{p50c}
\end{align}
where 
\begin{align*}
\begin{split}
\boldsymbol{C}_{{d^*(s2)}(11)} &=
\begin{bmatrix}
\boldsymbol{T}_{d^*}^{'} \left( \hatmat_n \otimes \boldsymbol{\mathit{\Omega}}_1 \right) \boldsymbol{T}_{d^*} & -\boldsymbol{T}_{d^*}^{'} \left( \hatmat_n \otimes \boldsymbol{\mathit{\Omega}}_2 \right) \boldsymbol{T}_{d^*}\\
-\boldsymbol{T}_{d^*}^{'} \left( \hatmat_n \otimes \boldsymbol{\mathit{\Omega}}_2 \right) \boldsymbol{T}_{d^*} & \boldsymbol{T}_{d^*}^{'} \left( \hatmat_n \otimes \boldsymbol{\mathit{\Omega}}_4 \right) \boldsymbol{T}_{d^*}
\end{bmatrix},\\
\boldsymbol{C}_{{d^*(s2)}(12)} &= \boldsymbol{C}^{'}_{{d^*(s2)}(21)} =
\begin{bmatrix}
\boldsymbol{T}_{d^*}^{'} \left( \hatmat_n \otimes \boldsymbol{\mathit{\Omega}}_1 \boldsymbol{\psi} \right) \boldsymbol{T}_{d^*} & -\boldsymbol{T}_{d^*}^{'} \left( \hatmat_n \otimes \boldsymbol{\mathit{\Omega}}_2 \boldsymbol{\psi} \right) \boldsymbol{T}_{d^*} \\
-\boldsymbol{T}_{d^*}^{'} \left( \hatmat_n \otimes \boldsymbol{\mathit{\Omega}}_2 \boldsymbol{\psi} \right) \boldsymbol{T}_{d^*}  & \boldsymbol{T}_{d^*}^{'} \left( \hatmat_n \otimes \boldsymbol{\mathit{\Omega}}_4 \boldsymbol{\psi} \right) \boldsymbol{T}_{d^*} 
\end{bmatrix}, \text{ and}\\
\boldsymbol{C}_{{d^*(s2)}(22)} &=
\begin{bmatrix}
 \boldsymbol{T}_{d^*}^{'} \left( \hatmat_n \otimes \boldsymbol{\psi}^{'} \boldsymbol{\mathit{\Omega}}_1 \boldsymbol{\psi} \right) \boldsymbol{T}_{d^*}  & - \boldsymbol{T}_{d^*}^{'} \left( \hatmat_n \otimes \boldsymbol{\psi}^{'} \boldsymbol{\mathit{\Omega}}_2 \boldsymbol{\psi} \right) \boldsymbol{T}_{d^*} \\
- \boldsymbol{T}_{d^*}^{'} \left( \hatmat_n \otimes \boldsymbol{\psi}^{'} \boldsymbol{\mathit{\Omega}}_2 \boldsymbol{\psi} \right) \boldsymbol{T}_{d^*}  &  \boldsymbol{T}_{d^*}^{'} \left( \hatmat_n \otimes \boldsymbol{\psi}^{'} \boldsymbol{\mathit{\Omega}}_4 \boldsymbol{\psi} \right) \boldsymbol{T}_{d^*} 
\end{bmatrix}.
\end{split}
\end{align*}
Thus from the above equation and following Lemma 3, Lemma 4 and Lemma 5 from \citet{Martin1998Variance-balancedObservations}, we get
\begin{align}
\begin{split}
\boldsymbol{{C}}_{d^*(s2)(11)} &= \frac{n}{t-1}
\begin{bmatrix}
c_{11(11)} & -c_{11(12)}\\
-c_{11(12)} & c_{11(22)}
\end{bmatrix}
\otimes \hatmat_t,\\
\boldsymbol{{C}}_{d^*(s2)(12)} &= \boldsymbol{\tilde{C}}_{d^*(s2)(21)} = \frac{n}{t-1}
\begin{bmatrix}
c_{12(11)} & -c_{12(12)}\\
-c_{12(12)} & c_{12(22)}
\end{bmatrix}
\otimes \hatmat_t, \text{ and}\\
\boldsymbol{{C}}_{d^*(s2)(22)} &= \frac{n}{t-1}
\begin{bmatrix}
c_{22(11)} & -c_{22(12)}\\
-c_{22(12)} & c_{22(22)}
\end{bmatrix}
\otimes \hatmat_t,
\end{split}
\label{p36c}
\end{align}
where $
c_{11(11)} = tr \left( 
\boldsymbol{\mathit{\Omega}}_1
\right)$, $c_{11(12)} = tr \left( 
\boldsymbol{\mathit{\Omega}}_2
\right)$, $c_{11(22)} = tr \left( 
\boldsymbol{\mathit{\Omega}}_4
\right)$, $c_{12(11)} = tr \left( 
\boldsymbol{\mathit{\Omega}}_1
\boldsymbol{\psi}
\right)$, $c_{12(12)} = tr \left( 
\boldsymbol{\mathit{\Omega}}_2
\boldsymbol{\psi}
\right)$, $c_{12(22)} = tr \left( 
\boldsymbol{\mathit{\Omega}}_4
\boldsymbol{\psi} 
\right)$, $c_{22(11)} = tr \left( \hatmat_p
\boldsymbol{\psi}^{'}
\boldsymbol{\mathit{\Omega}}_1
\boldsymbol{\psi}
\right)$, $c_{22(12)} = tr \left( \hatmat_p
\boldsymbol{\psi}^{'}
\boldsymbol{\mathit{\Omega}}_2
\boldsymbol{\psi}
\right)$, $c_{22(22)} = tr \left( \hatmat_p
\boldsymbol{\psi}^{'}
\boldsymbol{\mathit{\Omega}}_4
\boldsymbol{\psi}
\right)$. From \autoref{lemma3-c4-1}, we get
\begin{align}
\begin{split}
\boldsymbol{\mathit{\Omega}}_1 &= \boldsymbol{V}_1^* + \frac{\bar{\rho}^2}{\sigma_{12}} \boldsymbol{V}_R^* = \boldsymbol{V}_1^* + \bar{\rho}^2 \boldsymbol{\mathit{\Omega}}_4,\\
\boldsymbol{\mathit{\Omega}}_2 &= \bar{\rho} \boldsymbol{\mathit{\Omega}}_4, \text{ and}\\
\boldsymbol{\mathit{\Omega}}_4 &= \frac{1}{\sigma_{12}} \boldsymbol{V}_R^*,
\end{split}
\label{markov-t29c-1}
\end{align}
where $\bar{\rho} = \sqrt{\frac{\sigma_{22}}{\sigma_{11}}}$ and $\sigma_{12} = \sigma_{22} (1-\rho^2)$. Thus using the expressions of $c_{11(11)}$, $c_{11(12)}$, $c_{11(22)}$, $c_{12(11)}$, $c_{12(12)}$, $c_{12(22)}$, $c_{22(11)}$, $c_{22(12)}$ and $c_{22(22)}$, we get
\begin{align}
\begin{split}
c_{11(11)} &= tr \left( \boldsymbol{V}_1^* \right) + \bar{\rho}^2  c_{11(22)},\\
c_{11(12)} &= \bar{\rho} \times c_{11(22)},\\
c_{11(22)} &= \frac{1}{\sigma_{12}} tr \left( \boldsymbol{V}_R^*  \right),\\
c_{12(11)} &= tr \left(  \boldsymbol{V}_1^* \boldsymbol{\psi} \right) + \bar{\rho}^2  c_{12(22)},\\
c_{12(12)} &= \bar{\rho} \times c_{12(22)},\\
c_{12(22)} &= \frac{1}{\sigma_{12}} tr \left(  \boldsymbol{V}_R^*  \boldsymbol{\psi} \right),\\
c_{22(11)} &= tr \left( \hatmat_p \boldsymbol{\psi}^{'} \boldsymbol{V}_1^* \boldsymbol{\psi} \right) + \bar{\rho}^2  c_{22(22)},\\
c_{22(12)} &= \bar{\rho} \times c_{22(22)}, \text{ and}\\
c_{22(22)} &= \frac{1}{\sigma_{12}} tr \left( \hatmat_p \boldsymbol{\psi}^{'} \boldsymbol{V}_R^*  \boldsymbol{\psi} \right).
\end{split}
\label{markov-t16c-1}
\end{align}
Using the above equation, we get that $c_{22(11)} c_{22(22)} - c^2_{22(12)} = \frac{1}{\sigma_{12}} tr \left( \hatmat_p \boldsymbol{\psi}^{'} \boldsymbol{V}_1^* \boldsymbol{\psi} \right) \times$\\ $tr \left( \hatmat_p \boldsymbol{\psi}^{'} \boldsymbol{V}_R^* \boldsymbol{\psi} \right)$. Since $\sigma_{12}>0$, applying \autoref{lemma-markov-1}, we get that $c_{22(11)} c_{22(22)} - c^2_{22(12)}>0 $. So from \eqref{p36c}, we get
\begin{align}
\boldsymbol{{C}}_{d^*(s2)(22)}^{-} &= \frac{t-1}{n \left( c_{22(11)} c_{22(22)} - c^2_{22(12)} \right)}
\begin{bmatrix}
c_{22(22)} & c_{22(12)}\\
c_{22(12)} & c_{22(11)}
\end{bmatrix}
\otimes \identity_t.
\label{markov-t17c-1}
\end{align}
From \eqref{p50c}, \eqref{p36c}, \eqref{markov-t16c-1} and \eqref{markov-t17c-1}, by calculation we get that $\boldsymbol{{C}}_{d^*(s2)} = \frac{n}{t-1} 
\begin{bmatrix}
\boldsymbol{\mathit{\Lambda}}_1 & \boldsymbol{\mathit{\Lambda}}_2\\
\boldsymbol{\mathit{\Lambda}}_2 & \boldsymbol{\mathit{\Lambda}}_4
\end{bmatrix}$, where
\begin{align*}
\begin{split}
\boldsymbol{\mathit{\Lambda}}_1 &= \left( tr \left( \boldsymbol{V}_1^* \right) + \frac{\bar{\rho}^2}{\sigma_{12}} tr \left( \boldsymbol{V}_R^* \right) - \frac{\left(  tr \left( \boldsymbol{V}_1^* \boldsymbol{\psi} \right) \right)^2}{tr \left( \hatmat_p \boldsymbol{\psi}^{'} \boldsymbol{V}_1^* \boldsymbol{\psi} \right) } - \frac{ \bar{\rho}^2}{\sigma_{12}} \frac{\left(tr \left( \boldsymbol{V}_R^* \boldsymbol{\psi} \right) \right)^2}{tr \left( \hatmat_p \boldsymbol{\psi}^{'} \boldsymbol{V}_R^* \boldsymbol{\psi} \right)} \right) \hatmat_t,\\
\boldsymbol{\mathit{\Lambda}}_2 &= - \frac{ \bar{\rho} }{\sigma_{12}} \left( tr \left(\boldsymbol{V}_R^* \right) -  \frac{\left(tr \left( \boldsymbol{V}_R^* \boldsymbol{\psi} \right) \right)^2}{tr \left( \hatmat_p \boldsymbol{\psi}^{'} \boldsymbol{V}_R^* \boldsymbol{\psi} \right)} \right) \hatmat_t, \text{ and}\\
\boldsymbol{\mathit{\Lambda}}_4 &= \frac{1}{\sigma_{12}} \left(  tr \left( \boldsymbol{V}_R^* \right) - \frac{\left(tr \left( \boldsymbol{V}_R^* \boldsymbol{\psi} \right) \right)^2}{tr \left( \hatmat_p \boldsymbol{\psi}^{'} \boldsymbol{V}_R^* \boldsymbol{\psi} \right)} \right) \hatmat_t.
\end{split}
\end{align*}

\noindent From the above equation, we get
\begin{align*}
tr \left( \boldsymbol{{C}}_{d^*(s2)} \right) &= n \left( tr \left( \boldsymbol{V}_1^* \right) + \frac{1 + \bar{\rho}^2}{\sigma_{12}} tr \left( \boldsymbol{V}_R^* \right) -\frac{\left(  tr \left( \boldsymbol{V}_1^* \boldsymbol{\psi} \right) \right)^2}{tr \left( \hatmat_p \boldsymbol{\psi}^{'} \boldsymbol{V}_1^* \boldsymbol{\psi} \right) } - \frac{1 + \bar{\rho}^2}{\sigma_{12}} \frac{\left(tr \left( \boldsymbol{V}_R^* \boldsymbol{\psi} \right) \right)^2}{tr \left( \hatmat_p \boldsymbol{\psi}^{'} \boldsymbol{V}_R^* \boldsymbol{\psi} \right)} \right) .
\end{align*}
\end{proof}

\begin{proof}[Proof of \autoref{markov-remark-2}]
\label{markov-ss10}
Here, $d^*$ is a design given by $OA_{I} \left( n=\lambda t \left(t-1 \right), p=t, t, 2 \right)$, where $\lambda$ is a positive integer and $t \geq 3$. From \autoref{thm5-c4}, we get that the information matrix for the direct effects corresponding to $d^*$ is expressed as
\begin{align}
\boldsymbol{{C}}_{d^*(s2)} &= \frac{n}{t-1} 
\begin{bmatrix}
\boldsymbol{\mathit{\Lambda}}_1 & \boldsymbol{\mathit{\Lambda}}_2\\
\boldsymbol{\mathit{\Lambda}}_2 & \boldsymbol{\mathit{\Lambda}}_4
\end{bmatrix},
\label{p48c}
\end{align}
where 
\begin{align*}
\begin{split}
\boldsymbol{\mathit{\Lambda}}_1 &= \left( tr \left( \boldsymbol{V}_1^* \right) + \frac{\bar{\rho}^2}{\sigma_{12}} tr \left( \boldsymbol{V}_R^* \right) - \frac{\left(  tr \left( \boldsymbol{V}_1^* \boldsymbol{\psi} \right) \right)^2}{tr \left( \hatmat_p \boldsymbol{\psi}^{'} \boldsymbol{V}_1^* \boldsymbol{\psi} \right) } - \frac{ \bar{\rho}^2}{\sigma_{12}} \frac{\left(tr \left( \boldsymbol{V}_R^* \boldsymbol{\psi} \right) \right)^2}{tr \left( \hatmat_p \boldsymbol{\psi}^{'} \boldsymbol{V}_R^* \boldsymbol{\psi} \right)} \right) \hatmat_t,\\
\boldsymbol{\mathit{\Lambda}}_2 &= - \frac{ \bar{\rho} }{\sigma_{12}} \left( tr \left(\boldsymbol{V}_R^* \right) -  \frac{\left(tr \left( \boldsymbol{V}_R^* \boldsymbol{\psi} \right) \right)^2}{tr \left( \hatmat_p \boldsymbol{\psi}^{'} \boldsymbol{V}_R^* \boldsymbol{\psi} \right)} \right) \hatmat_t, \text{ and}\\
\boldsymbol{\mathit{\Lambda}}_4 &= \frac{1}{\sigma_{12}} \left(  tr \left( \boldsymbol{V}_R^* \right) - \frac{\left(tr \left( \boldsymbol{V}_R^* \boldsymbol{\psi} \right) \right)^2}{tr \left( \hatmat_p \boldsymbol{\psi}^{'} \boldsymbol{V}_R^* \boldsymbol{\psi} \right)} \right) \hatmat_t.
\end{split}
\end{align*}
From the above equation, we get that $\boldsymbol{{C}}_{d^*(s2)}$ is completely symmetric if and only if 
\begin{multline}
\frac{1}{t} \left( tr \left( \boldsymbol{V}_1^* \right) + \frac{\bar{\rho}^2}{\sigma_{12}} tr \left( \boldsymbol{V}_R^* \right) - \frac{\left(  tr \left( \boldsymbol{V}_1^* \boldsymbol{\psi} \right) \right)^2}{tr \left( \hatmat_p \boldsymbol{\psi}^{'} \boldsymbol{V}_1^* \boldsymbol{\psi} \right) } - \frac{ \bar{\rho}^2}{\sigma_{12}} \frac{\left(tr \left( \boldsymbol{V}_R^* \boldsymbol{\psi} \right) \right)^2}{tr \left( \hatmat_p \boldsymbol{\psi}^{'} \boldsymbol{V}_R^* \boldsymbol{\psi} \right)} \right) =\\ \frac{1}{\sigma_{12} t} \left(  tr \left( \boldsymbol{V}_R^* \right) - \frac{\left(tr \left( \boldsymbol{V}_R^* \boldsymbol{\psi} \right) \right)^2}{tr \left( \hatmat_p \boldsymbol{\psi}^{'} \boldsymbol{V}_R^* \boldsymbol{\psi} \right)} \right)  = - \frac{ \bar{\rho} }{\sigma_{12} t } \left( tr \left(\boldsymbol{V}_R^* \right) -  \frac{\left(tr \left( \boldsymbol{V}_R^* \boldsymbol{\psi} \right) \right)^2}{tr \left( \hatmat_p \boldsymbol{\psi}^{'} \boldsymbol{V}_R^* \boldsymbol{\psi} \right)} \right)
\label{p51c}
\end{multline}
and
\begin{align}
\frac{ \bar{\rho} }{\sigma_{12} t } \left( tr \left(\boldsymbol{V}_R^* \right) -  \frac{\left(tr \left( \boldsymbol{V}_R^* \boldsymbol{\psi} \right) \right)^2}{tr \left( \hatmat_p \boldsymbol{\psi}^{'} \boldsymbol{V}_R^* \boldsymbol{\psi} \right)} \right) &= -\frac{ \bar{\rho} \left(t-1 \right)}{\sigma_{12} t } \left( tr \left(\boldsymbol{V}_R^* \right) -  \frac{\left(tr \left( \boldsymbol{V}_R^* \boldsymbol{\psi} \right) \right)^2}{tr \left( \hatmat_p \boldsymbol{\psi}^{'} \boldsymbol{V}_R^* \boldsymbol{\psi} \right)} \right).
\label{p52c}
\end{align}
From \autoref{lemma-markov-1}, we get that $tr \left( \hatmat_p \boldsymbol{\psi}^{'} \boldsymbol{V}_R^* \boldsymbol{\psi} \right)>0$. If we consider the univariate response case with the dispersion matrix of the vector error terms as $\sigma^2 \left( \identity_n \otimes \boldsymbol{V}_R \right)$, then from Lemma 6 of \citet{Martin1998Variance-balancedObservations}, we get that $tr \left(\boldsymbol{V}_R^* \right) tr \left( \hatmat_p \boldsymbol{\psi}^{'} \boldsymbol{V}_R^* \boldsymbol{\psi} \right) - \left(tr \left( \boldsymbol{V}_R^* \boldsymbol{\psi} \right) \right)^2 \neq 0$. Since $\sigma_{12}>0$ and $|\bar{\rho}|>0$, \eqref{p51c} and \eqref{p52c} does not hold true. Hence, $\boldsymbol{{C}}_{d^*(s2)}$ is not a completely symmetric matrix.
\end{proof}

\renewcommand{\thesection}{B}

\section{Efficiency under the generalized Markov-type structure}
\label{markov-efficiency}

\renewcommand{\thesection}{B}

\begin{proof}[Proof of \autoref{markov-lemma-1}]
\label{markov-ss1}
Here, $d^* \in \mathcal{D}^{(1)}_{t,n=\lambda t (t-1),p=t}$ is a design given by $OA_{I} \big( n=\lambda t \left(t-1 \right), p=t, t, 2 \big)$, where $\lambda$ is a positive integer and $t \geq 3$. Using the expression of $tr \left( \boldsymbol{{C}}_{d^*(s2)} \right)$ as given in \autoref{thm5-c4} and the expression of $u \left( t, n, p, \sigma_{11}, \sigma_{22}, \rho, \boldsymbol{V}_1, \boldsymbol{V}_R \right)$  as given in \autoref{thm4-c4}, where $p=t$ and $n=\lambda t \left(t-1 \right)$, by calculation we get 
\begin{multline}
u \left( t, n, p, \sigma_{11}, \sigma_{22}, \rho, \boldsymbol{V}_1, \boldsymbol{V}_R \right)  - tr \left( \boldsymbol{{C}}_{d^*(s2)} \right) \\= 
n \left( 1+ \bar{\rho}^2 \right) \frac{\left(  tr \left( \boldsymbol{V}_1^* \boldsymbol{\psi} \right) tr \left( \hatmat_p \boldsymbol{\psi}^{'} \boldsymbol{V}_R^* \boldsymbol{\psi} \right) - tr \left( \boldsymbol{V}_R^* \boldsymbol{\psi} \right) tr \left( \hatmat_p \boldsymbol{\psi}^{'} \boldsymbol{V}_1^* \boldsymbol{\psi} \right) \right)^2 }{  \left( \sigma_{12} tr \left( \hatmat_p \boldsymbol{\psi}^{'} \boldsymbol{V}_1^* \boldsymbol{\psi} \right) + \left( 1 + \bar{\rho}^2 \right) tr \left( \hatmat_p \boldsymbol{\psi}^{'} \boldsymbol{V}_R^* \boldsymbol{\psi} \right) \right) tr \left( \hatmat_p \boldsymbol{\psi}^{'} \boldsymbol{V}_1^* \boldsymbol{\psi} \right)  tr \left( \hatmat_p \boldsymbol{\psi}^{'} \boldsymbol{V}_R^* \boldsymbol{\psi} \right) }.
\label{markov-t20c-1}
\end{multline}
Thus from \eqref{markov-t20c-1}, we get that $d^*$ attains $u \left( t, n, p, \sigma_{11}, \sigma_{22}, \rho, \boldsymbol{V}_1, \boldsymbol{V}_R \right)$, if and only if 
\begin{align*}
tr \left( \boldsymbol{V}_1^* \boldsymbol{\psi} \right) tr \left( \hatmat_p \boldsymbol{\psi}^{'} \boldsymbol{V}_R^* \boldsymbol{\psi} \right) = tr \left( \boldsymbol{V}_R^* \boldsymbol{\psi} \right) tr \left( \hatmat_p \boldsymbol{\psi}^{'} \boldsymbol{V}_1^* \boldsymbol{\psi} \right).
\end{align*}
\end{proof}

\begin{proof}[Proof of \autoref{markov-lemma-sigma}]
\label{markov-ss2}
Here $d^{(0)} \in \mathcal{D}^{(1)}_{t,n=\lambda t \left( t-1 \right),p=t}$, where $\lambda$ is a positive integer, $p=t \geq 3$ and $\sigma_{11} = \sigma_{22}$. So, $\bar{\rho} = \rho \sqrt{\frac{\sigma_{22}}{\sigma_{11}}} = \rho$ and $\sigma_{12} = \sigma_{22} \left( 1- \bar{\rho}^2 \right) = \sigma_{11} \left( 1- \rho^2 \right)$. Since $\boldsymbol{V}_1 = \sigma_{11} \boldsymbol{V}_C$, from \autoref{thm4-c4}, we get
\begin{multline}
u=u \left( t, n, p, \sigma_{11}, \sigma_{22} = \sigma_{11}, \rho, \boldsymbol{V}_1, \boldsymbol{V}_R \right) = \frac{1}{\sigma_{11}} tr \left(  \boldsymbol{V}_C^*  \right) + \frac{1+\rho^2}{\sigma_{11} (1- \rho^2) } tr \left( \boldsymbol{V}_R^*  \right) \\-
\frac{1}{\sigma_{11}} \left( tr \left(  \boldsymbol{V}_C^* \boldsymbol{\psi} \right)  + \frac{1+\rho^2}{1-\rho^2} tr \left(   \boldsymbol{V}_R^* \boldsymbol{\psi} \right) \right)^2 / \left( tr \left( \hatmat_p \boldsymbol{\psi}^{'} \boldsymbol{V}_C^* \boldsymbol{\psi} \right) + \frac{1+\rho^2}{1 - \rho^2} tr \left( \hatmat_p \boldsymbol{\psi}^{'} \boldsymbol{V}_R^* \boldsymbol{\psi} \right) \right),
\label{markov-t22c-1}
\end{multline}
where $\boldsymbol{V}_C^* = \boldsymbol{V}_C^{-1} - \left( \vecone_{p}^{'} \boldsymbol{V}_C^{-1} \vecone_p \right)^{-1} \boldsymbol{V}_C^{-1} \matone_{p \times p} \boldsymbol{V}_C^{-1}$ and $\boldsymbol{V}_R^* = \boldsymbol{V}_R^{-1} - \left( \vecone_{p}^{'} \boldsymbol{V}_R^{-1} \vecone_p \right)^{-1} \boldsymbol{V}_R^{-1} \matone_{p \times p} \boldsymbol{V}_R^{-1}$, and using the expressions of $\boldsymbol{\mathit{\Omega}}_1$, $\boldsymbol{\mathit{\Omega}}_2$ and $\boldsymbol{\mathit{\Omega}}_4$ as given in \autoref{lemma3-c4-1}, we get
\begin{align}
\begin{split}
\boldsymbol{\mathit{\Omega}}_1 &= \frac{1}{\sigma_{11}} \boldsymbol{V}^*_C + \frac{\rho^2}{\sigma_{11} \left(1-\rho^2 \right)} \boldsymbol{V}_R^*,\\
\boldsymbol{\mathit{\Omega}}_2 &= \frac{\rho}{\sigma_{11} \left(1-\rho^2 \right)} \boldsymbol{V}_R^*, \text{ and}\\
\boldsymbol{\mathit{\Omega}}_4 &= \frac{1}{\sigma_{11} \left(1-\rho^2 \right)} \boldsymbol{V}_R^*.
\end{split}
\label{markov-t23c-1}
\end{align}
Let $\boldsymbol{B}^* = \sigma_{11} \boldsymbol{A}^*$. Hence, using the expression of $\boldsymbol{C}_{d(s2)}$ from \autoref{lemma5-c4}, we get
\begin{align}
\boldsymbol{C}_{d(s2)} &= \boldsymbol{C}_{d(s2)(11)} - \boldsymbol{C}_{{d(s2)}(12)} \boldsymbol{C}^{-}_{{d(s2)}(22)} \boldsymbol{C}_{{d(s2)}(21)},
\label{markov-t24c-1}
\end{align}
where $\boldsymbol{C}_{d(s2)(11)} = \frac{1}{\sigma_{11}} \left( \identity_2 \otimes \boldsymbol{T}^{'}_d \right) \boldsymbol{B}^* \left( \identity_2 \otimes \boldsymbol{T}_d \right)$, $\boldsymbol{C}_{d(s2)(12)} = \boldsymbol{C}^{'}_{d(s2)(21)} = \frac{1}{\sigma_{11}} \left( \identity_2 \otimes \boldsymbol{T}^{'}_d \right) \boldsymbol{B}^* \times$ \\ $\left( \identity_2 \otimes \boldsymbol{F}_d \right) $ and $\boldsymbol{C}_{d(s2)(22)} = \frac{1}{\sigma_{11}} \left( \identity_2 \otimes \boldsymbol{F}^{'}_d \right) \boldsymbol{B}^* \left( \identity_2 \otimes \boldsymbol{F}_d \right)$.\par

\noindent It should be noted that since the matrices $\boldsymbol{V}^*_C$ and $\boldsymbol{V}_R^*$ are independent of $\sigma_{11}$, from \eqref{markov-t23c-1} and using the expression of $\boldsymbol{A}^*$ in \autoref{lemma3-c4-1}, we get that the matrix $\boldsymbol{B}^*$ is independent of $\sigma_{11}$. Thus from \eqref{markov-t21c-1}, \eqref{markov-t22c-1} and \eqref{markov-t24c-1}, we get that $RD_{d^{(0)}}$ is independent of the value of $\sigma_{11}$.
\end{proof}

\end{appendices}

\section*{Declaration of competing interest}
The authors declare no conflicts of interest.

\section*{Funding}
This research did not receive any specific grant from funding agencies in the public, commercial, or not-for-profit sectors.

\section*{Author contributions}
\textbf{Shubham Niphadkar}: Formal analysis, Methodology, Software, Writing-Original draft preparation and Editing. \textbf{Siuli Mukhopadhyay}: Conceptualization, Methodology, Validation, Writing-Reviewing and Editing.


\AtBeginEnvironment{appendices}{%
 \titleformat{\section}{\bfseries\Large}{\appendixname}{0.5em}{}%
 \titleformat{\subsection}{\bfseries\large}{\thesubsection}{0.5em}{}%
\counterwithin{equation}{section}
}
\begin{appendices}
\renewcommand{\appendixname}{Supplementary Material}
\renewcommand{\thesection}{\hspace{-0.3em}}
\renewcommand{\thesubsection}{S.1}

\section{}
\label{supp-a}

\subsection{Table}
\label{supp-table}

\renewcommand{\thesection}{S}
\renewcommand{\thesubsection}{S.1}

\setlength{\arrayrulewidth}{0.5mm}
\setlength{\tabcolsep}{18pt}
\renewcommand{\arraystretch}{1.5}

\renewcommand{\thetable}{S.\arabic{table}}

\setcounter{table}{0}

\renewcommand{\thetable}{S.\arabic{table}}

\FloatBarrier
\begin{table}[ht]
\centering
\begin{tabular}{ p{4cm}|p{10cm}  }
\hline
Gene Marker & Characteristics\\
\hline
{\small{BPI (X100126967\_TGI\_AT) }} &  {\small{Defends against bacterial infections by neutralizing bacterial endotoxins within the innate immune system}}.\\
{\small{MTA3 (X100154937\_TGI\_AT)}} & {\small{Exhibits chromatin-binding and histone deacetylase activities, influencing transcriptional regulation and chromatin remodeling}}.\\
{\small{ LTBP2 (X100129645\_TGI\_AT) }} & {\small{ Crucial in the extracellular matrix, regulating transforming growth factor-beta (TGF-$\beta$) activity, influencing tissue homeostasis, and cellular behavior}}.\\
\hline
\end{tabular}
\caption{Characteristics of the $3$ gene markers. This information is obtained from GEO's annotation SOFT \href{https://www.ncbi.nlm.nih.gov/geo/query/acc.cgi?acc=GPL10379}{table}.}
\label{table}
\end{table}
\FloatBarrier

\renewcommand{\thesubsection}{S.2}
\subsection{Results for the proportional covariance structure}
\label{proportion-supp}

\begin{lemmas}
\label{lemma5-c4-i}
Under the proportional structure, for $g>1$, the information matrix for the direct effects can be expressed as
\begin{align}
\boldsymbol{C}_{d(s1)} &=  \boldsymbol{C}_{d(s1)(11)} - \boldsymbol{C}_{{d(s1)}(12)} \boldsymbol{C}^{-}_{{d 
(s1)}(22)} \boldsymbol{C}_{{d (s1)}(21)}, \label{z4-1}
\end{align}
where $\boldsymbol{C}_{d(s1)(11)} = \left( \identity_g \otimes \boldsymbol{T}^{'}_d \right) \boldsymbol{A}^* \left( \identity_g \otimes \boldsymbol{T}_d \right)$, $\boldsymbol{C}_{d(s1)(12)} = \boldsymbol{C}^{'}_{d(s1)(21)} = \left( \identity_g \otimes \boldsymbol{T}^{'}_d \right) \boldsymbol{A}^* \left( \identity_g \otimes \boldsymbol{F}_d \right) $ and $\boldsymbol{C}_{d(s1)(22)} = \left( \identity_g \otimes \boldsymbol{F}^{'}_d \right) \boldsymbol{A}^* \left( \identity_g \otimes \boldsymbol{F}_d \right)$. Here, $\boldsymbol{A}^* = \boldsymbol{\mathit{\Sigma}}^{-1/2} pr^{\perp} \left(\boldsymbol{\mathit{\Sigma}}^{-1/2} \boldsymbol{Z_1}
\right) \boldsymbol{\mathit{\Sigma}}^{-1/2}$.
\end{lemmas}
\begin{proof}
Let $\boldsymbol{\eta}^{(1)} = \begin{bmatrix}
\boldsymbol{\tau}^{'}_1 & \cdots & \boldsymbol{\tau}^{'}_g
\end{bmatrix}^{'}$ and $\boldsymbol{\eta}^{(2)} = \begin{bmatrix}
\mu_1 & \mu_2 &
\boldsymbol{\alpha}^{'}_1 &
\boldsymbol{\beta}^{'}_1 & 
\cdots &
\boldsymbol{\alpha}^{'}_g &
\boldsymbol{\beta}^{'}_g &
\boldsymbol{\rho}^{'}_1 &
\cdots &
\boldsymbol{\rho}^{'}_g
\end{bmatrix}^{'}
$. So $\boldsymbol{\eta}^{(1)}$ denotes the vector of direct effects. Then by rearranging the parameters, the model \eqref{model2} can be equivalently expressed as
\begin{align}
\begin{bmatrix}
 \boldsymbol{Y}^{'}_{d1} &
 \cdots &
\boldsymbol{Y}^{'}_{dg}
\end{bmatrix}^{'}
&= 
\left( \identity_g \otimes \boldsymbol{T}_d \right) \boldsymbol{\eta}^{(1)} + 
\begin{bmatrix}
 \identity_g \otimes \vecone_{np}  & \boldsymbol{Z_1} & \identity_g \otimes \boldsymbol{F}_d
\end{bmatrix}
\boldsymbol{\eta}^{(2)} + 
\begin{bmatrix}
 \boldsymbol{\eps}^{'}_{1} &
 \cdots &
\boldsymbol{\eps}^{'}_{g}
\end{bmatrix}^{'}.
\label{markov-t7c-1-i}
\end{align}

\noindent Premultiplying \eqref{markov-t7c-1-i} by $\boldsymbol{\mathit{\Sigma}}^{-1/2}$, we get the model as
\begin{multline}
\begin{bmatrix}
 \boldsymbol{Y}^{'}_{d1(new)} &
 \cdots &
\boldsymbol{Y}^{'}_{dg(new)}
\end{bmatrix}^{'}
= 
\boldsymbol{\mathit{\Sigma}}^{-1/2}
\left( \identity_g \otimes \boldsymbol{T}_d \right) \boldsymbol{\eta}^{(1)} + 
\boldsymbol{\mathit{\Sigma}}^{-1/2}
\begin{bmatrix}
 \identity_g \otimes \vecone_{np}  & \boldsymbol{Z_1} & \identity_g \otimes \boldsymbol{F}_d
\end{bmatrix}
\boldsymbol{\eta}^{(2)} + \\
\begin{bmatrix}
 \boldsymbol{\eps}^{'}_{1(new)} &
 \cdots &
\boldsymbol{\eps}^{'}_{g(new)}
\end{bmatrix}^{'}.
\label{markov-t8c-1-i}
\end{multline}

\noindent Thus using the expression of the information matrix from \citet{kunert1983optimal}, we get that under the proportional structure, the information matrix for the direct effects can be expressed as
\begin{align}
\boldsymbol{C}_{d(s1)} &=
\left( \identity_g \otimes \boldsymbol{T}_d^{'} \right) \boldsymbol{\mathit{\Sigma}}^{-1/2} pr^{\perp} \left( \boldsymbol{\mathit{\Sigma}}^{-1/2} 
\begin{bmatrix}
\boldsymbol{Z_1} & \identity_g \otimes \boldsymbol{F}_d
\end{bmatrix}  
\right) \boldsymbol{\mathit{\Sigma}}^{-1/2} \left( \identity_g \otimes \boldsymbol{T}_d\right).
\label{markov-t9c-1-i}
\end{align}

\noindent From \eqref{markov-t9c-1-i} and using the result $pr^{\perp} \left( \begin{bmatrix}
\boldsymbol{\mathit{\zeta}_1} & \boldsymbol{\mathit{\zeta}_2} 
\end{bmatrix} \right) = pr^{\perp} \left( \boldsymbol{\mathit{\zeta}_1} \right) - pr^{\perp} \left( \boldsymbol{\mathit{\zeta}_1} \right) \boldsymbol{\mathit{\zeta}_2} \left( \boldsymbol{\mathit{\zeta}}^{'}_\mathbf{2} pr^{\perp} \left( \boldsymbol{\mathit{\zeta}_1} \right)\boldsymbol{\mathit{\zeta}_2} \right)^{-} \times $\\$\boldsymbol{\mathit{\zeta}}^{'}_\mathbf{2} pr^{\perp} \left( \boldsymbol{\mathit{\zeta}_1} \right) $ \citep{kunert1983optimal}, we get
\begin{align}
\boldsymbol{C}_{d(s1)} &=
\boldsymbol{C}_{d(s1)(11)} - \boldsymbol{C}_{{d(s1)}(12)} \boldsymbol{C}^{-}_{{d 
(s1)}(22)} \boldsymbol{C}_{{d (s1)}(21)},
\label{markov-t10c-1-i}
\end{align}
where $\boldsymbol{C}_{d(s1)(11)} = \left( \identity_g \otimes \boldsymbol{T}^{'}_d \right) \boldsymbol{A}^* \left( \identity_g \otimes \boldsymbol{T}_d \right)$, $\boldsymbol{C}_{d(s1)(12)} = \boldsymbol{C}^{'}_{d(s1)(21)} = \left( \identity_g \otimes \boldsymbol{T}^{'}_d \right) \boldsymbol{A}^* \left( \identity_g \otimes \boldsymbol{F}_d \right) $ and $\boldsymbol{C}_{d(s1)(22)} = \left( \identity_g \otimes \boldsymbol{F}^{'}_d \right) \boldsymbol{A}^* \left( \identity_g \otimes \boldsymbol{F}_d \right)$. Here, $\boldsymbol{A}^* = \boldsymbol{\mathit{\Sigma}}^{-1/2} pr^{\perp} \left(\boldsymbol{\mathit{\Sigma}}^{-1/2} \boldsymbol{Z_1}
\right) \boldsymbol{\mathit{\Sigma}}^{-1/2}$.
\end{proof}

\begin{lemmas}
\label{prop-lemma3-c4-1}
Under the proportional structure, for $g>1$, the matrix $\boldsymbol{A}^*$ as given in \autoref{lemma5-c4-i} can be expressed as 
\begin{align}
\boldsymbol{A}^* &=
\boldsymbol{\mathit{\Gamma}}^{-1} \otimes \left( \hatmat_n \otimes \boldsymbol{V}^* \right),
\label{prop-p31c-1}
\end{align}
where $\boldsymbol{V}^* = \boldsymbol{V}^{-1} - \left( \vecone_{p}^{'} \boldsymbol{V}^{-1} \vecone_p \right)^{-1} \boldsymbol{V}^{-1} \matone_{p \times p} \boldsymbol{V}^{-1}$, and $\boldsymbol{\mathit{\Gamma}}$ is a $g \times g$ positive definite and symmetric matrix with non-zero off-diagonal matrix.
\end{lemmas}

\begin{proof}
Under the proportional structure, for $g>1$, the dispersion matrix of the vector of error terms is given as
\begin{align}
\boldsymbol{\mathit{\Sigma}} &= \boldsymbol{\mathit{\Gamma}} \otimes \left( \identity_n \otimes \boldsymbol{V} \right),
\label{prop-t17c-1}
\end{align}
where $\boldsymbol{V}$ is a known positive definite and symmetric matrix. $\boldsymbol{\mathit{\Gamma}}$ is a $g \times g$ positive definite and symmetric matrix with non-zero off-diagonal elements.

\noindent Using \eqref{Zd} and properties of kronecker product, we get
\begin{align}
\begin{split}
\boldsymbol{\mathit{\Sigma}}^{-1} &= 
\boldsymbol{\mathit{\Gamma}}^{-1} \otimes \left(\identity_n \otimes \boldsymbol{V}^{-1} \right),\\
\boldsymbol{\mathit{\Sigma}}^{-1} \boldsymbol{Z_1} &= 
\boldsymbol{\mathit{\Gamma}}^{-1} \otimes \left(\identity_n \otimes \boldsymbol{V}^{-1} \right) \boldsymbol{X_1}, \text{ and}\\
\boldsymbol{Z}^{'}_{\mathbf{1}} \boldsymbol{\mathit{\Sigma}}^{-1} \boldsymbol{Z_1} &= 
\boldsymbol{\mathit{\Gamma}}^{-1} \otimes \boldsymbol{X}^{'}_{\mathbf{1}} \left(\identity_n \otimes \boldsymbol{V}^{-1} \right) \boldsymbol{X_1}.
\end{split}
\label{prop-t18c-1}
\end{align}

\noindent From \citet{Bose2009OptimalDesigns}, we get that  $\left[ \left( \identity_n \otimes \boldsymbol{V}^{-1/2} \right) pr^{\perp} \left( \left( \identity_n \otimes \boldsymbol{V}^{-1/2} \right) \boldsymbol{X_1} \right)  \left( \identity_n \otimes \boldsymbol{V}^{-1/2} \right) \right] = \hatmat_n \otimes \boldsymbol{V}^*$, where $\boldsymbol{V}^* = \boldsymbol{V}^{-1} - \left( \vecone_{p}^{'} \boldsymbol{V}^{-1} \vecone_p \right)^{-1} \boldsymbol{V}^{-1} \matone_{p \times p} \boldsymbol{V}^{-1}$. Thus from \eqref{prop-t18c-1} and using the expression $\boldsymbol{A}^* = \boldsymbol{\mathit{\Sigma}}^{-1/2} pr^{\perp} \left(\boldsymbol{\mathit{\Sigma}}^{-1/2} \boldsymbol{Z_1}
\right) \boldsymbol{\mathit{\Sigma}}^{-1/2}$ as given in \autoref{lemma5-c4-i}, we get that for $g>1$
\begin{align*}
\boldsymbol{A}^* &= 
\boldsymbol{\mathit{\Gamma}}^{-1} \otimes
\left[ \left( \identity_n \otimes \boldsymbol{V}^{-1/2} \right) pr^{\perp} \left( \left( \identity_n \otimes \boldsymbol{V}^{-1/2} \right) \boldsymbol{X_1} \right)  \left( \identity_n \otimes \boldsymbol{V}^{-1/2} \right) \right]\\
&= \boldsymbol{\mathit{\Gamma}}^{-1} \otimes \left(\hatmat_n \otimes \boldsymbol{V}^* \right).
\end{align*}
\end{proof}

\renewcommand{\thesubsection}{S.3}
\subsection{Results for the generalized Markov-type covariance structure}
\label{markov-supp}

\begin{lemmas}
\label{lemma5-c4}
Under the generalized Markov-type structure, the information matrix for the direct effects can be expressed as
\begin{align}
\boldsymbol{C}_{d(s2)} &=  \boldsymbol{C}_{d(s2)(11)} - \boldsymbol{C}_{{d(s2)}(12)} \boldsymbol{C}^{-}_{{d 
(s2)}(22)} \boldsymbol{C}_{{d (s2)}(21)} \label{z4}\\
&=
\left( \identity_2 \otimes \boldsymbol{T}_d^{'} \right) \boldsymbol{\mathit{\Sigma}}^{-1/2} pr^{\perp} \left( \boldsymbol{\mathit{\Sigma}}^{-1/2} 
\begin{bmatrix}
\boldsymbol{Z_1} & \identity_2 \otimes \boldsymbol{F}_d \hatmat_t
\end{bmatrix}  
\right) \boldsymbol{\mathit{\Sigma}}^{-1/2} \left( \identity_2 \otimes \boldsymbol{T}_d\right)
\label{z42}\\
&= \boldsymbol{C}_{{d(s2)}(11)(1)} - \boldsymbol{C}_{{d(s2)}(12)(1)} \boldsymbol{C}^{-}_{{d (s2)}(22)(1)} \boldsymbol{C}_{{d (s2)}(21)(1)},
\label{z43}
\end{align}
where $\boldsymbol{C}_{d(s2)(11)} = \left( \identity_2 \otimes \boldsymbol{T}^{'}_d \right) \boldsymbol{A}^* \left( \identity_2 \otimes \boldsymbol{T}_d \right)$, $\boldsymbol{C}_{d(s2)(12)} = \boldsymbol{C}^{'}_{d(s2)(21)} = \left( \identity_2 \otimes \boldsymbol{T}^{'}_d \right) \boldsymbol{A}^* \left( \identity_2 \otimes \boldsymbol{F}_d \right) $, \\$\boldsymbol{C}_{d(s2)(22)} = \left( \identity_2 \otimes \boldsymbol{F}^{'}_d \right) \boldsymbol{A}^* \left( \identity_2 \otimes \boldsymbol{F}_d \right)$, $\boldsymbol{C}_{{d(s2)}(11)(1)} = \left( \identity_2 \otimes \boldsymbol{T}_d^{'} \right) \boldsymbol{A}^* \left( \identity_2 \otimes \boldsymbol{T}_d\right)$, $\boldsymbol{C}_{{d (s2)}(12)(1)} =$\\$ \boldsymbol{C}^{'}_{{d(s2)}(21)(1)}=\left( \identity_2 \otimes \boldsymbol{T}_d^{'} \right) \boldsymbol{A}^* \left(\identity_2 \otimes \boldsymbol{F}_d \hatmat_t\right)$ and $\boldsymbol{C}_{{d (s2)}(22)(1)} = \left(\identity_2 \otimes \boldsymbol{F}^{'}_d \hatmat_t\right) \boldsymbol{A}^* \left(\identity_2 \otimes \boldsymbol{F}_d \hatmat_t\right)$. Here, $\boldsymbol{A}^* = \boldsymbol{\mathit{\Sigma}}^{-1/2} pr^{\perp} \left(\boldsymbol{\mathit{\Sigma}}^{-1/2} \boldsymbol{Z_1}
\right) \boldsymbol{\mathit{\Sigma}}^{-1/2}$.
\end{lemmas}
\begin{proof}
Under the generalized Markov-type structure, following the proof of \autoref{lemma5-c4-i} till \eqref{markov-t9c-1-i} by taking $g=2$, we get
\begin{align}
\boldsymbol{C}_{d(s2)} &=
\left( \identity_2 \otimes \boldsymbol{T}_d^{'} \right) \boldsymbol{\mathit{\Sigma}}^{-1/2} pr^{\perp} \left( \boldsymbol{\mathit{\Sigma}}^{-1/2} 
\begin{bmatrix}
\boldsymbol{Z_1} & \identity_2 \otimes \boldsymbol{F}_d
\end{bmatrix}  
\right) \boldsymbol{\mathit{\Sigma}}^{-1/2} \left( \identity_2 \otimes \boldsymbol{T}_d\right).
\label{markov-t9c-1}
\end{align}
From \eqref{markov-t9c-1} and using the result $pr^{\perp} \left( \begin{bmatrix}
\boldsymbol{\mathit{\zeta}_1} & \boldsymbol{\mathit{\zeta}_2} 
\end{bmatrix} \right) = pr^{\perp} \left( \boldsymbol{\mathit{\zeta}_1} \right) - pr^{\perp} \left( \boldsymbol{\mathit{\zeta}_1} \right) \boldsymbol{\mathit{\zeta}_2} \left( \boldsymbol{\mathit{\zeta}}^{'}_\mathbf{2} pr^{\perp} \left( \boldsymbol{\mathit{\zeta}_1} \right)\boldsymbol{\mathit{\zeta}_2} \right)^{-} \times $\\$ \boldsymbol{\mathit{\zeta}}^{'}_\mathbf{2} pr^{\perp} \left( \boldsymbol{\mathit{\zeta}_1} \right) $ \citep{kunert1983optimal}, we get
\begin{align}
\boldsymbol{C}_{d(s2)} &=
\boldsymbol{C}_{d(s2)(11)} - \boldsymbol{C}_{{d(s2)}(12)} \boldsymbol{C}^{-}_{{d 
(s2)}(22)} \boldsymbol{C}_{{d (s2)}(21)},
\label{markov-t10c-1}
\end{align}
where $\boldsymbol{C}_{d(s2)(11)} = \left( \identity_2 \otimes \boldsymbol{T}^{'}_d \right) \boldsymbol{A}^* \left( \identity_2 \otimes \boldsymbol{T}_d \right)$, $\boldsymbol{C}_{d(s2)(12)} = \boldsymbol{C}^{'}_{d(s2)(21)} = \left( \identity_2 \otimes \boldsymbol{T}^{'}_d \right) \boldsymbol{A}^* \left( \identity_2 \otimes \boldsymbol{F}_d \right) $ and $\boldsymbol{C}_{d(s2)(22)} = \left( \identity_2 \otimes \boldsymbol{F}^{'}_d \right) \boldsymbol{A}^* \left( \identity_2 \otimes \boldsymbol{F}_d \right)$. Here, $\boldsymbol{A}^* = \boldsymbol{\mathit{\Sigma}}^{-1/2} pr^{\perp} \left(\boldsymbol{\mathit{\Sigma}}^{-1/2} \boldsymbol{Z_1}
\right) \boldsymbol{\mathit{\Sigma}}^{-1/2}$.\par

\noindent Since $
\boldsymbol{F}_d \vecone_t = \vecone_n \otimes \begin{bmatrix} 0 & \vecone_{p-1}^{'} \end{bmatrix}^{'} = \boldsymbol{P} \begin{bmatrix} 0 & \vecone_{p-1}^{'} \end{bmatrix}^{'}$, we get $\left( \identity_2 \otimes \frac{1}{t} \boldsymbol{F}_d \matone_{t \times t} \right) =  \left( \identity_2 \otimes \boldsymbol{P} \right) \times$\\$\left( \identity_2 \otimes  \frac{1}{t} \begin{bmatrix} 0 & \vecone_{p-1}^{'} \end{bmatrix}^{'} \vecone_t^{'} \right) $. Hence, $\identity_2 \otimes \frac{1}{t} \boldsymbol{F}_d \matone_{t \times t}$ belongs to the column space of $\identity_2 \otimes \boldsymbol{P}$. So \\$\boldsymbol{\mathit{\Sigma}}^{-1/2} \left( \identity_2 \otimes \frac{1}{t} \boldsymbol{F}_d \matone_{t \times t} \right)$ belongs to the column space of $\boldsymbol{\mathit{\Sigma}}^{-1/2} \boldsymbol{Z_1}$. Thus from \eqref{markov-t9c-1}, we can prove \eqref{z42}.\par

\noindent Further using the result $pr^{\perp} \left( \begin{bmatrix}
\boldsymbol{\mathit{\zeta}_1} & \boldsymbol{\mathit{\zeta}_2} 
\end{bmatrix} \right) = pr^{\perp} \left( \boldsymbol{\mathit{\zeta}_1} \right) - pr^{\perp} \left( \boldsymbol{\mathit{\zeta}_1} \right) \boldsymbol{\mathit{\zeta}_2} \left( \boldsymbol{\mathit{\zeta}}^{'}_\mathbf{2} pr^{\perp} \left( \boldsymbol{\mathit{\zeta}_1} \right)\boldsymbol{\mathit{\zeta}_2} \right)^{-} \boldsymbol{\mathit{\zeta}}^{'}_\mathbf{2} pr^{\perp} \left( \boldsymbol{\mathit{\zeta}_1} \right) $ \citep{kunert1983optimal}, we get
\begin{align}
\boldsymbol{C}_{d(s2)} &= \boldsymbol{C}_{{d(s2)}(11)(1)} - \boldsymbol{C}_{{d(s2)}(12)(1)} \boldsymbol{C}^{-}_{{d (s2)}(22)(1)} \boldsymbol{C}_{{d (s2)}(21)(1)},
\end{align}
where $\boldsymbol{C}_{{d(s2)}(11)(1)} = \left( \identity_2 \otimes \boldsymbol{T}_d^{'} \right) \boldsymbol{A}^* \left( \identity_2 \otimes \boldsymbol{T}_d\right)$, $\boldsymbol{C}_{{d (s2)}(12)(1)} = \boldsymbol{C}^{'}_{{d(s2)}(21)(1)}=\left( \identity_2 \otimes \boldsymbol{T}_d^{'} \right) \boldsymbol{A}^* \times$\\$ \left(\identity_2 \otimes \boldsymbol{F}_d \hatmat_t\right)$ and $\boldsymbol{C}_{{d (s2)}(22)(1)} = \left(\identity_2 \otimes \boldsymbol{F}^{'}_d \hatmat_t\right) \boldsymbol{A}^*  \left(\identity_2 \otimes \boldsymbol{F}_d \hatmat_t\right)$.
\end{proof}

\begin{lemmas}
\label{lemma3-c4-1}
Under the generalized Markov-type structure, the matrix $\boldsymbol{A}^*$ as given in \autoref{lemma5-c4} can be expressed as
\begin{align}
\boldsymbol{A}^* &= 
\begin{bmatrix}
\hatmat_n \otimes \boldsymbol{\mathit{\Omega}}_1 & -\hatmat_n \otimes \boldsymbol{\mathit{\Omega}}_2\\
-\hatmat_n \otimes \boldsymbol{\mathit{\Omega}}_2 & \hatmat_n \otimes \boldsymbol{\mathit{\Omega}}_4
\end{bmatrix},
\label{p31c-1}
\end{align}
where
\begin{align}
\begin{split}
\boldsymbol{\mathit{\Omega}}_1 &= \boldsymbol{V}_1^* + \frac{\bar{\rho}^2}{\sigma_{12}} \boldsymbol{V}_R^*,\\
\boldsymbol{\mathit{\Omega}}_2 &=  \frac{\bar{\rho}}{\sigma_{12}} \boldsymbol{V}_R^*, \text{ and}\\
\boldsymbol{\mathit{\Omega}}_4 &= \frac{1}{\sigma_{12}} \boldsymbol{V}_R^*.
\end{split}
\label{markov-t1c-1}
\end{align}
Here, $\boldsymbol{V}_1^* = \boldsymbol{V}_1^{-1} - \left( \vecone_{p}^{'} \boldsymbol{V}_1^{-1} \vecone_p \right)^{-1} \boldsymbol{V}_1^{-1} \matone_{p \times p} \boldsymbol{V}_1^{-1}$, $\boldsymbol{V}_R^* = \boldsymbol{V}_R^{-1} - \left( \vecone_{p}^{'} \boldsymbol{V}_R^{-1} \vecone_p \right)^{-1} \boldsymbol{V}_R^{-1} \matone_{p \times p} \boldsymbol{V}_R^{-1}$, $\bar{\rho} = \rho \sqrt{\frac{\sigma_{22}}{\sigma_{11}}}$ and $\sigma_{12} = \sigma_{22} \left(1 
- \rho^2 \right)$.
\end{lemmas}

\begin{proof}
From \eqref{dispepsc}, under the generalized Markov-type structure,
\begin{align}
\boldsymbol{\mathit{\Sigma}} =
\begin{bmatrix}
\boldsymbol{\mathit{\Sigma}}_{11} & \boldsymbol{\mathit{\Sigma}}_{12}\\
\boldsymbol{\mathit{\Sigma}}_{21} & \boldsymbol{\mathit{\Sigma}}_{22}
\end{bmatrix},
\label{t17c-1}
\end{align}
where $\boldsymbol{\mathit{\Sigma}}_{11} = \identity_n \otimes \boldsymbol{V}_1$, $\boldsymbol{\mathit{\Sigma}}_{12} = \boldsymbol{\mathit{\Sigma}}_{21} = \rho \sqrt{\frac{\sigma_{22}}{\sigma_{11}} }\boldsymbol{\mathit{\Sigma}}_{11}$, $\boldsymbol{\mathit{\Sigma}}_{22} = \rho^2 \frac{\sigma_{22}}{\sigma_{11}} \boldsymbol{\mathit{\Sigma}}_{11} + \sigma_{22} (1-\rho^2) \boldsymbol{\mathit{\Sigma}}_{R}$ and $\boldsymbol{\mathit{\Sigma}}_{R} = \identity_n \otimes \boldsymbol{V}_R$. Here, $\boldsymbol{V}_1$ is a known positive definite and symmetric matrix with each diagonal element as $\sigma_{11}>0$, and $\boldsymbol{V}_R$ is a known correlation matrix. Also, $\sigma_{11} , \sigma_{22} >0$ and $0<|\rho| <1$. So
\begin{align}
\boldsymbol{\mathit{\Sigma}}^{-1} &= 
\begin{bmatrix}
\boldsymbol{\mathit{\Sigma}}^{11} & \boldsymbol{\mathit{\Sigma}}^{12}\\
\boldsymbol{\mathit{\Sigma}}^{21} & \boldsymbol{\mathit{\Sigma}}^{22}
\end{bmatrix},
\label{t18c-1}
\end{align}
where $\boldsymbol{\mathit{\Sigma}}^{11} = \boldsymbol{\mathit{\Sigma}}_{11}^{-1} + \boldsymbol{\mathit{\Sigma}}_{11}^{-1} \boldsymbol{\mathit{\Sigma}}_{12} \boldsymbol{\mathit{\Sigma}}^{22} \boldsymbol{\mathit{\Sigma}}_{21} \boldsymbol{\mathit{\Sigma}}_{11}^{-1}$, $\boldsymbol{\mathit{\Sigma}}^{12} = \left(\boldsymbol{\mathit{\Sigma}}^{21} \right)^{'}= - \boldsymbol{\mathit{\Sigma}}_{11}^{-1} \boldsymbol{\mathit{\Sigma}}_{12} \boldsymbol{\mathit{\Sigma}}^{22}$ and $\boldsymbol{\mathit{\Sigma}}^{22} = $\\$ \left( \boldsymbol{\mathit{\Sigma}}_{11} - \boldsymbol{\mathit{\Sigma}}_{12} \boldsymbol{\mathit{\Sigma}}^{-1}_{22} \boldsymbol{\mathit{\Sigma}}_{21}\right)^{-1}$.\par

\noindent Let $\bar{\rho} = \rho \sqrt{\frac{\sigma_{22}}{\sigma_{11}}}$ and $\sigma_{12} = \sigma_{22} \left( 1 - \rho^2 \right)$. Note that since $\sigma_{11}, \sigma_{22} >0$ and $0<|\rho|<1$, $|\bar{\rho}|>0$ and $\sigma_{12}>0$. By calculation we get
\begin{align}
\begin{split}
\boldsymbol{\mathit{\Sigma}}^{11} &= \identity_n \otimes\boldsymbol{K}_1^{-1}, \text{ and}\\
\boldsymbol{\mathit{\Sigma}}^{12} &=\boldsymbol{\mathit{\Sigma}}^{21} = -\bar{\rho} \boldsymbol{\mathit{\Sigma}}^{22} = - \identity_n \otimes \bar{\rho} \boldsymbol{K}_2^{-1}. 
\label{t19c-1}
\end{split}
\end{align}
where $\boldsymbol{K}_1^{-1} = \boldsymbol{V}_1^{-1} + \bar{\rho}^2 \boldsymbol{K}_2^{-1}$ and $\boldsymbol{K}_2 = \sigma_{12} \boldsymbol{V}_R$. Note that $\boldsymbol{K}_1$ and $\boldsymbol{K}_2$ are symmetric matrices. Since, $\boldsymbol{V}_R$ is positive definite matrix, $\sigma_{12}>0$ and $|\bar{\rho}|>0$, we get that $\boldsymbol{K}_2$ is a positive definite matrix. Since $\boldsymbol{\mathit{\Sigma}}$ is positive definite, $\boldsymbol{\mathit{\Sigma}}^{11}$ is also a positive definite matrix. Thus, $\boldsymbol{K}_1$ is a positive definite matrix.\par

\noindent Now, we will find $\left( \boldsymbol{Z}_1^{'} \boldsymbol{\mathit{\Sigma}}^{-1} \boldsymbol{Z}_1 \right)^{-}$. Since from \eqref{t19c-1}, $\boldsymbol{\mathit{\Sigma}}^{12} = \boldsymbol{\mathit{\Sigma}}^{21}$, using \eqref{Zd} and \eqref{t18c-1}, we get
\begin{align*}
\boldsymbol{Z}^{'}_{1} \boldsymbol{\mathit{\Sigma}}^{-1} \boldsymbol{Z}_{1} &= 
\begin{bmatrix}
\boldsymbol{X}^{'}_{1} \boldsymbol{\mathit{\Sigma}}^{11} \boldsymbol{X}_{1} & \boldsymbol{X}^{'}_{1} \boldsymbol{\mathit{\Sigma}}^{12} \boldsymbol{X}_{1}\\
\boldsymbol{X}^{'}_{1} \boldsymbol{\mathit{\Sigma}}^{12} \boldsymbol{X}_{1} & \boldsymbol{X}^{'}_{1} \boldsymbol{\mathit{\Sigma}}^{22} \boldsymbol{X}_{1}
\end{bmatrix}.
\end{align*}
Let $\boldsymbol{A}^{(1)} = \boldsymbol{X}^{'}_{1} \boldsymbol{\mathit{\Sigma}}^{11} \boldsymbol{X}_{1}$, $\boldsymbol{B}^{(1)}  = \boldsymbol{X}^{'}_{1} \boldsymbol{\mathit{\Sigma}}^{12} \boldsymbol{X}_{1}$ and $\boldsymbol{D}^{(1)} = \boldsymbol{X}^{'}_{1} \boldsymbol{\mathit{\Sigma}}^{22} \boldsymbol{X}_{1}$. So
\begin{align}
\left( \boldsymbol{Z}^{'}_{1} \boldsymbol{\mathit{\Sigma}}^{-1} \boldsymbol{Z}_{1} \right)^{-}&= 
\begin{bmatrix}
\boldsymbol{L}_1 & \boldsymbol{L}_2\\
\boldsymbol{L}_3 & \boldsymbol{L}_4
\end{bmatrix},
\label{t22c-1}
\end{align}
where $\boldsymbol{L}_1 = \left(\boldsymbol{A}^{(1)} \right)^{-} + \left(\boldsymbol{A}^{(1)} \right)^{-} \boldsymbol{B}^{(1)} \left( \boldsymbol{Z}^{(1)} \right)^{-} \boldsymbol{B}^{(1)} \left(\boldsymbol{A}^{(1)} \right)^{-}$, $\boldsymbol{L}_2 = \boldsymbol{L}_3^{'} =  - \left(\boldsymbol{A}^{(1)} \right)^{-} \boldsymbol{B}^{(1)}\left( \boldsymbol{Z}^{(1)} \right)^{-}$ and $\boldsymbol{L}_4 = \left( \boldsymbol{Z}^{(1)} \right)^{-}$. Here, $ \boldsymbol{Z}^{(1)} = \boldsymbol{D}^{(1)} - \boldsymbol{B}^{(1)} \left( \boldsymbol{A}^{(1)} \right)^{-} \boldsymbol{B}^{(1)} $.
From \eqref{Zd} and \eqref{t19c-1}, by calculation we get
\begin{align}
\boldsymbol{A}^{(1)} &=
\begin{bmatrix}
\boldsymbol{A}^{(2)} & \boldsymbol{B}^{(2)}\\
\boldsymbol{C}^{(2)} & \boldsymbol{D}^{(2)}
\end{bmatrix},
\label{t23c-1}
\end{align}
where $\boldsymbol{A}^{(2)} =  n \boldsymbol{K}_1^{-1}$, $\boldsymbol{B}^{(2)} = \left( \boldsymbol{C}^{(2)} \right)^{'} = \vecone_n^{'} \otimes  \boldsymbol{K}_1^{-1} \vecone_p$ and $\boldsymbol{D}^{(2)} = \delta_1^{-1} \identity_n$, where  $\delta_1 = 1/\vecone_p^{'} \boldsymbol{K}_1^{-1} \vecone_p$. Note that since $\boldsymbol{K}_1$ is a positive definite matrix, $\boldsymbol{K}_1^{-1}$ is a positive definite matrix and thus $\delta_1>0$.\par

\noindent Hence
\begin{align}
\left( \boldsymbol{A}^{(1)} \right)^{-} &=
\begin{bmatrix}
\left(\boldsymbol{A}^{(2)} \right)^{-} + \left(\boldsymbol{A}^{(2)} \right)^{-} \boldsymbol{B}^{(2)} \left( \boldsymbol{Z}^{(2)} \right)^{-} \boldsymbol{C}^{(2)} \left(\boldsymbol{A}^{(2)} \right)^{-} & - \left(\boldsymbol{A}^{(2)} \right)^{-} \boldsymbol{B}^{(2)}\left( \boldsymbol{Z}^{(2)} \right)^{-}\\
- \left( \boldsymbol{Z}^{(2)} \right)^{-} \boldsymbol{C}^{(2)} \left(\boldsymbol{A}^{(2)} \right)^{-}   & \left( \boldsymbol{Z}^{(2)} \right)^{-}
\end{bmatrix},
\label{t24c-1}
\end{align}
where $\boldsymbol{Z}^{(2)} = \boldsymbol{D}^{(2)} - \boldsymbol{C}^{(2)} \left( \boldsymbol{A}^{(2)} \right)^{-} \boldsymbol{B}^{(2)} $. By calculation we get
\begin{align}
\begin{split}
\boldsymbol{Z}^{(2)} &= \delta_1^{-1} \hatmat_{n},\\
\left(\boldsymbol{A}^{(2)} \right)^{-} \boldsymbol{B}^{(2)}\left( \boldsymbol{Z}^{(2)} \right)^{-} &= \frac{1}{n}  \vecone_n^{'} \otimes \delta_1 \vecone_p ,\\
\left( \boldsymbol{Z}^{(2)} \right)^{-} \boldsymbol{C}^{(2)} \left(\boldsymbol{A}^{(2)} \right)^{-}  &= \frac{1}{n}  \vecone_n \otimes \delta_1 \vecone^{'}_p, \text{ and}\\
\left(\boldsymbol{A}^{(2)} \right)^{-} + \left(\boldsymbol{A}^{(2)} \right)^{-} \boldsymbol{B}^{(2)} \left( \boldsymbol{Z}^{(2)} \right)^{-} \boldsymbol{C}^{(2)} \left(\boldsymbol{A}^{(2)} \right)^{-} &= \frac{1}{n} \boldsymbol{K}_1 + \frac{1}{n} \delta_1 \matone_{p \times p}.
\label{t28c-1}
\end{split}
\end{align}

\noindent So using \eqref{t28c-1}, \eqref{t24c-1} becomes
\begin{align}
\left(\boldsymbol{A}^{(1)} \right)^{-} &= 
\begin{bmatrix}
\frac{1}{n} \boldsymbol{K}_1 + \frac{1}{n} \delta_1 \matone_{p \times p} &  -\frac{1}{n}  \vecone_n^{'} \otimes \delta_1 \vecone_p \\
-\frac{1}{n}  \vecone_n \otimes \delta_1 \vecone^{'}_p  & \delta_1 \identity_n
\end{bmatrix}.
\label{t29c-1}
\end{align}
From \eqref{Zd} and \eqref{t19c-1}, by calculation we get
\begin{align}
\boldsymbol{B}^{(1)} &=  - \boldsymbol{D}^{(1)} \bar{\rho} = 
\begin{bmatrix}
-n \bar{\rho} \boldsymbol{K}_2^{-1} & - \vecone_n^{'} \otimes   \bar{\rho}  \boldsymbol{K}_2^{-1} \vecone_p\\
-\vecone_n \otimes  \bar{\rho} \vecone^{'}_p \boldsymbol{K}_2^{-1} & - \bar{\rho} \delta_2^{-1} \identity_n
\end{bmatrix}, 
\label{t31c-1}
\end{align}
where $\delta_2 = 1/\vecone^{'}_p  \boldsymbol{K}_2^{-1} \vecone_p$. Note that since $\boldsymbol{K}_2$ is a positive definite matrix, $\boldsymbol{K}_2^{-1}$ is a positive definite matrix and thus $\delta_2>0$.\par

\noindent Let $\boldsymbol{K}_3 =  \bar{\rho}^2 \boldsymbol{V}_1 + \boldsymbol{K}_2 $. Note that $\boldsymbol{K}_3$ is a symmetric matrix. Since $\boldsymbol{V}_1$, $\boldsymbol{K}_1$ and $\boldsymbol{K}_2$ are positive definite matrices, using the result that $ \identity_p + \boldsymbol{\mathit{\zeta}_3} \boldsymbol{\mathit{\zeta}_4}$ is invertible if and only if $\identity_p + \boldsymbol{\mathit{\zeta}_4} \boldsymbol{\mathit{\zeta}_3}$ is invertible, where $\boldsymbol{\mathit{\zeta}_3}$ and $\boldsymbol{\mathit{\zeta}_4}$ are matrices such that $\boldsymbol{\mathit{\zeta}_3} \boldsymbol{\mathit{\zeta}_4}$ is a $p \times p$ matrix, we get that $\boldsymbol{K}_3$ is a invertible matrix. Using the result $\left( \identity_p + \boldsymbol{\mathit{\zeta}_3} \boldsymbol{\mathit{\zeta}_4} \right)^{-1} = \identity_p - \boldsymbol{\mathit{\zeta}_3}  \left( \identity_p + \boldsymbol{\mathit{\zeta}_4} \boldsymbol{\mathit{\zeta}_3} \right)^{-1}   \boldsymbol{\mathit{\zeta}_4}$, where $\boldsymbol{\mathit{\zeta}_3}$ and $\boldsymbol{\mathit{\zeta}_4}$ are matrices such that $\boldsymbol{\mathit{\zeta}_3} \boldsymbol{\mathit{\zeta}_4}$ is a $p \times p$ matrix, we get that $\boldsymbol{K}_2^{-1} - \bar{\rho}^2 \boldsymbol{K}_2^{-1} \boldsymbol{K}_1 \boldsymbol{K}_2^{-1} = \boldsymbol{K}_3^{-1}$. Thus from \eqref{t29c-1} and \eqref{t31c-1}, by calculation we get
\begin{align}
\boldsymbol{Z}^{(1)} &=\boldsymbol{D}^{(1)} - \boldsymbol{B}^{(1)} \left(\boldsymbol{A}^{(1)} \right)^{-} \boldsymbol{B}^{(1)} \nonumber\\
&=
\begin{bmatrix}
n \boldsymbol{K}_3^{-1} & \vecone_n^{'} \otimes \boldsymbol{K}_3^{-1} \vecone_p\\
\vecone_n \otimes \vecone_p^{'} \boldsymbol{K}_3^{-1} & \delta_4 \hatmat_n + \frac{1}{n} \left( \vecone_p^{'} \boldsymbol{K}_3^{-1}  \vecone_p \right) \matone_{n \times n}
\end{bmatrix},
\label{t37c-1}
\end{align}
where $\delta_4 = \delta_2^{-1} - \bar{\rho}^2 \delta_1 \delta_2^{-2} $. Using the expression of $\boldsymbol{K}_1^{-1}$, $\delta_1$ and $\delta_2$, we get
\begin{align}
\delta_1 \vecone_{p}^{'} \boldsymbol{V}_1^{-1} \vecone_{p} + \bar{\rho}^2 \delta_1 \delta_2^{-1} = 1.
\label{markov-t5c-1}
\end{align}
From \eqref{markov-t5c-1} and using the expression of $\delta_4$, we get
\begin{align}
\delta_4 = \delta_1 \delta_2^{-1} \vecone_{p}^{'} \boldsymbol{V}_1^{-1} \vecone_{p}.
\label{markov-t6c-1}
\end{align}
Since $\delta_1, \delta_2>0$ and $\boldsymbol{V}_1$ is a positive definite matrix, from \eqref{markov-t6c-1}, we get that $\delta_4>0$.\par

\noindent Let $\boldsymbol{A}^{(3)}= n \boldsymbol{K}_3^{-1}$, $\boldsymbol{B}^{(3)} = \left(\boldsymbol{C}^{(3)} \right)^{'} = \vecone_n^{'} \otimes \boldsymbol{K}_3^{-1} \vecone_p$ and $\boldsymbol{D}^{(3)} = \delta_4 \hatmat_n + \frac{1}{n} \left( \vecone_p^{'} \boldsymbol{K}_3^{-1}  \vecone_p \right) \matone_{n \times n}$. So
\begin{align}
\left( \boldsymbol{Z}^{(1)} \right)^{-} &=
\begin{bmatrix}
\left(\boldsymbol{A}^{(3)} \right)^{-} + \left(\boldsymbol{A}^{(3)} \right)^{-} \boldsymbol{B}^{(3)} \left( \boldsymbol{Z}^{(3)} \right)^{-} \boldsymbol{C}^{(3)} \left(\boldsymbol{A}^{(3)} \right)^{-} & - \left(\boldsymbol{A}^{(3)} \right)^{-} \boldsymbol{B}^{(3)}\left( \boldsymbol{Z}^{(3)} \right)^{-}\\
- \left( \boldsymbol{Z}^{(3)} \right)^{-} \boldsymbol{C}^{(3)} \left(\boldsymbol{A}^{(3)} \right)^{-}   & \left( \boldsymbol{Z}^{(3)} \right)^{-}
\end{bmatrix},
\label{t38c-1}
\end{align}
where $\boldsymbol{Z}^{(3)} = \boldsymbol{D}^{(3)} - \boldsymbol{C}^{(3)} \left(\boldsymbol{A}^{(3)} \right)^{-} \boldsymbol{B}^{(3)}$. Using the expression of the matrices $\boldsymbol{A}^{(3)}$, $\boldsymbol{B}^{(3)}$, $\boldsymbol{C}^{(3)}$ and $\boldsymbol{D}^{(3)}$, by calculation we get
\begin{align}
\begin{split}
\boldsymbol{Z}^{(3)} &= \delta_4 \hatmat_n,\\
\left(  \boldsymbol{A}^{(3)}  \right)^{-} \boldsymbol{B}^{(3)} \left( \boldsymbol{Z}^{(3)} \right)^{-} &= \frac{1}{n}  \vecone_n^{'} \otimes \delta_4^{-1} \vecone_p,\\
\left( \boldsymbol{Z}^{(3)} \right)^{-} \boldsymbol{C}^{(3)} \left(  \boldsymbol{A}^{(3)}  \right)^{-} &= \frac{1}{n}   \vecone_n \otimes \delta_4^{-1} \vecone_p^{'}, \text{ and} \\
\left(\boldsymbol{A}^{(3)} \right)^{-} + \left(\boldsymbol{A}^{(3)} \right)^{-} \boldsymbol{B}^{(3)} \left( \boldsymbol{Z}^{(3)} \right)^{-} \boldsymbol{C}^{(3)} \left(\boldsymbol{A}^{(3)} \right)^{-} &= \frac{1}{n} \boldsymbol{K}_3 + \frac{1}{n} \delta_4^{-1} \matone_{p \times p}.
\end{split}
\label{t40c-1}
\end{align}

\noindent From \eqref{t38c-1} and \eqref{t40c-1}, we get
\begin{align}
\left( \boldsymbol{Z}^{(1)} \right)^{-} &=
\begin{bmatrix}
\frac{1}{n} \boldsymbol{K}_3 + \frac{1}{n} \delta_4^{-1} \matone_{p \times p} & -\frac{1}{n}  \vecone_n^{'} \otimes \delta_4^{-1} \vecone_p \\
-\frac{1}{n}   \vecone_n \otimes \delta_4^{-1} \vecone_p^{'}  & \delta_4^{-1} \identity_n
\end{bmatrix}.
\label{t43c-1}
\end{align}
Using the expression of matrices $\boldsymbol{K}_1$, $\boldsymbol{K}_2$ and $\boldsymbol{K}_3$, we get 
\begin{align}
\boldsymbol{K}_1 \boldsymbol{K}_2^{-1} = \boldsymbol{V}_1 \boldsymbol{K}_3^{-1}.
\label{markov-t2c-1}
\end{align}
So from \eqref{t29c-1}, \eqref{t31c-1} and \eqref{t43c-1}, by calculation we get
\begin{align}
\left( \boldsymbol{A}^{(1)} \right)^{-}  \boldsymbol{B}^{(1)}  \left( \boldsymbol{Z}^{(1)} \right)^{-} = \left( \boldsymbol{Z}^{(1)} \right)^{-}  \boldsymbol{B}^{(1)}  \left( \boldsymbol{A}^{(1)} \right)^{-} =
\begin{bmatrix}
-\frac{1}{n} \bar{\rho} \boldsymbol{V}_1 & \zero_{p \times n}\\
\zero_{n \times p} & - \bar{\rho} \delta_4^{-1} \delta_1 \delta_2^{-1} \hatmat_n
\end{bmatrix}
\label{markov-t3c-1}
\end{align}
and
\begin{multline}
\left( \boldsymbol{A}^{(1)} \right)^{-}+\left( \boldsymbol{A}^{(1)} \right)^{-}  \boldsymbol{B}^{(1)}  \left( \boldsymbol{Z}^{(1)} \right)^{-} \boldsymbol{C}^{(1)}  \left( \boldsymbol{A}^{(1)} \right)^{-}  \\=
\begin{bmatrix}
\frac{1}{n} \boldsymbol{K}_1 + \frac{1}{n} \delta_1 \matone_{p \times p} + \frac{1}{n} \bar{\rho}^2 \boldsymbol{V}_1 \boldsymbol{K}_2^{-1} \boldsymbol{K}_1 &  -\frac{1}{n}  \vecone_n^{'} \otimes \delta_1 \vecone_p \\
-\frac{1}{n}  \vecone_n \otimes \delta_1 \vecone^{'}_p  & \delta_1 \identity_n  + \bar{\rho}^2 \delta_4^{-1} \delta_1^2 \delta_2^{-2} \hatmat_n
\end{bmatrix}.
\label{t44c-1}
\end{multline}

\noindent From \eqref{t43c-1}, \eqref{markov-t3c-1} and \eqref{t44c-1}, \eqref{t22c-1} becomes
\begin{align}
\left( \boldsymbol{Z}^{'}_{1} \boldsymbol{\mathit{\Sigma}}^{-1} \boldsymbol{Z}_{1} \right)^{-}&= 
\begin{bmatrix}
\boldsymbol{L}_1 & \boldsymbol{L}_2\\
\boldsymbol{L}_3 & \boldsymbol{L}_4
\end{bmatrix},
\label{t48c-1}
\end{align}
where
\begin{align}
\begin{split}
\boldsymbol{L}_1 &=
\begin{bmatrix}
\frac{1}{n} \boldsymbol{K}_1 + \frac{1}{n} \delta_1 \matone_{p \times p} + \frac{1}{n} \bar{\rho}^2 \boldsymbol{V}_1 \boldsymbol{K}_2^{-1} \boldsymbol{K}_1 &  -\frac{1}{n}  \vecone_n^{'} \otimes \delta_1 \vecone_p \\
-\frac{1}{n}  \vecone_n \otimes \delta_1 \vecone^{'}_p  & \delta_1 \identity_n  + \bar{\rho}^2 \delta_4^{-1} \delta_1^2 \delta_2^{-2} \hatmat_n
\end{bmatrix},\\
\boldsymbol{L}_2 &= \boldsymbol{L}_3 =
\begin{bmatrix}
\frac{1}{n} \bar{\rho} \boldsymbol{V}_1 & \zero_{p \times n}\\
\zero_{n \times p} & \bar{\rho} \delta_4^{-1} \delta_1 \delta_2^{-1} \hatmat_n
\end{bmatrix}, \text{ and}\\
\boldsymbol{L}_4 &= 
\begin{bmatrix}
\frac{1}{n} \boldsymbol{K}_3 + \frac{1}{n} \delta_4^{-1} \matone_{p \times p} & -\frac{1}{n}   \vecone_n^{'} \otimes \delta_4^{-1} \vecone_p  \\
-\frac{1}{n}   \vecone_n \otimes \delta_4^{-1} \vecone_p^{'}  & \delta_4^{-1} \identity_n
\end{bmatrix}.
\end{split}
\label{t49c-1}
\end{align}
Now, we find $\boldsymbol{\mathit{\Sigma}}^{-1} \boldsymbol{Z_1}$. From \eqref{Zd}, \eqref{t18c-1} and \eqref{t19c-1}, by calculation we get
\begin{align}
\boldsymbol{\mathit{\Sigma}}^{-1} \boldsymbol{Z_1} &= 
\begin{bmatrix}
\boldsymbol{\mathit{\Sigma}}^{11} \boldsymbol{X_1} & \boldsymbol{\mathit{\Sigma}}^{12} \boldsymbol{X_1}\\
\boldsymbol{\mathit{\Sigma}}^{12} \boldsymbol{X_1} & \boldsymbol{\mathit{\Sigma}}^{22} \boldsymbol{X_1}
\end{bmatrix},
\label{t50c-1}
\end{align}
where
\begin{align}
\begin{split}
\boldsymbol{\mathit{\Sigma}}^{11} \boldsymbol{X_1} &= 
\begin{bmatrix}
\vecone_n \otimes \boldsymbol{K}_1^{-1} & \identity_n \otimes \boldsymbol{K}_1^{-1} \vecone_p
\end{bmatrix}, \text{ and}\\
\boldsymbol{\mathit{\Sigma}}^{12} \boldsymbol{X_1} &= -\bar{\rho} \boldsymbol{\mathit{\Sigma}}^{22} \boldsymbol{X_1} = \begin{bmatrix}
-\vecone_n \otimes \bar{\rho} \boldsymbol{K}_2^{-1}  & -\identity_n \otimes \bar{\rho}  \boldsymbol{K}_2^{-1} \vecone_p
\end{bmatrix}. 
\end{split}
\label{t51c-1}
\end{align}

\noindent Now, we find $\boldsymbol{\mathit{\Sigma}}^{-1} \boldsymbol{Z_1} \left( \boldsymbol{Z}^{'}_{1} \boldsymbol{\mathit{\Sigma}}^{-1} \boldsymbol{Z}_{1} \right)^{-} \boldsymbol{Z}^{'}_{1} \boldsymbol{\mathit{\Sigma}}^{-1}$. Since $\boldsymbol{V}_1$, $\boldsymbol{K}_1$, $\boldsymbol{K}_2$ and $\boldsymbol{K}_3$ are symmetric matrices, from \eqref{markov-t2c-1}, we get that $\boldsymbol{K}_2^{-1} \boldsymbol{K}_3 = \boldsymbol{K}_1^{-1} \boldsymbol{V}_1$ and $\boldsymbol{K}_1^{-1} \boldsymbol{V}_1 \boldsymbol{K}_2^{-1} \boldsymbol{K}_1 = \boldsymbol{K}_2^{-1} \boldsymbol{V}_1$. So, using the expression of matrix $\boldsymbol{K}_3$, we get that $\boldsymbol{K}_2^{-1} \boldsymbol{K}_3 = \identity_p + \bar{\rho}^2 \boldsymbol{K}_2^{-1} \boldsymbol{V}_1$ and $\boldsymbol{K}_2^{-1} \boldsymbol{K}_1 + \bar{\rho}^2 \boldsymbol{K}_2^{-1} \boldsymbol{V}_1 \boldsymbol{K}_2^{-1} \boldsymbol{K}_1  = \boldsymbol{K}_2^{-1} \boldsymbol{V}_1$. Using the expression of $\delta_4$, by solving we get $\delta_1 + \bar{\rho}^2 \delta_4^{-1} \delta_1^2 \delta_2^{-2} =
\delta_4^{-1} \delta_1 \delta_2^{-1}
$. Thus from \eqref{t48c-1}, \eqref{t49c-1}, \eqref{t50c-1} and \eqref{t51c-1}, by calculation we get
\begin{align}
\boldsymbol{\mathit{\Sigma}}^{-1} \boldsymbol{Z_1} \left( \boldsymbol{Z}^{'}_{1} \boldsymbol{\mathit{\Sigma}}^{-1} \boldsymbol{Z}_{1} \right)^{-} \boldsymbol{Z}^{'}_{1} \boldsymbol{\mathit{\Sigma}}^{-1} &=
\begin{bmatrix}
\boldsymbol{A}^{(4)} & \boldsymbol{B}^{(4)}\\
\boldsymbol{B}^{(4)} & \boldsymbol{D}^{(4)}
\end{bmatrix},
\label{t79c-1}
\end{align}
where
\begin{multline*}
\boldsymbol{A}^{(4)} = 
\frac{1}{n} \matone_{n \times n} \otimes \boldsymbol{K}_1^{-1} + \hatmat_n \otimes \delta_4^{-1} \delta_1 \delta_2^{-1} \boldsymbol{K}_1^{-1} \matone_{p \times p} \boldsymbol{K}_1^{-1}  - \hatmat_n \otimes \bar{\rho}^2 \delta_4^{-1}  \delta_1 \delta_2^{-1} \boldsymbol{K}_2^{-1}  \matone_{p \times p} \boldsymbol{K}_1^{-1} \\+ \hatmat_n \otimes \bar{\rho}^2 \delta_4^{-1}  \boldsymbol{K}_2^{-1} \matone_{p \times p} \boldsymbol{K}_2^{-1}  - \hatmat_n \otimes \bar{\rho}^2 \delta_4^{-1}  \delta_1 \delta_2^{-1} \boldsymbol{K}_1^{-1} \matone_{p \times p}  \boldsymbol{K}_2^{-1} , \text{ and}
 \end{multline*}
 \begin{align*}
\boldsymbol{B}^{(4)} &= -\bar{\rho} \boldsymbol{D}^{(4)}\\
&=
-\frac{1}{n} \matone_{n \times n} \otimes \bar{\rho}  \boldsymbol{K}_2^{-1}  + \hatmat_n \otimes \bar{\rho}^3 \delta_4^{-1} \delta_1 \delta_2^{-1} \boldsymbol{K}_2^{-1} \matone_{p \times p}  \boldsymbol{K}_2^{-1}  -\hatmat_n \otimes \bar{\rho} \delta_4^{-1}   \boldsymbol{K}_2^{-1} \matone_{p \times p} \boldsymbol{K}_2^{-1}.
\end{align*}
Thus from \eqref{markov-t5c-1}, \eqref{markov-t6c-1} and using the expression of $\boldsymbol{A}^*$ from \autoref{lemma5-c4}, by calculation we can prove \eqref{p31c-1}.
\end{proof}

\begin{remarks}
\label{remark1-c4}
The matrices $\boldsymbol{\mathit{\Omega}}_1$, $\boldsymbol{\mathit{\Omega}}_2$, and $\boldsymbol{\mathit{\Omega}}_4$ have row sums and column sums as zero.
\end{remarks}
\begin{proof}
Using the expression of the matrices $\boldsymbol{\mathit{\Omega}}_1$, $\boldsymbol{\mathit{\Omega}}_2$, and $\boldsymbol{\mathit{\Omega}}_4$, as given in \autoref{lemma3-c4-1}, we get that $\vecone_p^{'} \boldsymbol{\mathit{\Omega}}_1 = \vecone_p^{'} \boldsymbol{\mathit{\Omega}}_2 = \vecone_p^{'} \boldsymbol{\mathit{\Omega}}_4 = \zero_{1 \times p}$ and $ \boldsymbol{\mathit{\Omega}}_1 \vecone_p=  \boldsymbol{\mathit{\Omega}}_2 \vecone_p=   \boldsymbol{\mathit{\Omega}}_4 \vecone_p= \zero_{p \times 1}$.
\end{proof}

\begin{lemmas}
\label{lemma5-c4-1}
Under the generalized Markov-type structure, the information matrix for the direct effects in the absence of the period effects can be expressed as
\begin{align}
\begin{split}
\boldsymbol{\tilde{C}}_{d(s2)} &= \left( \identity_2 \otimes \boldsymbol{T}_d^{'} \right) \boldsymbol{\mathit{\Sigma}}^{-1/2} pr^{\perp} \left( \boldsymbol{\mathit{\Sigma}}^{-1/2} 
\begin{bmatrix}
\identity_2 \otimes \boldsymbol{U} & \identity_2 \otimes \boldsymbol{F}_d \hatmat_t
\end{bmatrix}  
\right) \boldsymbol{\mathit{\Sigma}}^{-1/2} \left( \identity_2 \otimes \boldsymbol{T}_d\right)\\
&= \boldsymbol{\tilde{C}}_{d(s2)(11)} - \boldsymbol{\tilde{C}}_{d(s2)(12)} \boldsymbol{\tilde{C}}^{-}_{d(s2)(22)} \boldsymbol{\tilde{C}}_{d(s2)(21)},
\end{split}
\label{p25c}
\end{align}
where $\boldsymbol{\tilde{C}}_{d(s2)(11)} = \left( \identity_2 \otimes \boldsymbol{T}_d^{'} \right) \boldsymbol{\mathit{\Sigma}}^{-1/2} pr^{\perp} \left( \boldsymbol{\mathit{\Sigma}}^{-1/2} 
\begin{bmatrix}
\identity_2 \otimes \boldsymbol{U} 
\end{bmatrix}  
\right) \boldsymbol{\mathit{\Sigma}}^{-1/2} \left( \identity_2 \otimes \boldsymbol{T}_d\right)$,\\
$\boldsymbol{\tilde{C}}_{d(s2)(12)} = \boldsymbol{\tilde{C}}^{'}_{d(s2)(21)} = \left( \identity_2 \otimes \boldsymbol{T}_d^{'} \right) \boldsymbol{\mathit{\Sigma}}^{-1/2} pr^{\perp} \left( \boldsymbol{\mathit{\Sigma}}^{-1/2} 
\begin{bmatrix}
\identity_2 \otimes \boldsymbol{U} 
\end{bmatrix}  
\right) \boldsymbol{\mathit{\Sigma}}^{-1/2} \left( \identity_2 \otimes \boldsymbol{F}_d \hatmat_t \right)$, and\\
$\boldsymbol{\tilde{C}}_{d(s2)(22)} = \left( \identity_2 \otimes \hatmat_t \boldsymbol{F}_d^{'} \right) \boldsymbol{\mathit{\Sigma}}^{-1/2} pr^{\perp} \left( \boldsymbol{\mathit{\Sigma}}^{-1/2} 
\begin{bmatrix}
\identity_2 \otimes \boldsymbol{U} 
\end{bmatrix}  
\right) \boldsymbol{\mathit{\Sigma}}^{-1/2} \left( \identity_2 \otimes \boldsymbol{F}_d \hatmat_t \right)$.
\end{lemmas}
\begin{proof}
Following the same approach as to obtain the information matrix for the direct effects in the presence of the period effects as discussed in \autoref{lemma5-c4}, we can prove \eqref{p25c}.
\end{proof}

\begin{lemmas}
\label{lemma7-c4-1}
The matrix $\boldsymbol{\mathit{\Sigma}}^{-1/2} pr^{\perp} \left( \boldsymbol{\mathit{\Sigma}}^{-1/2} 
\begin{bmatrix}
\identity_2 \otimes \boldsymbol{U} 
\end{bmatrix}  
\right) \boldsymbol{\mathit{\Sigma}}^{-1/2}$ can be expressed as
\begin{align}
\boldsymbol{\mathit{\Sigma}}^{-1/2} pr^{\perp} \left( \boldsymbol{\mathit{\Sigma}}^{-1/2} 
\begin{bmatrix}
\identity_2 \otimes \boldsymbol{U} 
\end{bmatrix}  
\right) \boldsymbol{\mathit{\Sigma}}^{-1/2} &= 
\begin{bmatrix}
\identity_n \otimes \boldsymbol{\mathit{\Omega}}_1 & -\identity_n \otimes \boldsymbol{\mathit{\Omega}}_2\\
-\identity_n \otimes \boldsymbol{\mathit{\Omega}}_2 & \identity_n \otimes \boldsymbol{\mathit{\Omega}}_4
\end{bmatrix},
\label{p8c-1}
\end{align}
where $\boldsymbol{\mathit{\Omega}}_1$, $\boldsymbol{\mathit{\Omega}}_2$, and $\boldsymbol{\mathit{\Omega}}_4$ are as given in \autoref{lemma3-c4-1}.
\end{lemmas}

\begin{proof}
From \eqref{Zd}, \eqref{t18c-1} and \eqref{t19c-1}, by calculation we get
\begin{align}
\left(\identity_2 \otimes \boldsymbol{U}^{'} \right) \boldsymbol{\mathit{\Sigma}}^{-1} \left( \identity_2 \otimes \boldsymbol{U} \right) &=
\begin{bmatrix}
\boldsymbol{A}^{(5)} & \boldsymbol{B}^{(5)}\\
\boldsymbol{B}^{(5)} & \boldsymbol{D}^{(5)}
\end{bmatrix},
\label{p13c-1}
\end{align}
where $\boldsymbol{A}^{(5)} = \delta_1^{-1} \identity_n$, $\boldsymbol{B}^{(5)}  = -\bar{\rho} \delta_2^{-1}  \identity_n$, and $\boldsymbol{D}^{(5)}  = \delta_2^{-1} \identity_n$. Here, $\delta_1 = 1/ \vecone_{p}^{'} \boldsymbol{K}_1^{-1} \vecone_{p}$, $\delta_2 = 1/ \vecone_{p}^{'} \boldsymbol{K}_2^{-1} \vecone_{p}$, $\boldsymbol{K}_1^{-1} = \boldsymbol{V}_1^{-1} + \bar{\rho}^2 \boldsymbol{K}_2^{-1}$ and $\boldsymbol{K}_2 = \sigma_{12} \boldsymbol{V}_R$, $\bar{\rho} = \rho \sqrt{\frac{\sigma_{22}}{\sigma_{11}}}$ and $\sigma_{12} = \sigma_{22} \left( 1 - \rho^2 \right)$. So
\begin{multline}
\left[ \left(\identity_2 \otimes \boldsymbol{U}^{'} \right) \boldsymbol{\mathit{\Sigma}}^{-1} \left( \identity_2 \otimes \boldsymbol{U} \right) \right]^{-} =\\
\begin{bmatrix}
\left(\boldsymbol{A}^{(5)} \right)^{-} + \left(\boldsymbol{A}^{(5)} \right)^{-} \boldsymbol{B}^{(5)} \left( \boldsymbol{Z}^{(5)} \right)^{-} \boldsymbol{B}^{(5)} \left(\boldsymbol{A}^{(5)} \right)^{-} & - \left(\boldsymbol{A}^{(5)} \right)^{-} \boldsymbol{B}^{(5)}\left( \boldsymbol{Z}^{(5)} \right)^{-}\\
- \left( \boldsymbol{Z}^{(5)} \right)^{-} \boldsymbol{B}^{(5)} \left(\boldsymbol{A}^{(5)} \right)^{-}   & \left( \boldsymbol{Z}^{(5)} \right)^{-}
\end{bmatrix},
\label{p16c-1}
\end{multline}
where $\boldsymbol{Z}^{(5)} = \boldsymbol{D}^{(5)} - \boldsymbol{B}^{(5)} \left(\boldsymbol{A}^{(5)} \right)^{-} \boldsymbol{B}^{(5)}$. Using the expressions of $\boldsymbol{A}^{(5)}$, $\boldsymbol{B}^{(5)}$ and $\boldsymbol{D}^{(5)}$, by solving we get
{\small{
\begin{align}
\begin{split}
\left(\boldsymbol{Z}^{(5)} \right)^{-} &= \delta_4^{-1} \identity_n,\\
\left(\boldsymbol{A}^{(5)} \right)^{-} \boldsymbol{B}^{(5)}\left( \boldsymbol{Z}^{(5)} \right)^{-} &=  \left( \boldsymbol{Z}^{(5)} \right)^{-} \boldsymbol{B}^{(5)} \left(\boldsymbol{A}^{(5)} \right)^{-} = - \bar{\rho} \delta_4^{-1} \delta_1 \delta_2^{-1} \identity_n, \text{ and}\\
\left(\boldsymbol{A}^{(5)} \right)^{-} + \left(\boldsymbol{A}^{(5)} \right)^{-} \boldsymbol{B}^{(5)} \left( \boldsymbol{Z}^{(5)} \right)^{-} \boldsymbol{B}^{(5)} \left(\boldsymbol{A}^{(5)} \right)^{-} &= \left( \delta_1 + \bar{\rho}^2 \delta_4^{-1}  \delta_1^2 \delta_2^{-2} \right) \identity_n,
\end{split}
\label{p17c-1}
\end{align}
}}
where $\delta_4 = \delta_2^{-1} - \bar{\rho}^2 \delta_1 \delta_2^{-2} $. Using the expression $\delta_4$, we get that $
\delta_1 + \bar{\rho}^2 \delta_4^{-1}  \delta_1^2 \delta_2^{-2} =$\\$
\delta_4^{-1} \delta_1 \delta_2^{-1}
$. So from \eqref{p16c-1} and \eqref{p17c-1}, we get
\begin{align}
\left[ \left(\identity_2 \otimes \boldsymbol{U}^{'} \right) \boldsymbol{\mathit{\Sigma}}^{-1} \left( \identity_2 \otimes \boldsymbol{U} \right) \right]^{-} &= \delta_4^{-1}
\begin{bmatrix}
\delta_1 \delta_2^{-1} \identity_n & \bar{\rho} \delta_1 \delta_2^{-1} \identity_n\\
\bar{\rho} \delta_1 \delta_2^{-1} \identity_n & \identity_n
\end{bmatrix}.
\label{p19c-1}
\end{align}
From \eqref{Zd}, \eqref{t18c-1} and \eqref{t19c-1}, by calculation we get
\begin{align}
\left(\identity_2 \otimes \boldsymbol{U}^{'} \right) \boldsymbol{\mathit{\Sigma}}^{-1} &=
\begin{bmatrix}
\identity_n \otimes \vecone_{p}^{'} \boldsymbol{K}_1^{-1} & -\identity_n \otimes \bar{\rho} \vecone_{p}^{'} \boldsymbol{K}_2^{-1}\\
-\identity_n \otimes \bar{\rho} \vecone_{p}^{'} \boldsymbol{K}_2^{-1} & \identity_n \otimes \vecone_p^{'} \boldsymbol{K}_2^{-1} 
\end{bmatrix}.
\label{p20c-1}
\end{align}
From \eqref{p19c-1} and \eqref{p20c-1}, by calculation we get
\begin{align}
\boldsymbol{\mathit{\Sigma}}^{-1} \left(\identity_2 \otimes \boldsymbol{U} \right) \left[ \left(\identity_2 \otimes \boldsymbol{U}^{'} \right) \boldsymbol{\mathit{\Sigma}}^{-1} \left( \identity_2 \otimes \boldsymbol{U} \right) \right]^{-} \left(\identity_2 \otimes \boldsymbol{U}^{'} \right) \boldsymbol{\mathit{\Sigma}}^{-1} &=
\begin{bmatrix}
\identity_n \otimes \boldsymbol{A}^{(6)} & -\identity_n \otimes \boldsymbol{B}^{(6)}\\
-\identity_n \otimes \boldsymbol{B}^{(6)} & \identity_n \otimes \boldsymbol{D}^{(6)}
\end{bmatrix},
\label{p21c-1}
\end{align}
where
\begin{multline*}
\boldsymbol{A}^{(6)} = \delta_4^{-1} \delta_1 \delta_2^{-1}  \boldsymbol{K}_1^{-1} \matone_{p \times p}  \boldsymbol{K}_1^{-1} - \bar{\rho}^2 \delta_4^{-1}  \delta_1 \delta_2^{-1} \boldsymbol{K}_1^{-1} \matone_{p \times p}  \boldsymbol{K}_2^{-1}  + \bar{\rho}^2 \delta_4^{-1}   \boldsymbol{K}_2^{-1} \matone_{p \times p}  \boldsymbol{K}_2^{-1} \\ - \bar{\rho}^2 \delta_4^{-1}   \delta_1 \delta_2^{-1} \boldsymbol{K}_2^{-1} \matone_{p \times p}  \boldsymbol{K}_1^{-1}, \text{ and}
\end{multline*}
\begin{align*}
\boldsymbol{B}^{(6)} = \bar{\rho} \boldsymbol{D}^{(6)}  = - \bar{\rho}^3 \delta_4^{-1}  \delta_1 \delta_2^{-1} \boldsymbol{K}_2^{-1} \matone_{p \times p} \boldsymbol{K}_2^{-1} + \bar{\rho} \delta_4^{-1}  \boldsymbol{K}_2^{-1} \matone_{p \times p} \boldsymbol{K}_2^{-1}. 
\end{align*}
Thus from \eqref{t18c-1}, \eqref{t19c-1}, \eqref{markov-t5c-1}, \eqref{markov-t6c-1} and \eqref{p21c-1}, by calculation we can prove \eqref{p8c-1}.
\end{proof}

\begin{lemmas}
\label{thm3-c4}
For any design $d \in \Omega_{t,n,p}$, $\boldsymbol{{C}}_{d(s2)} \leq \boldsymbol{\tilde{C}}_{d(s2)}$. If $d$ is a design uniform on periods, then $\boldsymbol{{C}}_{d(s2)}=\boldsymbol{\tilde{C}}_{d(s2)}$.
\end{lemmas}
\begin{proof}
Since $pr^{\perp} \left( \boldsymbol{\mathit{\Sigma}}^{-1/2} 
\begin{bmatrix}
\boldsymbol{Z_1} & \identity_2 \otimes \boldsymbol{F}_d \hatmat_t
\end{bmatrix}  
\right) = pr^{\perp} \left( \boldsymbol{\mathit{\Sigma}}^{-1/2} 
\begin{bmatrix}
\identity_2 \otimes  \boldsymbol{P} & \identity_2 \otimes \boldsymbol{U} & \identity_2 \otimes \boldsymbol{F}_d \hatmat_t
\end{bmatrix}  
\right) $\\$ \leq pr^{\perp} \left( \boldsymbol{\mathit{\Sigma}}^{-1/2} 
\begin{bmatrix}
 \identity_2 \otimes \boldsymbol{U} & \identity_2 \otimes \boldsymbol{F}_d \hatmat_t
\end{bmatrix}  
\right)$, using the expression of the matrices $\boldsymbol{C}_{d(s2)}$ and $\boldsymbol{\tilde{C}}_{d(s2)}$ as given in \autoref{lemma5-c4} and \autoref{lemma5-c4-1}, respectively, we get that $\boldsymbol{C}_{d(s2)} \leq \boldsymbol{\tilde{C}}_{d(s2)}$.\par

\noindent Since $\boldsymbol{F}_d = \left( \identity_n \otimes \boldsymbol{\psi} \right) \boldsymbol{T}_d$, where $\boldsymbol{\psi} =  
\begin{bmatrix}
\zero^{'}_{p-1 \times 1} & 0\\
\identity_{p-1} & \zero_{p-1 \times 1}
\end{bmatrix}$, using the expression of the matrices $\boldsymbol{\tilde{C}}_{d(s2)(11)}$, $\boldsymbol{\tilde{C}}_{d(s2)(12)}$, $\boldsymbol{\tilde{C}}_{d(s2)(21)}$ and $\boldsymbol{\tilde{C}}_{d(s2)(22)}$ as given in \autoref{lemma5-c4-1}, and the expression of $\boldsymbol{\mathit{\Sigma}}^{-1/2} pr^{\perp} \left( \boldsymbol{\mathit{\Sigma}}^{-1/2} 
\begin{bmatrix}
\identity_2 \otimes \boldsymbol{U} 
\end{bmatrix}  
\right) \boldsymbol{\mathit{\Sigma}}^{-1/2}$ as given in \autoref{lemma7-c4-1}, we get
\begin{align}
\begin{split}
\boldsymbol{\tilde{C}}_{d(s2)(11)} &=
\begin{bmatrix}
\boldsymbol{T}_d^{'} \left( \identity_n \otimes \boldsymbol{\mathit{\Omega}}_1 \right) \boldsymbol{T}_d & -\boldsymbol{T}_d^{'} \left( \identity_n \otimes \boldsymbol{\mathit{\Omega}}_2 \right) \boldsymbol{T}_d\\
-\boldsymbol{T}_d^{'} \left( \identity_n \otimes \boldsymbol{\mathit{\Omega}}_2 \right) \boldsymbol{T}_d & \boldsymbol{T}_d^{'} \left( \identity_n \otimes \boldsymbol{\mathit{\Omega}}_4 \right) \boldsymbol{T}_d
\end{bmatrix},\\
\boldsymbol{\tilde{C}}_{d(s2)(12)} &= \boldsymbol{\tilde{C}}^{'}_{d(s2)(21)} =
\begin{bmatrix}
\boldsymbol{T}_d^{'} \left( \identity_n \otimes \boldsymbol{\mathit{\Omega}}_1 \boldsymbol{\psi} \right) \boldsymbol{T}_d \hatmat_t & -\boldsymbol{T}_d^{'} \left( \identity_n \otimes \boldsymbol{\mathit{\Omega}}_2 \boldsymbol{\psi} \right) \boldsymbol{T}_d \hatmat_t\\
-\boldsymbol{T}_d^{'} \left( \identity_n \otimes \boldsymbol{\mathit{\Omega}}_2 \boldsymbol{\psi} \right) \boldsymbol{T}_d \hatmat_t & \boldsymbol{T}_d^{'} \left( \identity_n \otimes \boldsymbol{\mathit{\Omega}}_4 \boldsymbol{\psi} \right) \boldsymbol{T}_d \hatmat_t
\end{bmatrix}, \text{ and}\\
\boldsymbol{\tilde{C}}_{d(s2)(22)} &=
\begin{bmatrix}
\hatmat_t \boldsymbol{T}_d^{'} \left( \identity_n \otimes \boldsymbol{\psi}^{'} \boldsymbol{\mathit{\Omega}}_1 \boldsymbol{\psi} \right) \boldsymbol{T}_d \hatmat_t & -\hatmat_t \boldsymbol{T}_d^{'} \left( \identity_n \otimes \boldsymbol{\psi}^{'} \boldsymbol{\mathit{\Omega}}_2 \boldsymbol{\psi} \right) \boldsymbol{T}_d \hatmat_t\\
-\hatmat_t \boldsymbol{T}_d^{'} \left( \identity_n \otimes \boldsymbol{\psi}^{'} \boldsymbol{\mathit{\Omega}}_2 \boldsymbol{\psi} \right) \boldsymbol{T}_d \hatmat_t & \hatmat_t \boldsymbol{T}_d^{'} \left( \identity_n \otimes \boldsymbol{\psi}^{'} \boldsymbol{\mathit{\Omega}}_4 \boldsymbol{\psi} \right) \boldsymbol{T}_d \hatmat_t
\end{bmatrix},
\end{split}
\label{p34c}
\end{align}
where $\boldsymbol{\mathit{\Omega}}_1$, $\boldsymbol{\mathit{\Omega}}_2$, and $\boldsymbol{\mathit{\Omega}}_4$ are the matrices as given in \autoref{lemma3-c4-1}.\par

\noindent Since $\boldsymbol{F}_d = \left(\identity_n \otimes \boldsymbol{\psi} \right) \boldsymbol{T}_d$, from the expression of the matrices $\boldsymbol{C}_{{d(s2)}(11)(1)}$, $\boldsymbol{C}_{{d(s2)}(12)(1)}$, \\$\boldsymbol{C}_{{d(s2)}(21)(1)}$ and $\boldsymbol{C}_{{d(s2)}(22)(1)}$ as given in \autoref{lemma5-c4}, and using the expression of $\boldsymbol{A}^*$ as given in \autoref{lemma3-c4-1}, we get
\begin{align*}
\begin{split}
\boldsymbol{C}_{{d(s2)}(11)(1)} &=
\begin{bmatrix}
\boldsymbol{T}_d^{'} \left( \hatmat_n \otimes \boldsymbol{\mathit{\Omega}}_1 \right) \boldsymbol{T}_d & -\boldsymbol{T}_d^{'} \left( \hatmat_n \otimes \boldsymbol{\mathit{\Omega}}_2 \right) \boldsymbol{T}_d\\
-\boldsymbol{T}_d^{'} \left( \hatmat_n \otimes \boldsymbol{\mathit{\Omega}}_2 \right) \boldsymbol{T}_d & \boldsymbol{T}_d^{'} \left( \hatmat_n \otimes \boldsymbol{\mathit{\Omega}}_4 \right) \boldsymbol{T}_d
\end{bmatrix},\\
\boldsymbol{C}_{{d(s2)}(12)(1)} &= \boldsymbol{C}^{'}_{{d(s2)}(21)(1)} =
\begin{bmatrix}
\boldsymbol{T}_d^{'} \left( \hatmat_n \otimes \boldsymbol{\mathit{\Omega}}_1 \boldsymbol{\psi} \right) \boldsymbol{T}_d \hatmat_t & -\boldsymbol{T}_d^{'} \left( \hatmat_n \otimes \boldsymbol{\mathit{\Omega}}_2 \boldsymbol{\psi} \right) \boldsymbol{T}_d \hatmat_t\\
-\boldsymbol{T}_d^{'} \left( \hatmat_n \otimes \boldsymbol{\mathit{\Omega}}_2 \boldsymbol{\psi} \right) \boldsymbol{T}_d \hatmat_t & \boldsymbol{T}_d^{'} \left( \hatmat_n \otimes \boldsymbol{\mathit{\Omega}}_4 \boldsymbol{\psi} \right) \boldsymbol{T}_d \hatmat_t
\end{bmatrix}, \text{ and}\\
\boldsymbol{C}_{{d(s2)}(22)(1)} &=
\begin{bmatrix}
\hatmat_t \boldsymbol{T}_d^{'} \left( \hatmat_n \otimes \boldsymbol{\psi}^{'} \boldsymbol{\mathit{\Omega}}_1 \boldsymbol{\psi} \right) \boldsymbol{T}_d \hatmat_t & -\hatmat_t \boldsymbol{T}_d^{'} \left( \hatmat_n \otimes \boldsymbol{\psi}^{'} \boldsymbol{\mathit{\Omega}}_2 \boldsymbol{\psi} \right) \boldsymbol{T}_d \hatmat_t\\
-\hatmat_t \boldsymbol{T}_d^{'} \left( \hatmat_n \otimes \boldsymbol{\psi}^{'} \boldsymbol{\mathit{\Omega}}_2 \boldsymbol{\psi} \right) \boldsymbol{T}_d \hatmat_t & \hatmat_t \boldsymbol{T}_d^{'} \left( \hatmat_n \otimes \boldsymbol{\psi}^{'} \boldsymbol{\mathit{\Omega}}_4 \boldsymbol{\psi} \right) \boldsymbol{T}_d \hatmat_t
\end{bmatrix}.
\end{split}
\end{align*}
Thus using \autoref{remark1-c4} and following the proof of Lemma 4 in \citet{Martin1998Variance-balancedObservations}, since $\hatmat_t \vecone_t = \zero_{t \times 1}$,  we get that if $d$ is uniform on periods,
\begin{align}
\begin{split}
\boldsymbol{C}_{{d(s2)}(11)(1)} &=
\begin{bmatrix}
\boldsymbol{T}_d^{'} \left( \identity_n \otimes \boldsymbol{\mathit{\Omega}}_1 \right) \boldsymbol{T}_d & -\boldsymbol{T}_d^{'} \left( \identity_n \otimes \boldsymbol{\mathit{\Omega}}_2 \right) \boldsymbol{T}_d\\
-\boldsymbol{T}_d^{'} \left( \identity_n \otimes \boldsymbol{\mathit{\Omega}}_2 \right) \boldsymbol{T}_d & \boldsymbol{T}_d^{'} \left( \identity_n \otimes \boldsymbol{\mathit{\Omega}}_4 \right) \boldsymbol{T}_d
\end{bmatrix},\\
\boldsymbol{C}_{{d(s2)}(12)(1)} &= \boldsymbol{C}^{'}_{{d(s2)}(21)(1)} =
\begin{bmatrix}
\boldsymbol{T}_d^{'} \left( \identity_n \otimes \boldsymbol{\mathit{\Omega}}_1 \boldsymbol{\psi} \right) \boldsymbol{T}_d \hatmat_t & -\boldsymbol{T}_d^{'} \left( \identity_n \otimes \boldsymbol{\mathit{\Omega}}_2 \boldsymbol{\psi} \right) \boldsymbol{T}_d \hatmat_t\\
-\boldsymbol{T}_d^{'} \left( \identity_n \otimes \boldsymbol{\mathit{\Omega}}_2 \boldsymbol{\psi} \right) \boldsymbol{T}_d \hatmat_t & \boldsymbol{T}_d^{'} \left( \identity_n \otimes \boldsymbol{\mathit{\Omega}}_4 \boldsymbol{\psi} \right) \boldsymbol{T}_d \hatmat_t
\end{bmatrix}, \text{ and}\\
\boldsymbol{C}_{{d(s2)}(22)(1)} &=
\begin{bmatrix}
\hatmat_t \boldsymbol{T}_d^{'} \left( \identity_n \otimes \boldsymbol{\psi}^{'} \boldsymbol{\mathit{\Omega}}_1 \boldsymbol{\psi} \right) \boldsymbol{T}_d \hatmat_t & -\hatmat_t \boldsymbol{T}_d^{'} \left( \identity_n \otimes \boldsymbol{\psi}^{'} \boldsymbol{\mathit{\Omega}}_2 \boldsymbol{\psi} \right) \boldsymbol{T}_d \hatmat_t\\
-\hatmat_t \boldsymbol{T}_d^{'} \left( \identity_n \otimes \boldsymbol{\psi}^{'} \boldsymbol{\mathit{\Omega}}_2 \boldsymbol{\psi} \right) \boldsymbol{T}_d \hatmat_t & \hatmat_t \boldsymbol{T}_d^{'} \left( \identity_n \otimes \boldsymbol{\psi}^{'} \boldsymbol{\mathit{\Omega}}_4 \boldsymbol{\psi} \right) \boldsymbol{T}_d \hatmat_t
\end{bmatrix}.
\end{split}
\label{p32c}
\end{align}
So, if $d \in \Omega_{t,n,p}$ is a design uniform on periods, then from \eqref{p34c} and \eqref{p32c}, we get that for $f_1,f_2=1,2$, $
\boldsymbol{{C}}_{d(s2)(f_1f_2)(1)} = \boldsymbol{\tilde{C}}_{d(s2)(f_1f_2)}
$. Thus using the expression of the matrices $\boldsymbol{{C}}_{d(s2)}$ and $\boldsymbol{\tilde{C}}_{d(s2)}$ as given in \autoref{lemma5-c4} and \autoref{lemma5-c4-1}, respectively, we get that $\boldsymbol{{C}}_{d(s2)}=\boldsymbol{\tilde{C}}_{d(s2)}$.
\end{proof}

\begin{lemmas}
\label{lemma-markov-1}
Let $\boldsymbol{V}_1^* = \boldsymbol{V}_1^{-1} - \left( \vecone_{p}^{'} \boldsymbol{V}_1^{-1} \vecone_p \right)^{-1} \boldsymbol{V}_1^{-1} \matone_{p \times p} \boldsymbol{V}_1^{-1}$ and $\boldsymbol{V}_R^* = \boldsymbol{V}_R^{-1} - \left( \vecone_{p}^{'} \boldsymbol{V}_R^{-1} \vecone_p \right)^{-1} \times$\\$ \boldsymbol{V}_R^{-1} \matone_{p \times p} \boldsymbol{V}_R^{-1}$. Then $tr \left( \hatmat_{p} \boldsymbol{\psi}^{'} \boldsymbol{V}_1^* \boldsymbol{\psi} \right) > 0$ and $tr \left( \hatmat_{p} \boldsymbol{\psi}^{'} \boldsymbol{V}_R^* \boldsymbol{\psi} \right) > 0$.
\end{lemmas}
\begin{proof}
 Here, $\boldsymbol{V}_1$ is a symmetric and positive definite matrix. Since $pr^{\perp} \left( . \right)$ is a symmetric and idempotent matrix, $\boldsymbol{V}_1^* = \boldsymbol{V}_1^{-1} - \left( \vecone_{p}^{'} \boldsymbol{V}_1^{-1} \vecone_p \right)^{-1} \boldsymbol{V}_1^{-1} \matone_{p \times p} \boldsymbol{V}_1^{-1}$ can be written as
 \begin{align}
\boldsymbol{V}_1^* &= \boldsymbol{L}_{1}^{'} \boldsymbol{L}_{1},
\label{markov-t11c-1}
 \end{align}
 where $\boldsymbol{L}_{1} =   pr^{\perp} \left( \boldsymbol{V}_1^{-1/2} \matone_{p \times p} \right) \boldsymbol{V}_1^{-1/2}$. Since $\hatmat_p$ is a symmetric and idempotent matrix, \\$tr \left( \hatmat_p \boldsymbol{\psi}^{'} \boldsymbol{V}_1^* \boldsymbol{\psi}  \right) = tr \left(  \boldsymbol{L}_{2}^{'} \boldsymbol{L}_{2} \right) \geq 0$, where $\boldsymbol{L}_2 = \boldsymbol{L}_1 \boldsymbol{\psi} \hatmat_p $. It can be easily seen that $tr \left(  \boldsymbol{L}_{2}^{'} \boldsymbol{L}_{2} \right) = \zero_{p \times p}$ if and only if $\boldsymbol{L}_{2} = \zero_{p \times p}$. Now, suppose $tr \left( \hatmat_p \boldsymbol{\psi}^{'} \boldsymbol{V}_1^* \boldsymbol{\psi}  \right) = 0$. So, $\boldsymbol{L}_1 \boldsymbol{\psi} \hatmat_p = \zero_{p \times p}$ and thus from \eqref{markov-t11c-1}, we get
 \begin{align}
     \boldsymbol{V}_1^* \boldsymbol{\psi} \hatmat_{p} = \zero_{p \times p}.
     \label{markov-t12c-1}
 \end{align}
 From the expression of $\boldsymbol{V}_1^*$, we get that the row sums of $\boldsymbol{V}_1^*$ is $0$. Thus using \eqref{markov-t12c-1}, by calculation we get
 \begin{align}
 \begin{split}
\left( \boldsymbol{V}_1^* \right)_{1,2} &= \cdots = \left( \boldsymbol{V}_1^* \right)_{1,p} = -\frac{1}{p} \left( \boldsymbol{V}_1^* \right)_{1,1},\\
\vdots\\
\left( \boldsymbol{V}_1^* \right)_{p,2} &= \cdots = \left( \boldsymbol{V}_1^* \right)_{p,p} = -\frac{1}{p} \left( \boldsymbol{V}_1^* \right)_{p,1},
\end{split}
\label{markov-t13c-1}
 \end{align}
 where $\left( \boldsymbol{V}_1^* \right)_{l_1,l_2}$ denotes the $(l_1,l_2)^{th}$ element of $\boldsymbol{V}_1^*$. Using the property that the row sums of $\boldsymbol{V}_1^*$ is $0$ in \eqref{markov-t13c-1}, we get that $\boldsymbol{V}_1^{*} = \zero_{p \times p}$. Thus, by calculation we get
 \begin{align}
\left( \vecone_{p}^{'} \boldsymbol{V}_1^{-1} \vecone_p \right)^{-1} \matone_{p \times p} \boldsymbol{V}_1^{-1} &= \identity_p.
\label{markov-t14c-1}
 \end{align}
 The above equation \eqref{markov-t14c-1} leads to a contradiction. Hence our assumption that \\$tr \left( \hatmat_p \boldsymbol{\psi}^{'} \boldsymbol{V}_1^* \boldsymbol{\psi} \right) =0$ is wrong. Thus, $tr \left( \hatmat_p \boldsymbol{\psi}^{'} \boldsymbol{V}_1^* \boldsymbol{\psi} \right) >0$.\par
 
\noindent Since $\boldsymbol{V}_R$ is a positive definite matrix, similarly we can prove that $tr \left( \hatmat_p \boldsymbol{\psi}^{'} \boldsymbol{V}_R^* \boldsymbol{\psi} \right) >0$.
\end{proof}

\end{appendices}

\end{document}